\documentclass[journal,onecolumn,12pt]{IEEEtran}
\usepackage{times,amsmath,amsthm,epsfig,amssymb,graphicx,amsfonts}
\usepackage{hyphenat}
\usepackage{psfrag}
\usepackage{ifpdf}
\usepackage{cite}
\usepackage{array}
\usepackage{multirow}
\usepackage{graphicx}
\usepackage{color}
\usepackage{subfigure}
\usepackage{epsfig}
\usepackage{citesort}
\usepackage{setspace}
\usepackage{latexsym}
\usepackage{amssymb}
\usepackage{amsmath}

\title{On the Symmetric $K$-user Interference Channels with Limited Feedback}
\author{Mehdi Ashraphijuo, Vaneet Aggarwal and Xiaodong Wang \thanks{ This paper was presented in part at the 2014 Annual Allerton Conference on Communication, Control, and Computing.

M. Ashraphijuo and X. Wang are with the Electrical Engineering Department, Columbia University, New York, NY 10027, email: \{mehdi,wangx\}@ee.columbia.edu. V. Aggarwal is with the School of Industrial Engineering, Purdue University, West Lafayette, IN 47907, email: vaneet@purdue.edu .}}
\newtheorem{theorem}{Theorem}

\newtheorem{lemma}{Lemma}
\newtheorem{corollary}{Corollary}

\newtheorem{remark}{Remark}

\begin{document}
\maketitle
\begin{abstract}
In this paper, we develop achievability schemes for symmetric $K$-user interference channels with a rate-limited feedback from each receiver to the corresponding transmitter. We study this problem under two different channel models: the linear deterministic model, and the Gaussian model. For the deterministic model, the proposed scheme achieves a symmetric rate that is the minimum of the symmetric capacity with infinite feedback, and the sum of the symmetric capacity without feedback and the symmetric amount of feedback. For the  Gaussian interference channel, we use lattice codes to propose a transmission strategy that incorporates the techniques of Han-Kobayashi message splitting, interference decoding, and decode and forward. This strategy achieves a symmetric rate which is within a constant number of bits to the minimum of the symmetric capacity with infinite feedback, and the sum of the symmetric capacity without feedback and the  amount of symmetric feedback. This constant is obtained as a function of the number of users, $K$. The symmetric achievable rate is used to characterize the achievable generalized degrees of freedom which exhibits a gradual increase from no feedback to perfect feedback in the presence of feedback links with limited capacity.
\end{abstract}

\


\newpage
\section{Introduction}

The interference channel (IC) has been studied in the literature since 1970's to understand performance limits of multiuser communication networks \cite{Carleial}. Although the exact characterization of the capacity region of a two-user Gaussian IC is still unknown, several inner and outer bounds have been obtained. These bounds have resulted in an approximate characterization of the capacity region, within one bit, in \cite{Etkin} and \cite{Telatar}. Such characterization includes outer bounds on the capacity region for the two-user Gaussian IC, as well as encoding/decoding strategies based on the Han-Kobayashi scheme \cite{HK}, which performs close to the optimum. On the other hand, the $K$-user IC has been studied in \cite{Jafar3,Ordentlich} for a symmetric scenario, where all direct links (from each transmitter to its respective receiver) have the same gain, and similarly, the gains of all cross (interfering) links are identical. For such a $K$-user symmetric IC, the number of symmetric generalized degrees of freedom (${\mathsf {GDoF}}$) is characterized in \cite{Jafar3}, and an approximate sum capacity is given in \cite{Ordentlich}.

It is well known that feedback does not increase the capacity of point-to-point discrete memoryless channels \cite{Shannon}. However, feedback is beneficial in improving the capacity region of multi-user networks (see \cite{ElGamal} and references therein). A number of works on ICs explore feedback strategies, where each receiver feeds back the channel output to its own transmitter\cite{Cadambe2,Tsesuh,J1,Tuninetti1,Tuninetti1,Achal,Achal1st,Kramer,feedback1,feedback2}. Several coding schemes for the $K$-user Gaussian IC are developed in \cite{Kramer}. The effect of feedback on the capacity region of the two-user IC is studied in \cite{Cadambe2}, where it is shown that feedback provides a multiplicative gain in the sum capacity at high signal-to-noise ratio (${\mathsf {SNR}}$), when the interference links are much stronger than the direct links. The capacity region of the two-user Gaussian IC with unlimited feedback is characterized within a 2 bit gap in \cite{Tsesuh}. The $K$-user symmetric IC with unlimited feedback is considered in \cite{Mohajer}, where the ${\mathsf {GDoF}}$ is characterized. A more realistic feedback model is one where the feedback links are rate-limited. The impact of rate-limited feedback is studied for a two-user Gaussian IC in \cite{Vahid}, where it is shown that the maximum gain in the symmetric capacity with feedback is the amount of symmetric feedback.

In this paper, we study the impact of rate-limited feedback for a $K$-user IC. We first consider this problem for the linear deterministic model proposed in \cite{Avestimehr} as an approximation to the Gaussian model, and then treat the Gaussian model. For the Gaussian model, we develop an achievability scheme that employs the techniques of Han-Kobayashi message splitting, interference decoding and decode-and-forward. In order to effectively decode the interference, lattice codes are used such that the sum of signals can be decoded without decoding the individual signals. We also find the achievable symmetric ${\mathsf {GDoF}}$ with rate-limited feedback.

Roughly speaking, except for the pairs of (${\mathsf {SNR}},{\mathsf {INR}}$) where $\lim_{{\mathsf{SNR}}\to\infty}\frac{\log{\mathsf{INR}}}{\log{\mathsf{SNR}}} =1$, the effect of interference from the other $K-1$ users is as if there were only one interferer in the network. This is analogous to the result of \cite{Jafar3} and \cite{Mohajer}, where it is shown that for the cases of no feedback and unlimited feedback, respectively,  the symmetric ${\mathsf {GDoF}}$  of the $K$-user IC is the same as that of a two-user IC.

In order to get the maximal benefit of feedback, we use an encoding scheme which combines two well-known interference management techniques, namely, interference alignment and interference decoding. More precisely, the encoding at the transmitters is such that all the interfering signals are aligned at each receiver. However, a fundamental difference between our approach and the conventional interference alignment approach is that we need to decode interference to be able to remove it from the received signal, whereas the aligned interference is usually suppressed in conventional approaches. A challenge here, which makes the $K$-user  problem fundamentally different from the two-user problem \cite{Vahid}, is that the interference is a combination of multiple interfering messages instead of a single message as in the two-user case, and decoding all of them imposes strict bounds on the rate of the interfering messages. A key idea is that instead of decoding all the interfering messages individually, we will decode some combination of them that corrupts the intended message of interest. In the proposed scheme, the receiver decodes the sum of certain interfering signals and the intended signal, and sends it back to the transmitter. The transmitter, knowing the intended signal, can then decode the sum of interfering signals and transmits to the receiver in the next slot, to help  the receiver to decode the intended signal. In order to decode the sum of certain intended/interfering signals, all transmitters employ a common structured  lattice code \cite{Lattice} which has the property that the sum of different codewords is another codeword from the same codebook.

Our new scheme generalizes the prior works in  \cite{Vahid,Tsesuh,Mohajer,Etkin,Jafar,Ordentlich} as follows. In this paper, we investigate the cases of weak and strong regimes of interference channels. It is because as it will be explained later in the paper, a regime of medium interference level has been previously shown to not to have a improvement via feedback. A two-user IC with rate-limited feedback is considered in \cite{Vahid}, while this paper develops the achievability for a symmetric $K$-user IC. A two-user IC without feedback is treated in \cite{Etkin}, which is a special case of the $K$-user IC without feedback in \cite{Jafar,Ordentlich}. A two-user IC with unlimited feedback is considered in  \cite{Tsesuh}, which is a special case of \cite{Mohajer} where the $K$-user IC with unlimited feedback is treated. In this paper, we develop an achievability scheme for a $K$-user IC with limited feedback, which for the special cases of two-user, no feedback, and unlimited feedback, results in schemes that are different from those in \cite{Vahid}, \cite{Ordentlich}, and \cite{Mohajer}, respectively. This achievability scheme achieves a rate which is approximately equal to the symmetric capacity without feedback plus the symmetric amount of feedback up to some saturation point, which corresponds to reaching the symmetric capacity with infinite feedback which leads us to conjecture an outer bound. The challenge in proving that the conjectured upper bound is indeed an upper bound lies in the fact that when feedback links are available, it is not immediate that adding more interferers can only decrease capacity. Thus, the two-user bounds based on genie-aided information on all other users' messages become hard to extend since the other receivers can feed back certain signals, and thus other transmitters can potentially help increase rate. We further note that for the two-user case, the achievable symmetric rate in \cite{Vahid} is not within a constant gap to the upper bound for a certain interference region, and the result in this paper fixes the results of \cite{Vahid} thus providing the correct approximate capacity result for the case of $K=2$.

For the achievability scheme for two-user IC in \cite{Vahid}, the two transmitters have different and asymmetric encoding operations, and it cannot be generalized to arbitrary number of users. Also alignment of interfering signals and encoding them by a lattice code is not considered, because each receiver receives interference from only one transmitter. In our proposed achievability scheme, all the transmitters employ the same encoding operation and therefore all users are symmetric. Moreover, each receiver receives interference from the other $K-1$ transmitters and we align  and encode them using a lattice code so that the sum signal can be decoded.  On the other hand, in the achievability scheme for the $K$-user IC with infinite feedback in \cite{Mohajer}, each receiver simply sends back all received signals to the corresponding transmitter whereas in our scheme each receiver sends back a lattice codeword (via the rate-limited feedback channel) with a strategy that is chosen depending on the interference regime. Finally the achievability scheme for the $K$-user IC with no feedback in \cite{Ordentlich}  only performs alignment on the interfering signals and does not deal with feedback. A novelty in this paper is to decide which part of the signal and interference should be aligned to be decoded as a lattice codeword, to be fed back to the transmitter with limited feedback. The proposed scheme in this paper uses the concept of signal alignment with lattice codes in addition to rate-limited feedback that has not been jointly considered in \cite{Ordentlich,Mohajer,Vahid}.

The remainder of this paper is organized as follows. Section II gives the symmetric achievable rate for the deterministic model, with some examples to illustrate the main ideas of the proposed achievability scheme. Section III gives our results for the Gaussian model, where the proposed achievability scheme is described and the achievable symmetric rate, a conjectured upper bound, and the achievable ${\mathsf {GDoF}}$ are given. Finally, Section IV concludes the paper. Some of the proofs are given in appendices.

\section{Deterministic Model}\label{sec:det}

\subsection{System Model and Problem Formulation}
We first consider the linear deterministic $K$-user IC. This model was proposed in \cite{Avestimehr} to focus on signal interactions instead of the additive noise, and to obtain insights for the Gaussian model.  Let $s\ge K$ be a prime number and ${\mathsf {F}_s}$ be the finite field over the set $\{0, \dots, s-1\}$ with sum and product modulo $s$. Moreover, in this model there is a non-negative integer $n_{kj}$ representing the channel gain from transmitter $k$ to receiver $j$, $j,k \in\{1, \cdots, K\}$. We assume that $n_{jk}=\lfloor\log_s^{\mathsf{SNR}}\rfloor=n$ for $j=k$ and $n_{jk}=\lfloor\log_s^{\mathsf{INR}}\rfloor=m$ for $j\ne k$. Also, define $q \triangleq \max(m,n)$. We write the channel input at transmitter $k$ at time $i$ as $X_{k,i} = [X_{k,i}^1, X_{k,i}^2, \ldots, X_{k,i}^q] \in {\mathsf{F}_s}^q$, for $k \in \{1, 2, \cdots, K\}$, such that $X_{k,i}^1$ and $X_{k,i}^q$ represent the most and the least significant levels
of the transmitted signal, respectively. At each time $i$, the received signal at the $k^{\text{th}}$ receiver is given by
\begin{equation}
Y_{k,i}=D^{q-n}X_{k,i}+\sum_{j\neq k}{D^{q-m}X_{j,i}},
\end{equation}
where all the operations are performed modulo $s$. \footnote{For any prime number $s$, it holds that the equation $(s-1)x=a \mod s$ has a unique solution in $\{0,1,...,s-1\}$ and this property will be used frequently in this paper.} Also, $D$ is a $q\times q$ shift matrix. We assume that there is a feedback channel from the $k^{\text{th}}$ receiver to the $k^{\text{th}}$ transmitter which is of capacity $p$, and that $p$ is a multiple of $\log s$ because of using the finite field of size $s$. The feedback is causal and hence at time $i$ the signal received till time $i-1$ is available at each receiver for encoding and feeding back to the corresponding transmitter. Fig. \ref{fig:Example-det} depicts a linear deterministic IC for $n=5$ and $m=2$.
All the transmissions $a_{i,j}$, $1\le i\le K$, $1\le j\le\max\{n,m\}$ are $s$-ary.
\begin{figure}[htbp]
\centering
	\includegraphics[width=6cm]{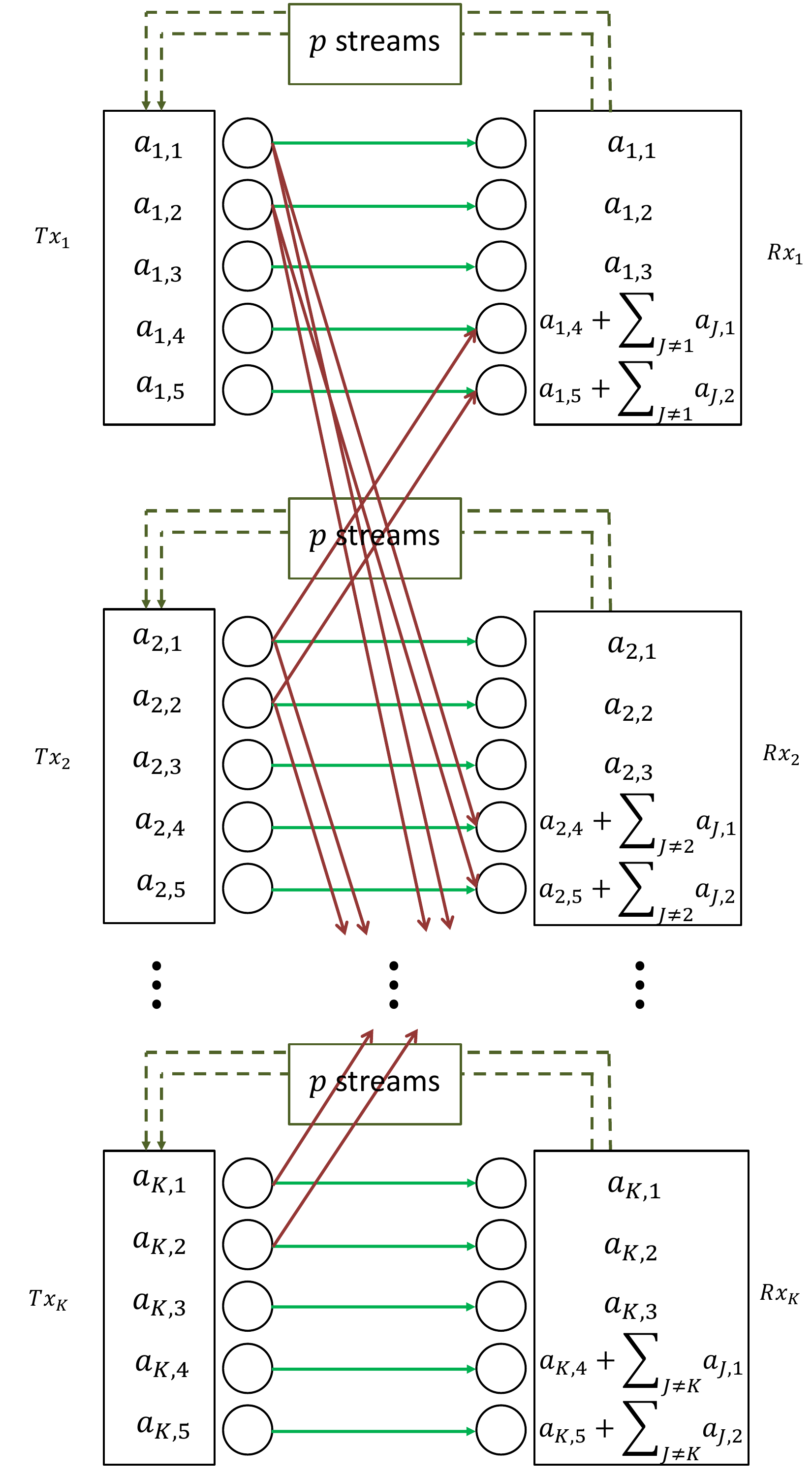}
\caption{A linear deterministic IC with $n=5$, $m=2$ and $p=2$.}
\label{fig:Example-det}
\end{figure}

For a deterministic IC, a symmetric rate $R_{sym}$ is said to be achievable if there is a strategy such that all users can get a rate $R_{sym}$. We further define $\alpha \triangleq m/n$ and $\beta \triangleq p/n$.

For the deterministic channel, in defining the achievable rate and decoding process, the zero-error probability model is assumed. Similarly, the notion of zero-error capacity is used for the converse proofs as in \cite{Avestimehr}.

\subsection{Results for Linear Deterministic IC Model}
In this section, we describe our proposed coding schemes for the $K$-user linear deterministic IC with rate-limited feedback. The following theorem gives our achievability result.

\begin{theorem}\label{thm_det}
For the $K$-user linear deterministic IC, the following symmetric rate is achievable:
\begin{eqnarray}\label{r-d}
{R}_{sym} /{\log s}=\left\{ \begin{array}{ll}
\min\{n-m+p,n-\frac{m}{2}\},& \text{ if } \ 0\le m \le \frac{n}{2}, \\
\min\{m+p,n-\frac{m}{2}\},& \text{ if } \ \frac{n}{2}\le m \le \frac{2n}{3}, \\
n-\frac{m}{2},& \text{ if }\frac{2n}{3}\le m < n, \\
\frac{n}{K},& \text{ if } \ \ \ \ \ \ \ m=n, \\
\frac{m}{2},& \text{ if } \ n < m \le 2n, \\
\min\{n+p,\frac{m}{2}\},& \text{ if } 2n \le m.
\end{array}
\right.
\end{eqnarray}
\end{theorem}

\begin{remark}\label{infde}
With infinite feedback, i.e., $p=\infty$, according to Theorem 4 of \cite{Mohajer}, the symmetric capacity is
\begin{eqnarray}\label{r-dp}
{C}_{sym,\infty} /{\log s} =\left\{ \begin{array}{ll}
n-\frac{m}{2},& \text{ if }0\le m < n, \\
\frac{n}{K},& \text{ if } \ \ \ \ \ m=n, \\
\frac{m}{2},& \text{ if } n < m.
\end{array}
\right.
\end{eqnarray}
\end{remark}

\begin{corollary}\label{zde}
With no feedback, i.e., $p=0$, the symmetric capacity is
\begin{eqnarray}\label{r-dn}
{C}_{sym,0} /{\log s}=\left\{ \begin{array}{ll}
n-m,& \text{ if } \ 0\le m \le \frac{n}{2}, \\
m,& \text{ if } \ \frac{n}{2}\le m \le \frac{2n}{3}, \\
n-\frac{m}{2},& \text{ if }\frac{2n}{3}\le m < n, \\
\frac{n}{K},& \text{ if } \ \ \ \ \ \ \ m=n, \\
\frac{m}{2},& \text{ if } \ n < m \le 2n, \\
n,& \text{ if } 2n \le m.
\end{array}
\right.
\end{eqnarray}
\end{corollary}
\begin{proof}
The achievability follows from Theorem \ref{thm_det} for $p=0$. The upper bound for $\frac{2n}{3}\le m \le 2n$ follows from Remark \ref{infde} and for $2n \le m$ it is a simple cutset bound. For $0 \le m\le \frac{2n}{3}$,  the proof is given as follows. Assume that the $k^{\text{th}}$ transmitter transmits $X_k={[a_{k,1},\ldots,a_{k,n}]}^T$ and let $S_k\triangleq{[a_{k,1},\ldots,a_{k,m}]}^T$, for $k\in\{1,\ldots,K\}$. The sum rate for any two users (say 1 and 2) can be found by giving genie-aided information on all other users' messages, the two-user bound still holds for any number of users. Thus, 
\begin{eqnarray}\label{r-dn}
R_1+R_2&\le&I(X_1;Y_1,S_1)+I(X_2;Y_2,S_2)\nonumber\\
&=&I(X_1;S_1)+I(X_1;Y_1|S_1)+I(X_2;S_2)+I(X_2;Y_2|S_2)\nonumber\\
&=&h(S_1)-h(S_1|X_1)+h(Y_1|S_1)-h(Y_1|S_1,X_1)+\nonumber\\
&&h(S_2)-h(S_2|X_2)+h(Y_2|S_2)-h(Y_2|S_2,X_2)\nonumber\\
&=&h(S_1)-h(S_1|X_1)+h(Y_1|S_1)-h(S_2)+\nonumber\\
&&h(S_2)-h(S_2|X_2)+h(Y_2|S_2)-h(S_1)\nonumber\\
&=&h(Y_1|S_1)+h(Y_2|S_2),
\end{eqnarray}
where for $0 \le m\le \frac{n}{2}$ we have $h(Y_k|S_k)\le n-m$, and also for $\frac{n}{2} \le m\le \frac{2n}{3}$ we have $h(Y_k|S_k)\le m$.

Since this holds for any two users, we can combine these upper bounds to get the result as in the statement.\end{proof}

Then, from \eqref{r-d}-\eqref{r-dn} we have ${R}_{sym} /{\log s}=\min\{{C}_{sym,\infty},{C}_{sym,0}+p\}$. For $K = 2$, this result has been shown to be tight in \cite{Vahid}, i.e., ${R}_{sym}$ is the symmetric capacity for $K=2$ , and we conjecture that ${R}_{sym}$ is the symmetric capacity for a general $K$.

\begin{figure}[htbp]
\centering
	\includegraphics[width=10cm]{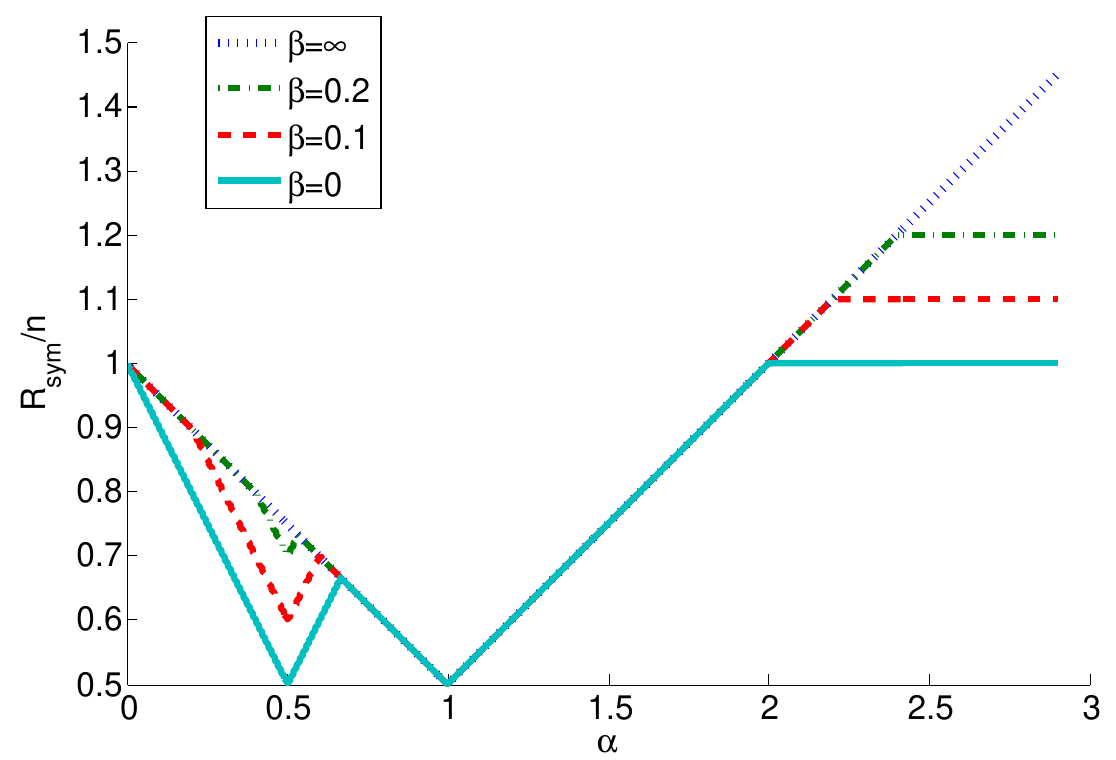}
\caption{Achievable symmetric rate of the deterministic IC with feedback.}
\label{fig:Example0d}
\end{figure}

Fig. \ref{fig:Example0d} illustrates the (normalized) symmetric rate as a function of $\alpha$, for different values of $\beta=0$ (i.e., Corollary \ref{zde}), $\beta=0.1$, $\beta=0.2$, and $\beta=\infty$ (i.e., Remark \ref{infde}).

The complete proof of Theorem \ref{thm_det} is given in Appendix \ref{apdx_innerdet}. In this section, we present several examples of the transmission schemes that  achieve the symmetric rate as claimed in Theorem \ref{thm_det}. For the range of $\frac{2}{3}\le \alpha \le 2$, we have ${C}_{sym,0}={C}_{sym,\infty}$. So, with limited feedback the symmetric capacity remains the same in this range. In the rest of this section, we will illustrate the proposed coding schemes for three ranges of $\alpha$. Generalization of the proposed coding strategy with specific channel parameter values for arbitrary $n$, $m$, and $p$ and its analysis is presented in Appendix \ref{apdx_innerdet}.

\subsubsection{Very Weak Interference Regime $(\alpha \le \frac{1}{2})$}

In the very weak interference regime, the goal is to achieve a symmetric rate of $R_{sym} /{\log s}=\min\{n-m+p,n-\frac{m}{2}\}$ bits per user. We propose an encoding scheme that operates on a block of length 2. The basic idea can be seen from Fig. \ref{fig:Example1}, where the coding scheme is demonstrated for $K=3$, $n=5$, $m=2$, and $p=0.5$.

\begin{figure}[htbp]
\centering
	\includegraphics[width=8cm]{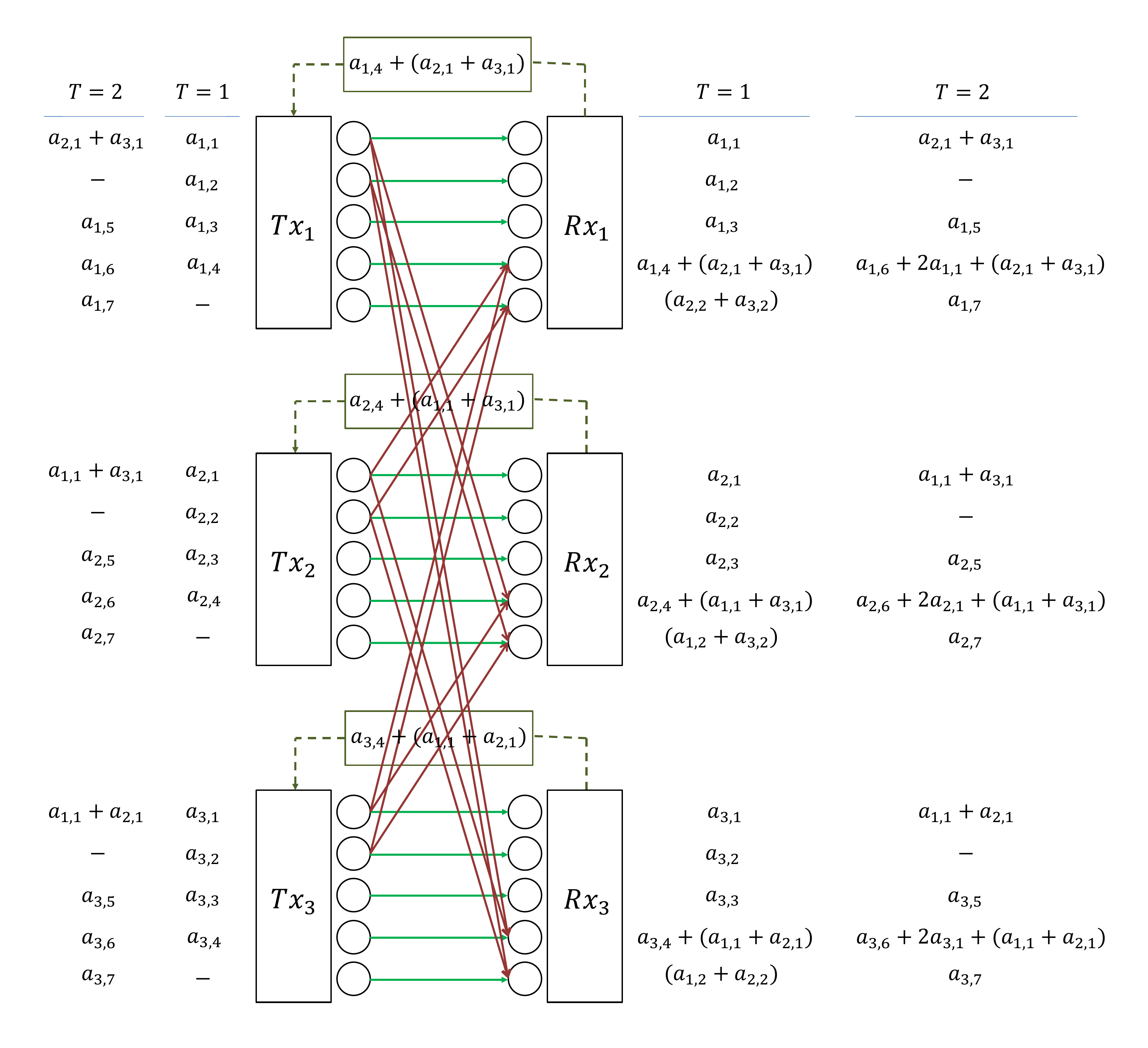}
\caption{Proposed coding scheme for the linear deterministic model in the very weak interference regime ($\alpha\le\frac{1}{2}$), for $K=3$, $n=5$, $m=2$ and $p=0.5$.}
\label{fig:Example1}
\end{figure}

As shown in Fig. \ref{fig:Example1}, the proposed coding scheme is able to convey seven intended symbols from each transmitter to its respective receiver in two channel uses, i.e., $2R_{sym}=7$. 


\subsubsection{Weak Interference Regime $(\frac{1}{2} \le \alpha \le \frac{2}{3})$}

In the weak interference regime, the goal is to achieve a symmetric rate of $R_{sym} /{\log s}=\min\{m+p,n-\frac{m}{2}\}$ bits per user. We propose an encoding scheme that operates on a block of length 2. The basic idea can be seen from Fig. \ref{fig:Example2}, where the coding scheme is demonstrated for $K=3$, $n=7$, $m=4$, and $p=0.5$.

\begin{figure}[htbp]
\centering
	\includegraphics[width=10cm]{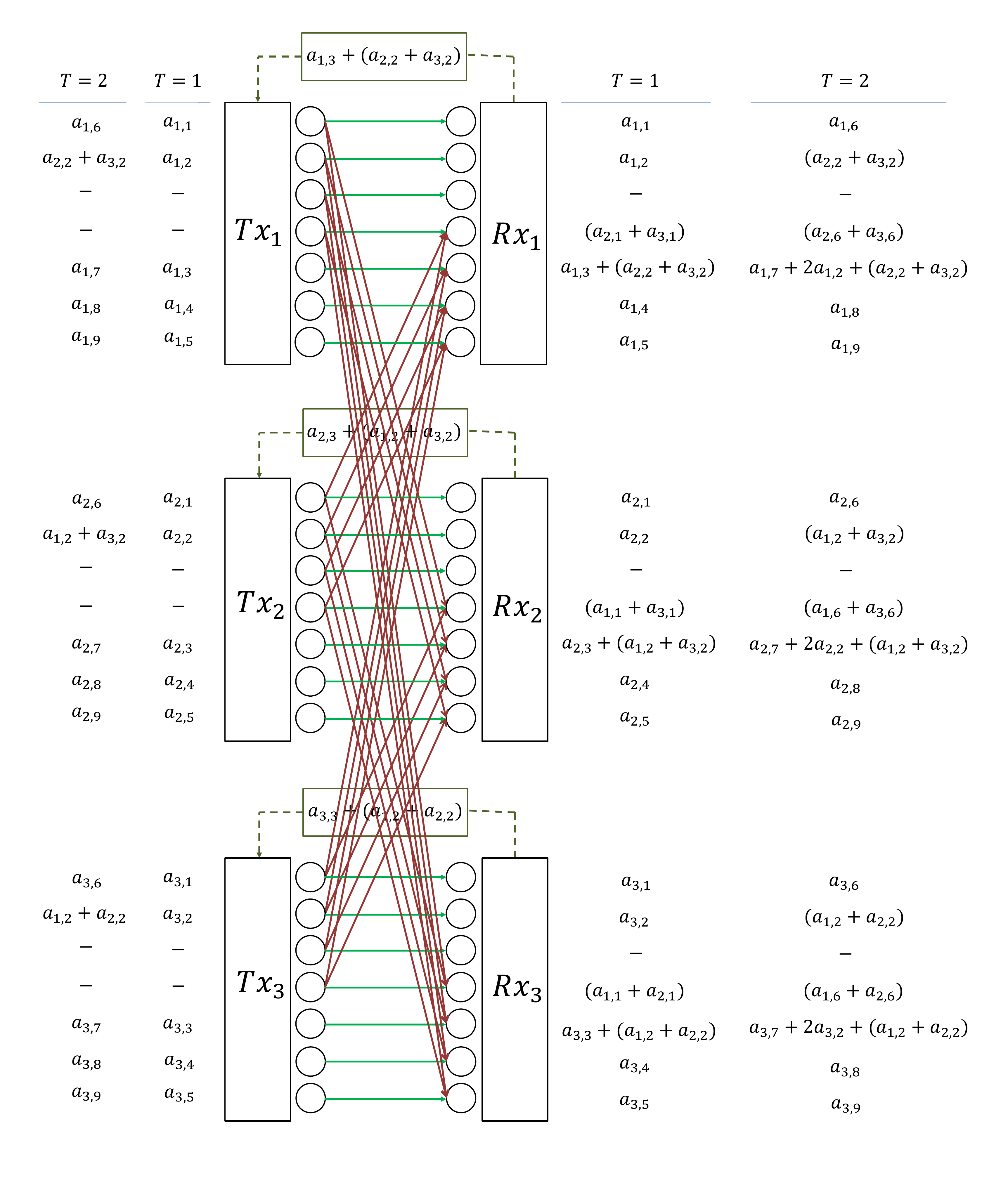}
\caption{Proposed coding scheme for the linear deterministic model in the weak interference regime ($\frac{1}{2}\le\alpha\le\frac{2}{3}$), for $K=3$, $n=7$, $m=4$ and $p=0.5$.}
\label{fig:Example2}
\end{figure}

As shown in Fig. \ref{fig:Example2}, the proposed coding scheme is able to convey nine intended symbols from each transmitter to its respective receiver in two channel uses, i.e., $2R_{sym}=9$. 


\subsubsection{Very Strong Interference Regime $(\alpha\ge 2)$}

In the very strong interference regime, the goal is to achieve a symmetric rate of $R_{sym} /{\log s}=\min\{n+p,\frac{m}{2}\}$ bits per user. We propose an encoding scheme that operates on a block of length 2. The basic idea can be seen from Fig. \ref{fig:Example3}, where the coding scheme is demonstrated for $K=3$, $n=2$, $m=6$, and $p=0.5$.

\begin{figure}[htbp]
\centering
	\includegraphics[width=10cm]{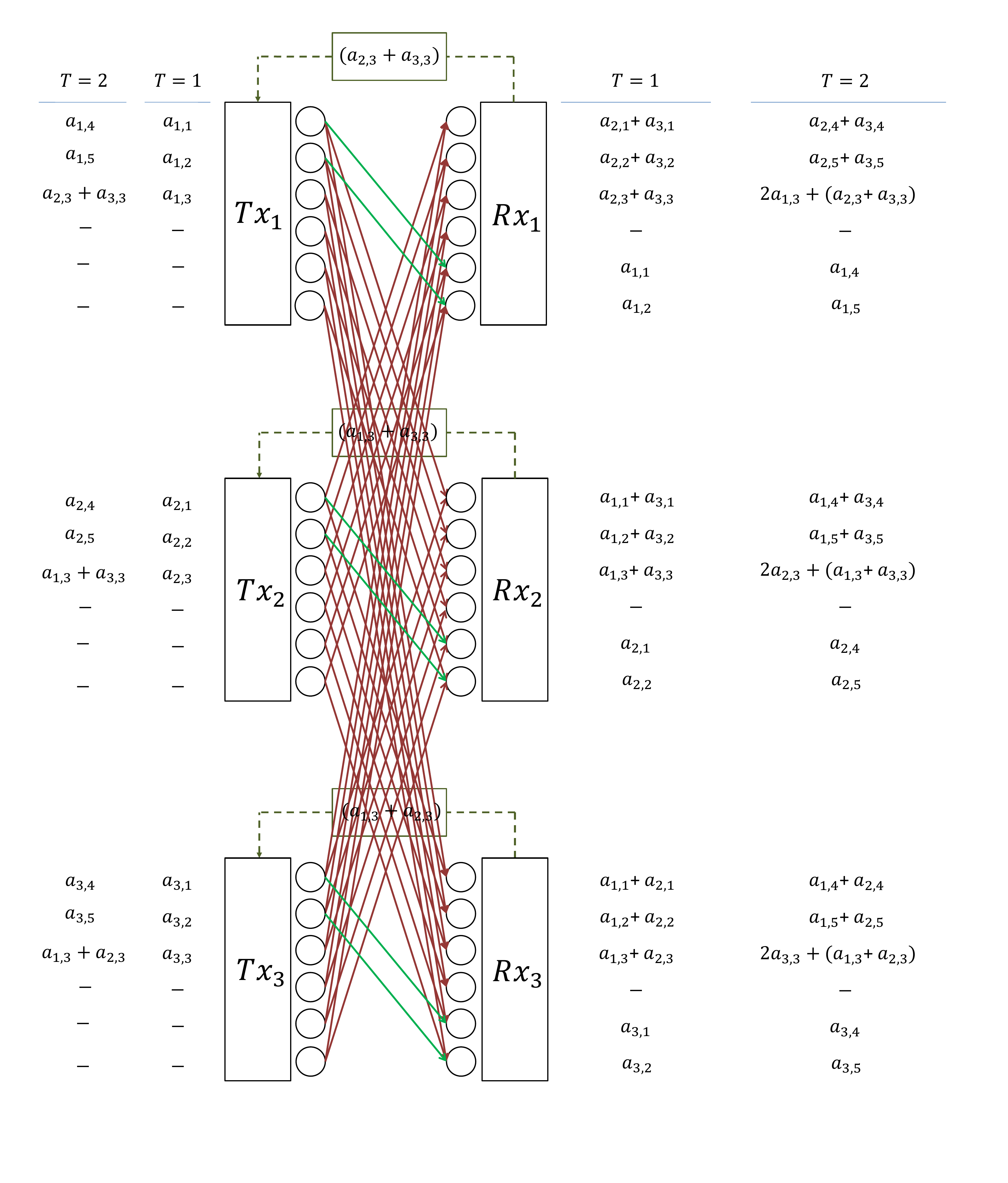}
\caption{Proposed coding scheme for the linear deterministic model in the very strong interference regime ($2 \le\alpha$), for $K=3$, $n=2$, $m=6$ and $p=0.5$.}
\label{fig:Example3}
\end{figure}

As shown in Fig. \ref{fig:Example3}, the proposed coding scheme is able to convey five intended symbols from each transmitter to its respective receiver in two channel uses, i.e., $2R_{sym}=5$. 


\begin{remark}\label{2p-p}
The feedback link is assumed to have a capacity of $C_{fb}$ (or $p{\log s}$ in the deterministic model) bits per channel use. However, the schemes described in the paper use  $2C_{fb}$ (or $2p{\log s}$) bits per channel use in odd slots while $0$ bits in even slots. We now show that the presented scheme where $N$ pairs of blocks are used, with feedback of $2C_{fb}$ (or $2p{\log s}$) bits in the first slot, and no feedback in the second can be adapted to a feedback of capacity $C_{fb}$ (or $p {\log s}$) bits in each slot. In order to see this, consider $2N$ slots in the proposed scheme where $a_{\lceil i/2\rceil,(i-1)\mod 2+1}$, $i = 1, \cdots, 2N$, is transmitted in the forward direction where the first subscript refers to one of the $N$ blocks and the second subscript represents the first transmission, where $a_{j,2}$ is a function of the  $2C_{fb}$ (or $2p {\log s}$) bits of feedback which is based on $a_{j,1}$, $j =1, \cdots, N$. The feedback based on $a_{j,1}$ is referred to as $\{b_{j,1}, b_{j,2}\}$ where each $b_{j,t}$ is of capacity $C_{fb}$ (or $p{\log s}$) bits, $j = 1, \cdots, N$.
\begin{enumerate}
\item $t=1$: The transmitter transmits $a_{1,1}$ and the receiver feeds back $b_{1,1}$.
\item $t=2j$, for $j\in\{1,2,\ldots\}$: The transmitter transmits $a_{j+1,1}$. The receiver has already sent $b_{j,1}$ and thus has information to transmit $b_{j,2}$. Thus, the receiver sends $b_{j,2}$.
\item $t=2j+1$, for $j\in\{1,2,\ldots\}$: The transmitter has received the feedback $b_{j,1}$ and $b_{j2}$ and thus can transmit $a_{j,2}$. The receiver has received $a_{j+1,1}$ and thus sends $b_{j+1,1}$ as  feedback.
\end{enumerate}
We note that we can have $a_{N+1,1}=0$ and thus $N-1$ blocks have been decoded, and as $N\to\infty$,  the same rate can be achieved using $C_{fb}$ (or $p {\log s}$) bits in each slot. Since the same procedure is done at each transmitter and receiver, the signals received at the destination will be the same to compute the desired feedback signals in the above method.
\end{remark}

\section{Gaussian Interference Channel}

\subsection{System Model and Problem Formulation}

In this section, we describe the $K$-user symmetric Gaussian IC which consists of $K$ transmitters and $K$ receivers. Transmitter $i$ has a message $W_i$ that it wishes to send to  receiver $i$. At time $t$, transmitter $i$ transmits a signal $X_{i}[t]$ over the channel with a power constraint ${{\rm tr}(\mathbb{E}(X_{i}X^{\dagger }_{i}))}\le 1$ ($A^\dagger$ is the conjugate of $A$).

The received signal at receiver $i$ at time $t$ is denoted as $Y_{i}[t]$ and can be written as
\begin{eqnarray}
Y_{i}[t]&=&\sqrt{\mathsf {SNR}}X_{i}[t]+\sum_{j=1, j\ne i}^{K}\sqrt{\mathsf {INR}}X_{j}[t]+Z_{i}[t],
\end{eqnarray}
where $Z_{i}[t]\sim\mathsf{CN}(0,1)$ is i.i.d. complex Gaussian noise, ${\mathsf {SNR}}$ is the received signal-to-noise-ratio from transmitter $i$  to receiver $i$, and ${\mathsf {INR}}$ is the received interference-to-noise-ratio from transmitter $i$ to receiver $j$ for $i, j\in \{1,...,K\},$ $i \ne j$. In other words, $\sqrt{\mathsf {SNR}}$ is the power attenuation factor of the direct links and $\sqrt{\mathsf {INR}}$ is the power attenuation factor of the interference links. Let $C_{FB}$ be the capacity of the feedback link from receiver $i$ to transmitter $i$, for all $i\in\{1, 2, \cdots, K\}$. We assume that the feedback channels are orthogonal to each other and they are also orthogonal to the data channels.

The encoding process at each node is causal, in the sense that the feedback signal transmitted from  receiver $i$ at time $t$ is a function of whatever is received over the data channel up to time $(t-1)$; and the transmitted signal by transmitter $i$ at time $t$ is a function of the message $W_i$ and the feedback received till time $t$. Each receiver decodes the message at $t=T$. If a message $W_i\in \{1,\dots ,2^{TR}\}$ transmitted from transmitter $i$ is decoded at receiver $i$ for each $i\in \{1, \cdots, K\}$ with error probability $e_{i,T}={\rm Pr}({\rm \ }\hat{W_i}\ne W_i)\to 0$ as $T\to\infty$, we say that the symmetric rate $R$ is achievable. We assume that $\mathsf{SNR},\mathsf{INR}> 1$ and also define
\begin{equation}
{\alpha}={\frac{{\log \mathsf{INR}}}{{\log \mathsf{SNR}}}\ }, \text{ and } {\beta}={\frac{C_{FB}}{{\log \mathsf{SNR}}}}.
\end{equation}

\subsection{Results of Gaussian IC Model}\label{vjhdz}

\subsubsection{Overview}

In this section, we will describe the achievability scheme for the symmetric $K$-user Gaussian IC with rate-limited feedback. This scheme will be shown to achieve a symmetric rate within a constant gap to a conjectured upper bound, which is the minimum of the symmetric rate upper bound with infinite feedback, and the sum of the symmetric rate upper bound without feedback and the amount of symmetric feedback.

The feedback helps decode the interference which can be useful for decoding the desired message. In addition, feedback helps to decode a part of the intended message that is conveyed from other transmitters through the feedback path. In a $K$-user IC, the receivers hear interference signals from multiple transmitters. Partial decoding of all interfering messages would dramatically decrease the maximum rate of the desired message. Thus, we decode the total interference from all the other users, without resolving the individual components of the interference. Our achievability strategy has two key features, namely, 1) interfering signals are aligned, and 2) the summation of interfering signals  belong to a message set of proper size which can be decoded at each receiver. Here, the first property is satisfied since the network is symmetric (all interfering links have the same gain), and therefore, all interfering messages are received at the same power level. In order to satisfy the second property, we use a common lattice code for all transmitters, instead of random Gaussian codebooks. The structure of a lattice codebook and its closedness with respect to summation imply that the sum of aligned interfering codewords observed at each receiver is still a codeword from the same codebook. This allows us to perform decoding by searching over a single codebook, instead of the Cartesian product of all codebooks.

Lattice codes are a class of codes that can achieve the capacity of the Gaussian channel \cite{Lattice1,Lattice2,nazer,nazer2012successive}, with lower complexity as compared to the conventional random codes. A $T$-dimensional lattice $\Lambda^T$ is a subset of $T$-tuples with real elements, such that $x, y \in \Lambda^T$ implies $-x\in \Lambda^T$ and $x+y \in \Lambda^T$. For an arbitrary $x\in {\mathbb R}^T$, we define $[x \ \mathrm{mod} \ \Lambda^T]=x-Q(x)$, where $Q(x)=\arg \min_{t \in \Lambda^T}||x-t||$, is the closest lattice point to $x$. The Voronoi cell of $\Lambda^T$, denoted by ${\mathcal{V}}_{\Lambda^T}$, is defined as ${\mathcal{V}}_{\Lambda^T}={x \in {\mathbb R}^T:Q(x)=0}$. The Voronoi volume $V({\mathcal{V}}_{\Lambda^T})$ and the second moment $\sigma^2(\Lambda^T)$ of the lattice are defined as $V({\mathcal{V}}_{\Lambda^T})=\int_{{\mathcal{V}}_{\Lambda^T}}dx$ and $\sigma^2(\Lambda^T)=\frac{\int_{{\mathcal{V}}_{\Lambda^T}}||x||^2dx}{TV({\mathcal{V}}_{\Lambda^T})}$, respectively. We further define the normalized second moment of $\Lambda^T$ as $G(\Lambda^T)=\frac{\sigma^2}{V({\mathcal{V}}_{\Lambda^T})^{2/T}}=\frac{\int_{{\mathcal{V}}_{\Lambda^T}}||x||^2dx}{TV({\mathcal{V}}_{\Lambda^T})^{1+\frac{2}{T}}}$. A sequence of lattices $\{\Lambda^T\}$ is called a {\em good quantization code} if $\lim_{T\rightarrow\infty}G(\Lambda^T)=\frac{1}{2\pi e}$. On the other hand, a sequence of lattices is known to be good for AWGN channel coding if $\lim_{T\rightarrow\infty}P[z^T\in {\mathcal{V}}_{\Lambda^T}]=1$, where $z^T\sim \mathcal{N}\left(0_{1\times T},\sigma^2(\Lambda^T)I_{T\times T}\right)$ is a random noise. It is shown in \cite{Lattice3} that there exist sequences of lattices $\{\Lambda^T\}$ that are simultaneously good for quantization and AWGN channel coding.

For the achievability scheme, we use a nested lattice code \cite{Lattice} which is generated using a good quantization lattice for shaping and a good channel coding lattice. We start with $T$-dimensional nested lattices $\Lambda_{c}\subseteq\Lambda_{f}$, where $\Lambda_{c}$ is a good quantization lattice with $\sigma^2(\Lambda_{c})=1$ and $G(\Lambda_{c})\approx1/2\pi e$, and $\Lambda_{f}$ is a good channel coding lattice. We construct a codebook ${\mathcal C}=\Lambda_{f} \cap {{\mathcal{V}}_{\Lambda_{c}}}$, where ${{\mathcal{V}}_{\Lambda_{c}}}$ is the Voronoi cell of the lattice $\Lambda_{c}$. Let $s$ be the lattice codeword in $\Lambda_{f} \cap {\mathcal{V}}_{\Lambda_{c}}$ to which the message is mapped, and build $X={[s - d ]} \ \mathrm{mod} \ \Lambda_{c}$ as the signal to be transmitted, where $d$ is a random dither uniformly distributed over ${\mathcal{V}}_{\Lambda_{c}}$, and shared between all users in the network. We will use the following properties of lattice codes \cite{Lattice2}:
\begin{enumerate}
\item Codebook ${\mathcal C}$ is a closed set with respect to summation under the ``$\mathrm{mod} \ \Lambda_{c}$" operation, i.e., if $x_1, x_2 \in {\mathcal C}$ are two codewords, then $(x_1+x_2) \ \mathrm{mod} \ \Lambda_{c} \in{\mathcal C}$ is also a codeword.

\item Lattice code ${\mathcal C}$ can be used to reliably transmit up to rate $R=\log(\mathsf {SNR})$ over a Gaussian channel modeled by $Y=\sqrt{\mathsf {SNR}}X+Z$ with $Z\sim \mathcal{N}\left(0,1\right)$.
\end{enumerate}

\subsubsection{Proposed Achievability Scheme}

We will now describe our achievability strategies for a $K$-user symmetric Gaussian IC, which is inspired by the proposed schemes in Section \ref{sec:det} for the deterministic IC. We split the result into three regions, denoted as very weak interference where $\alpha \le 1/2$, weak interference where $1/2<\alpha \le 2/3$, and strong interference where $\alpha\ge 2$. We do not consider the case of $2/3<\alpha<2$ since the upper bound for the symmetric capacity with perfect feedback in Theorem 3 of \cite{Mohajer} and the lower bound for the symmetric capacity with no feedback in Theorem 1 of \cite{Ordentlich} are within a constant of $\frac{1}{2}\log9+16+\frac{K-1}{2}+3\log K$ bits to each other for $2/3<\alpha<1$, and are within a constant of $\frac{1}{2}\log6+6+\frac{K-1}{2}+\log K$ bits to each other for $1<\alpha<2$. The achievability for $\alpha=1$ in \cite{Ordentlich,Mohajer} assumes that channel gains are outside an outage set. We also assume for our theorems that the channel gains are not in this outage set. We will use the notation $X^{(a:b)}\triangleq X^{(a)} + X^{(a+1)}+ \cdots + X^{(b)}$.

The next result gives the symmetric achievable rate for the very weak interference regime ($\alpha \le 1/2$).
\begin{theorem}\label{thm_gauss1}
For $\alpha \le 1/2$, a symmetric rate of $R_{sym}=\frac{R^{(1:4)}}{2}$ is achievable, for any $R^{(1)}, \cdots, R^{(4)}$ satisfying
\begin{eqnarray}
R^{(1)} &\le& \log \left(1+ \frac{{\mathsf {SNR}} P^{'(1)} }{{\mathsf {SNR}} P^{'(2:3)}+{\mathsf {SNR}}^{\alpha}P^{'(1:3)}(K-1)+1} \right),\label{1w}\\
R^{(2)} &\le& \log \left( 1+\frac{{\mathsf {SNR}} P^{'(2)}}{{\mathsf {SNR}} P^{'(3)}+{\mathsf {SNR}}^{\alpha} P^{'(1:3)}(K-1)+1} \right),\label{2w}\\
R^{(3)} &\le& \log \left(\frac{1}{K}+\frac{{\mathsf {SNR}} P^{'(3)} }{{\mathsf {SNR}}^{\alpha} P^{'(2:3)} (K-1)+1} \right),\label{3w}\\
R^{(1)} &\le& \log \left(\frac{1}{K}+\frac{{\mathsf {SNR}}^{\alpha} P^{'(1)} }{{\mathsf {SNR}}^{\alpha} P^{'(2:3)}(K-1)+1} \right),\label{4w}\\
R^{(1)}  &\le& 2 C_{FB} - \log K ,\label{7w1}\\
R^{(3)}  &\le& 2 C_{FB} - \log K ,\label{7w2}\\
R^{(1)} &\le& \log \left(\frac{1}{K-1}+\frac{(\sqrt{\mathsf {SNR}}{\left(\frac{1}{K-1}\right)}+\sqrt{{\mathsf {SNR}}^{\alpha}}\left(\frac{K-2}{K-1}\right))^{2} P^{''(1)} }{{\mathsf {SNR}}P^{''(4)}+{\mathsf {SNR}}^{\alpha} (K-1)P^{''(4)}+1} \right),\label{5w}\\
R^{(4)} &\le& \log \left( 1+ \frac{{\mathsf {SNR}} P^{''(4)}}{{\mathsf {SNR}}^{\alpha}(K-1)P^{''(4)}+1} \right),\label{6w}
\end{eqnarray}
for any non-negative set of power values that satisfy $P^{'(1:3)}=1$, $P^{''(1)}+P^{''(4)}=1$, and ${\mathsf {SNR}}P^{'(3)}={\mathsf {SNR}}^{\alpha}P^{'(1)}$.
\end{theorem}

\begin{proof}
Here, we will describe the achievability scheme only for the first user. Due to the symmetry of the scheme, the achievability for the other users is similar.

We take $M_1=\{M_1^{(1)}, M_1^{(2)}, M_1^{(3)}, M_1^{(4)}\}$ as the messages to be transmitted by the first transmitter. In order to encode $M_1^{(i)}$, for $i\in \{1,\ldots, 4\}$, we use the common quantization lattice ($\Lambda_{c_i}=\Lambda_{c}$, $i = 1,\ldots, 4$) but different channel coding lattices ($\Lambda_{f_i}$). The different codebooks ${\mathcal C_i} = \Lambda_{f_i} \cap {\mathcal{V}}_{\Lambda_{c}}$ are assumed to be of size $2^{T R^{(i)}}$. Section III of \cite{Lattice} gives a detailed construction of nested lattice codes. Let $s_1^{(i)}$ be the lattice codeword in $\Lambda_{f_i} \cap {\mathcal{V}}_{\Lambda_{c}}$ to which $M_1^{(i)}$ is mapped, and take $X_1^{(i)}={[s_1^{(i)}-d_1^{(i)}]} \ \mathrm{mod} \ \Lambda_{c}$ where $d_1^{(i)}$ is a random dither uniformly distributed over ${\mathcal{V}}_{\Lambda_{c}}$, and shared between all users in the network. The dithered lattice points can be treated as Gaussian noise in the analysis as shown in
Appendix A of \cite{nazer}. We also take $X_1=\{X_1^{(1)}, X_1^{(2)}, X_1^{(3)}, X_1^{(4)}\}$ as the set of signals that the first user transmits during two consecutive time-slots.

The encoded symbol $X_1^{(i)}$, for $i\in \{1, 2, 3, 4\}$, is of rate $R^{(i)}$ using the lattice codes, for $i\in \{1, 2, 3, 4\}$. The overall rate is thus  $R=\frac{R^{(1:4)}}{2}$. Let $P^{'(i)}$ be the power attenuation of $X_1^{(i)}$ transmitted in the first round, and  $P^{''(i)}$ be the power attenuation of the $X_1^{(i)}$ transmitted in the second round. The power allocations during the two rounds are chosen as
\begin{eqnarray}
P^{'(1)} = \frac{\mu^{(1)}}{\mu^{(1:3)}}, P^{'(2)} = \frac{\mu^{(2)}}{\mu^{(1:3)}}, P^{'(3)} = \frac{\mu^{(3)}}{\mu^{(1:3)}}, P^{'(4)}=0,\nonumber\\
P^{''(1)} = \frac{\mu^{(1)}}{\mu^{(1)}+\mu^{(4)}}, P^{''(2)} = P^{''(3)} =0, P^{''(4)} = \frac{\mu^{(4)}}{\mu^{(1)}+\mu^{(4)}}.
\end{eqnarray}

{\bf Transmission in the first time-slot:} In the first time-slot, the $j^{\text{th}}$ transmitter transmits $\sum_{i=1}^3\sqrt{P^{'(i)}}X_j^{(i)}$, where $X_j^{(i)}$ is of length $T$, for $j \in \{1, \cdots, K\}$.

{\bf Decoding and feedback:} The first receiver  receives
\begin{equation}
Y^{(1)}_1=\sqrt{\mathsf {SNR}}\sum_{i=1}^3\sqrt{P^{'(i)}}X_1^{(i)}+\sqrt{\mathsf {INR}}\sum\limits_{j\neq 1}\sum_{i=1}^3\sqrt{P^{'(i)}}X_j^{(i)}+Z^{(1)}_1.
\end{equation}
It decodes $X^{(1)}_1$, and consequently $s^{(1)}_1$, by treating the rest of the signals as noise. The signal power is $\mathsf{SNR} P^{'(1)}$ and the interference plus noise power is
\begin{equation}
1+ {\mathsf {SNR}} P^{'(2:3)}+{\mathsf {INR}}P^{'(1:3)}(K-1),
\end{equation}
and thus the decoding can be performed since \eqref{1w} holds. After removing $X^{(1)}_1$, then $X^{(2)}_1$, and consequently $s^{(2)}_1$, is decoded by treating the rest as noise. Since the signal power is $\mathsf{SNR} P^{'(2)}$ and the interference plus noise power is
\begin{equation}
1+ {\mathsf {SNR}} P^{'(3)}+{\mathsf {INR}}P^{'(1:3)}(K-1),
\end{equation}
$X^{(2)}_1$ can be decoded since \eqref{2w} holds.

The residual received signal after the contribution of $X^{(1)}_1$ and $X^{(2)}_1$ is removed is
\begin{equation}
\sqrt{\mathsf {SNR}}\sqrt{P^{'(3)}}X_1^{(3)}+\sqrt{\mathsf {INR}}\sum\limits_{j\neq 1}{\sum_{i=1}^3\sqrt{P^{'(i)}}X_j^{(i)}}+Z^{(1)}_1.
\end{equation}

The signal $I_1 \triangleq X^{(3)}_1 + \sum\limits_{j\neq 1}{X_j^{(1)}}$ can be recovered if  \eqref{3w} and \eqref{4w} hold, which follows from Lemma \ref{sjkhgsjkdh2} by considering $X^{(3)}_1$, $X_2^{(1)},\dots,X_K^{(1)}$ as the $K$ signals in the statement of Lemma \ref{sjkhgsjkdh2} which are all received at the same power level (${\mathsf{SNR}}{P^{'(3)}}={\mathsf{INR}}{P^{'(1)}}$).


After obtaining $I_1$, it is sent back to the transmitter. It can be verified using Lemma \ref{sjk2} that the rate of the feedback signal is smaller than the feedback capacity if  \eqref{7w1} and \eqref{7w2} hold.


{\bf Transmission in the second time-slot:} The first transmitter has received $I_1 = X^{(3)}_1 + \sum\limits_{j\neq 1}{X_j^{(1)}}$ from feedback. Since the transmitter already knows $X^{(3)}_1$, and thus obtains $\sum\limits_{j\neq 1}{X_j^{(1)}}$, and consequently transmits
\begin{equation}
\frac{\sqrt{P^{''(1)}}}{K-1}\sum\limits_{j\neq 1}{X_j^{(1)}}+\sqrt{P^{''(4)}}X_1^{(4)},
\end{equation}
and similarly, in general, the $i^{\text{th}}$ transmitter, $\forall i\in \{1, \cdots, K\}$, sends
\begin{equation}
\frac{\sqrt{P^{''(1)}}}{K-1}\sum\limits_{j\neq i}{X_j^{(1)}}+\sqrt{P^{''(4)}}X_i^{(4)},
\end{equation}

{\bf Decoding:} The first receiver receives the signal
\begin{equation}
Y^{(2)}_1=\sqrt{\mathsf {SNR}}\left(\frac{\sqrt{P^{''(1)}}}{K-1}\sum\limits_{j\neq 1}{X_j^{(1)}}+\sqrt{P^{''(4)}}X_1^{(4)}\right)+\sqrt{\mathsf {INR}}\sum\limits_{j\neq 1}{\left(\frac{\sqrt{P^{''(1)}}}{K-1}\sum\limits_{i\neq j}{X_i^{(1)}}+\sqrt{P^{''(4)}}X_j^{(4)}\right)}+Z^{(2)}_1.
\end{equation}

First, the receiver subtracts the  $X_1^{(1)}$ term from the received signal and obtains 
\begin{equation}
Y'^{(2)}_1=\left(\sqrt{\mathsf {SNR}}+(K-2)\sqrt{\mathsf {INR}}\right)\frac{\sqrt{P^{''(1)}}}{K-1}\sum\limits_{j\neq 1}{X_j^{(1)}}+\sqrt{\mathsf {SNR}}\sqrt{P^{''(4)}}X_1^{(4)}+\sqrt{\mathsf {INR}}\sqrt{P^{''(4)}}\sum\limits_{j\neq 1}X_j^{(4)}+Z^{(2)}_1.
\end{equation}
Let $I_2 \triangleq \sum\limits_{j\neq 1}{X_j^{(1)}}$. The receiver obtains $I_2$ treating $\sqrt{\mathsf {SNR}}\sqrt{P^{''(4)}}X_1^{(4)} + \sqrt{\mathsf {INR}}\sqrt{P^{''(4)}}\sum\limits_{j\neq 1}X_j^{(4)}$ as noise. It can be seen that $I_2$ can be obtained if  \eqref{5w} holds for $R_1$, which follows from Lemma \ref{sjkhgsjkdh2} (with  $X^{(1)}_2,\dots,X_K^{(1)}$ as the $(K-1)$ signals). Having decoded $I_1$ and $I_2$, then $X_1^{(3)}$ can be decoded since it is the difference of the two. Having $I_2$ decoded, the residual signal is
\begin{equation}
\sqrt{\mathsf {SNR}}\sqrt{P^{''(4)}}X_1^{(4)} + \sqrt{\mathsf {INR}}\sqrt{P^{''(4)}}\sum\limits_{j\neq 1}X_j^{(4)} + Z^{(2)}_1,
\end{equation}
from which $X^{(4)}_1$ can be decoded by treating
\begin{equation}
\sqrt{\mathsf {INR}}\sqrt{P^{''(4)}}\sum\limits_{j\neq 1}X_j^{(4)} + Z^{(2)}_1,
\end{equation}
as noise since \eqref{6w} holds.
\end{proof}

The next result gives the symmetric achievable rate for the weak interference regime ($1/2<\alpha \le 2/3$).

\begin{theorem}\label{thm_gauss2} For $1/2<\alpha \le 2/3$, the symmetric rate of $R_{sym}=\frac{R^{(1:6)}}{2}$ is achievable, for any $R^{(1)}, \cdots, R^{(6)}$ satisfying

\begin{eqnarray}
R^{(1)} &\le& \log \left(1+ \frac{{\mathsf {SNR}} P^{'(1)} }{{\mathsf {SNR}} P^{'(2:4)}+{\mathsf {SNR}}^{\alpha} (K-1)P^{'(1:4)}+1} \right),\label{1}\\
R^{(2)} &\le& \log \left(1+ \frac{{\mathsf {SNR}} P^{'(2)}}{{\mathsf {SNR}} P^{'(3:4)}+{\mathsf {SNR}}^{\alpha} (K-1)P^{'(1:4)}+1} \right),\label{2}\\
R^{(1)} &\le& \log \left( \frac{1}{K-1}+\frac{{\mathsf {SNR}}^{\alpha} P^{'(1)} }{{\mathsf {SNR}} P^{'(3:4)}+{\mathsf {SNR}}^{\alpha} (K-1)P^{'(2:4)}+1} \right),\label{3}\\
R^{(3)} &\le& \log \left(\frac{1}{K}+\frac{{\mathsf {SNR}}P^{'(3)}}{{\mathsf {SNR}}P^{'(4)}+{\mathsf {SNR}}^{\alpha}  (K-1)P^{'(3:4)}+1} \right),\label{4}\\
R^{(2)} &\le& \log \left(\frac{1}{K}+\frac{{\mathsf {SNR}}^{\alpha} P^{'(2)}}{{\mathsf {SNR}}P^{'(4)}+{\mathsf {SNR}}^{\alpha} (K-1)P^{'(3:4)}+1} \right),\label{5}\\
R^{(4)} &\le& \log \left( 1+\frac{{\mathsf {SNR}} P^{'(4)} }{{\mathsf {SNR}}^{\alpha} (K-1)P^{'(3:4)}+1} \right),\label{6}\\
R^{(2)} &\le& 2 C_{FB}-\log K,\label{11.1}\\
R^{(3)} &\le& 2 C_{FB}-\log K,\label{11.2}\\
R^{(5)} &\le& \log \left( 1+\frac{{\mathsf {SNR}} P^{''(5)}}{{\mathsf {SNR}}(P^{''(2)}+P^{''(6)})+{\mathsf {SNR}}^{\alpha}  (K-1)(P^{''(2)}+P^{''(5:6)})+1-{\mathsf {SNR}}^{\alpha} P^{''(2)} } \right),\label{7}\\
R^{(2)} &\le& \log \left(\frac{1}{K-1}+\frac{\left(\sqrt{\mathsf {SNR}}\left(\frac{1}{K-1}\right)+\sqrt{{\mathsf {SNR}}^{\alpha}}\left(\frac{K-2}{K-1}\right)\right)^{2} P^{''(2)}}{{\mathsf {SNR}}P^{''(6)}+{\mathsf {SNR}}^{\alpha} (K-1) P^{''(5:6)}+1} \right),\label{8}\\
R^{(5)} &\le& \log \left(\frac{1}{K-1}+\frac{{\mathsf {SNR}}^{\alpha}  P^{''(5)} }{{\mathsf {SNR}} P^{''(6)}+{\mathsf {SNR}}^{\alpha} (K-1) P^{''(6)}+1} \right),\label{9}\\
R^{(6)} &\le& \log \left( 1+\frac{{\mathsf {SNR}}P^{''(6)}}{{\mathsf {SNR}}^{\alpha} (K-1)P^{''(6)}+1} \right),\label{10}
\end{eqnarray}
for any non-negative set of power values that satisfy $P^{'(1:4)}=1$, $P^{''(2)}+P^{''(5:6)}=1$, and ${\mathsf {SNR}}P^{'(3)}={\mathsf {SNR}}^{\alpha}P^{'(2)}$.
\end{theorem}

\begin{proof}
We take $M_1=\{M_1^{(1)}, M_1^{(2)}, M_1^{(3)}, M_1^{(4)}, M_1^{(5)}, M_1^{(6)}\}$ as the messages to be transmitted by the first transmitter. In order to encode $M_1^{(i)}$, for $i\in \{1,\ldots, 6\}$, we use the common quantization lattice  but different channel coding lattices ($\Lambda_{c_i}=\Lambda_{c}$, $i = 1,\ldots, 6$). The different codebooks ${\mathcal C_i} = \Lambda_{f_i} \cap {\mathcal{V}}_{\Lambda_{c}}$ are assumed to be of size $2^{T R^{(i)}}$. Let $s_1^{(i)}$ be the lattice codeword in $\Lambda_{f_i} \cap {\mathcal{V}}_{\Lambda_{c}}$ to which $M_1^{(i)}$ is mapped, and take $X_1^{(i)}={[s_1^{(i)}-d_1^{(i)}]} \ \mathrm{mod} \ \Lambda_{c}$ where $d_1^{(i)}$ is a random dither uniformly distributed over ${\mathcal{V}}_{\Lambda_{c}}$, and shared between all users in the network. We also take $X_1=\{X_1^{(1)}, X_1^{(2)}, X_1^{(3)}, X_1^{(4)}, X_1^{(5)}, X_1^{(6)}\}$ as the set of signals that the first user transmits during two consecutive time-slots.

The encoded symbol $X_1^{(i)}$, for $i\in \{1, 2, 3, 4, 5, 6\}$, is of rate $R^{(i)}$ using the lattice codes, for $i\in \{1, 2, 3, 4, 5, 6\}$. The overall rate is thus  $R=\frac{R^{(1:6)}}{2}$. Let $P^{'(i)}$ be the power attenuation of $X_1^{(i)}$ transmitted in the first round, and  $P^{''(i)}$ be the power attenuation of the $X_1^{(i)}$ transmitted in the second round. The power allocations during the two rounds are chosen as
\begin{eqnarray}
P^{'(1)} = \frac{\mu^{(1)}}{\mu^{(1:4)}}, P^{'(2)} = \frac{\mu^{(2)}}{\mu^{(1:4)}}, P^{'(3)} = \frac{\mu^{(3)}}{\mu^{(1:4)}}, P^{'(4)} = \frac{\mu^{(4)}}{\mu^{(1:4)}},P^{'(5)} =P^{'(6)} =0,\nonumber\\
P^{''(1)} =0, P^{''(2)} = \frac{\mu^{(2)}}{\mu^{(2)}+\mu^{(5:6)}}, P^{''(3)} =P^{''(4)} =0, P^{''(5)} = \frac{\mu^{(5)}}{\mu^{(2)}+\mu^{(5:6)}}, P^{''(6)} = \frac{\mu^{(6)}}{\mu^{(2)}+\mu^{(5:6)}}.
\end{eqnarray}

{\bf Transmission in the first time-slot:} In the first time-slot, the $j^{\text{th}}$ transmitter transmits $\sum_{i=1}^4\sqrt{P^{'(i)}}X_j^{(i)}$, where $X_j^{(i)}$ is of length $T$, for $j \in \{1, \cdots, K\}$.

{\bf Decoding and feedback:} The first receiver receives
\begin{equation}
Y^{(1)}_1=\sqrt{\mathsf {SNR}}\sum_{i=1}^4\sqrt{P^{'(i)}}{X_1^{(i)}}+\sqrt{\mathsf {INR}}\sum\limits_{j\neq 1}{\sum_{i=1}^4\sqrt{P^{'(i)}}X_j^{(i)}}+Z^{(1)}_1.
\end{equation}
The receiver first decodes $X^{(1)}_1$, and consequently $s^{(1)}_1$, by treating the rest of the signals as noise. Due to the rate constraint \eqref{1}, $X^{(1)}_1$ can be decoded. After cancelling the signals containing $X^{(1)}_1$, then $X^{(2)}_1$, and consequently $s^{(2)}_1$, can further be decoded by treating the remaining signals as noise due to \eqref{2}. After cancelling $X^{(1)}_1$ and $X^{(2)}_1$, the receiver obtains the lattice point $\sum\limits_{i\neq 1}{X_i^{(1)}}$ as the sum of $(K-1)$ lattice points which are all received at the same power level, by treating all the other signals as noise. The signal power is $\mathsf{INR}{P^{'(1)}}$ and the interference plus noise power is
\begin{equation}
1+\mathsf{SNR}P^{'(3:4)} + \mathsf{INR}(K-1)P^{'(2:4)}.
\end{equation}
The lattice point can be obtained if \eqref{3} holds which can be seen using Lemma \ref{sjkhgsjkdh2} (with $X_2^{(1)},\dots,X_K^{(1)}$ as the $(K-1)$ signals which are all received at the same power level). Then, the residual signal is
\begin{align}
&\sqrt{\mathsf {SNR}}\sum_{i=3}^4\sqrt{P^{'(i)}}{X_1^{(i)}}+\sqrt{\mathsf {INR}}\sum\limits_{j\neq 1}{\sum_{i=2}^4\sqrt{P^{'(i)}}X_j^{(i)}}+Z^{(1)}_1 \nonumber\\
&= \sqrt{\mathsf {SNR}}\sqrt{P^{'(3)}}{X_1^{(3)}}+\sqrt{\mathsf {INR}}\sqrt{P^{'(2)}}\sum\limits_{j\neq 1}X_j^{(2)}+\sqrt{\mathsf {SNR}}\sqrt{P^{'(4)}}{X_1^{(4)}}+\sqrt{\mathsf {INR}}\sum\limits_{j\neq 1}{\sum_{i=3}^4\sqrt{P^{'(i)}}X_j^{(i)}}+Z^{(1)}_1.
\end{align}
Let $I_1 \triangleq X^{(3)}_1+ \sum\limits_{j\neq 1}{X_j^{(2)}}$. We can obtain $I_1$ treating $\sqrt{\mathsf {SNR}}\sqrt{P^{'(4)}}{X_1^{(4)}}+\sqrt{\mathsf {INR}}\sum\limits_{j\neq 1}{\sum_{i=3}^4\sqrt{P^{'(i)}}X_j^{(i)}}$ as noise if  \eqref{4} and \eqref{5} hold using Lemma  \ref{sjkhgsjkdh2}  (with  $X^{(3)}_1$, $X_2^{(2)},\dots,X_K^{(2)}$ as the $K$ signals) and ${\mathsf{SNR}}{P^{'(3)}}={\mathsf{INR}}{P^{'(2)}}$. After decoding $I_1$, it is sent back to the transmitter. It can be verified using Lemma \ref{sjk2} that the rate of the feedback signal is smaller than the feedback capacity if  \eqref{11.1} and \eqref{11.2} hold. After cancelling $I_1$, and then $X^{(4)}_1$, $s^{(4)}_1$ can be obtained due to  \eqref{6}.

{\bf Transmission in the second time-slot:} For the first transmitter, ${X_1^{(3)}}$ is known and $I_1$ is given from the feedback, and consequently the first transmitter obtains $\sum\limits_{j\neq 1}X_j^{(2)}$. Using this, it transmits
\begin{equation}
\frac{\sqrt{P^{''(2)}}}{K-1}\sum\limits_{j\neq 1}{X_j^{(2)}}+\sum_{i=5}^{6}{\sqrt {P^{''(i)}}}X_1^{(i)}.
\end{equation}
In general, the $k^{th}$ transmitter, $\forall k\in \{1, \cdots, K\}$, transmits
\begin{equation}
\frac{\sqrt{P^{''(2)}}}{K-1}\sum\limits_{j\neq k}{X_j^{(2)}}+\sum_{i=5}^{6}{\sqrt {P^{''(i)}}}X_k^{(i)}.
\end{equation}

{\bf Decoding:} The first receiver receives
\begin{equation}
Y^{(2)}_1=\sqrt{\mathsf {SNR}}\left(\frac{{\sqrt {P^{''(2)}}}}{K-1}\sum\limits_{j\neq 1}{X_j^{(2)}}+ \sum_{i=5}^{6}{\sqrt {P^{''(i)}}}X_1^{(i)}\right)+ \\ \sqrt{\mathsf {INR}}\sum\limits_{j\neq 1}{\left(\frac{P^{''(2)}}{K-1}\sum\limits_{i\neq j}{X_i^{(2)}}+ \sum_{i=5}^{6}{\sqrt {P^{''(i)}}}X_j^{(i)}\right)}+Z^{(2)}_1.
\end{equation}
Based on this, the receiver cancels the signal $X_1^{(2)}$ and obtains
\begin{equation}
\sqrt{\mathsf {SNR}}\left(\frac{{\sqrt {P^{''(2)}}}}{K-1}\sum\limits_{j\neq 1}{X_j^{(2)}}+ \sum_{i=5}^{6}{\sqrt {P^{''(i)}}}X_1^{(i)}\right)+ \\ \sqrt{\mathsf {INR}}\frac{P^{''(2)}(K-2)}{K-1} \sum\limits_{j\neq 1}{X_j^{(2)}} + \sqrt{\mathsf {INR}}\sum\limits_{j\neq 1}\sum_{i=5}^{6}{\sqrt {P^{''(i)}}}X_j^{(i)}+Z^{(2)}_1.
\end{equation}
From this residual signal, $X^{(5)}_1$ can be decoded by treating the other signals as noise due to \eqref{7}. Let $I_2 \triangleq \sum\limits_{j\neq 1}{X_j^{(2)}}$. Note that $I_2$ is a lattice point in ${\mathcal C}_2$, and thus we can obtain $I_2$ treating ${\sqrt {P^{''(6)}}}X_1^{(6)}+  \sqrt{\mathsf {INR}}\sum\limits_{j\neq 1}\sum_{i=5}^{6}{\sqrt {P^{''(i)}}}X_j^{(i)}$ as noise. Since the term $I_2$ is a lattice point in ${\mathcal C}_2$ which is a codebook of rate $R_2$, it can be obtained if  \eqref{8} holds for $R_2$, which follows from Lemma \ref{sjkhgsjkdh2}  (with  $X_2^{(2)},\dots,X_K^{(2)}$ as the $(K-1)$ signals). From  $I_1$ and $I_2$, $X_1^{(3)}$ can be obtained since it is the difference of the two. Further, $\sum\limits_{j\neq 1}{X_j^{(5)}}$ can be obtained after cancelling $I_2$ due to equation \eqref{9} and Lemma \ref{sjkhgsjkdh2}  (with $X_2^{(5)},\dots,X_K^{(5)}$ as the $(K-1)$ signals). After cancelling  $\sum\limits_{j\neq 1}{X_j^{(5)}}$, the residual signal is
\begin{equation}
{\sqrt {P^{''(6)}}}X_1^{(6)}+  \sqrt{\mathsf {INR}}\sum\limits_{j\neq 1}{\sqrt {P^{''(6)}}}X_j^{(6)}+Z^{(2)}_1.
\end{equation}
From this, $X^{(6)}_1$ can be decoded by treating $X_j^{(6)}$, $j\ne 1$ as noise due to \eqref{10}.
\end{proof}

The next result gives the symmetric achievable rate for the strong interference regime ($\alpha\ge 2$).

\begin{theorem}\label{thm_gauss3} For $\alpha \ge 2$, the symmetric rate of $R_{sym}=\frac{R^{(1:3)}}{2}$ is achievable, for any $R^{(1)}, \cdots, R^{(3)}$ satisfying
\begin{eqnarray}
R^{(1)} &\le& \log \left(\frac{1}{K-1} + \frac{{\mathsf {INR}} P^{'(1)} }{{\mathsf {SNR}} P^{'(1:2)}+{\mathsf {SNR}}^{\alpha}(K-1)P^{'(2)}+1} \right),\label{1s}\\
R^{(2)} &\le& \log \left(\frac{1}{K-1} + \frac{{\mathsf {INR}} P^{'(2)} }{{\mathsf {SNR}} P^{'(1:2)}+1 } \right),\label{2s}\\
R^{(1)} &\le& \log \left( 1+\frac{{\mathsf {SNR}} P^{'(1)} }{{\mathsf {SNR}}P^{'(2)}+1 } \right),\label{3s}\\
R^{(2)} &\le& 2 C_{FB}-\log{(K-1)},\label{7s}\\
R^{(3)} &\le& \log \left(\frac{1}{K-1} + \frac{{\mathsf {SNR}}^{\alpha} P^{''(3)} }{{\mathsf {SNR}}^{\alpha} P^{''(2)}+{\mathsf {SNR}}P^{''(3)}+1} \right),\label{4s}\\
R^{(2)} &\le& \log \left(1+ \frac{{\mathsf {SNR}}^{\alpha} P^{''(2)}}{{\mathsf {SNR}}P^{''(3)}+1} \right),\label{5s}\\
R^{(3)} &\le& \log \left( 1+\frac{{\mathsf {SNR}} P^{''(3)}}{1} \right),\label{6s}
\end{eqnarray}
for any non-negative set of power values that satisfy $P^{'(1:2)}=1$, and $P^{''(2:3)}=1$. 
\end{theorem}

\begin{proof}
We take $M_1=\{M_1^{(1)}, M_1^{(2)}, M_1^{(3)}\}$ as the messages to be transmitted by the first transmitter. In order to encode $M_1^{(i)}$, for $i\in \{1, 2, 3\}$, we use the common quantization lattice ($\Lambda_{c_i}=\Lambda_{c}$, $i = 1, 2, 3$) but different channel coding lattices ($\Lambda_{f_i}$, $i = 1, 2, 3$). The different codebooks ${\mathcal C_i} = \Lambda_{f_i} \cap {\mathcal{V}}_{\Lambda_{c}}$ are assumed to be of size $2^{T R^{(i)}}$. Let $s_1^{(i)}$ be the lattice codeword in $\Lambda_{f_i} \cap {\mathcal{V}}_{\Lambda_{c}}$ to which $M_1^{(i)}$ is mapped, and take $X_1^{(i)}={[s_1^{(i)}-d_1^{(i)}]} \ \mathrm{mod} \ \Lambda_{c}$ where $d_1^{(i)}$ is a random dither uniformly distributed over ${\mathcal{V}}_{\Lambda_{c}}$, and shared between all users in the network. We also take $X_1=\{X_1^{(1)}, X_1^{(2)}, X_1^{(3)}\}$ as the set of signals that the first user transmits during two consecutive time-slots.

The encoded symbol $X_1^{(i)}$, for $i\in \{1, 2, 3\}$, is of rate $R^{(i)}$ using the lattice codes, for $i\in \{1, 2, 3\}$. The overall rate is thus  $R=\frac{R^{(1:3)}}{2}$. Let $P^{'(i)}$ be the power attenuation of $X_1^{(i)}$ transmitted in the first round, and  $P^{''(i)}$ be the power attenuation of the $X_1^{(i)}$ transmitted in the second round. The power allocations during the two rounds are chosen as
\begin{eqnarray}
P^{'(1)} = \frac{\mu^{(1)}}{\mu^{(1:2)}}, P^{'(2)} = \frac{\mu^{(2)}}{\mu^{(1:2)}}, P^{'(3)}=0,\nonumber\\
P^{''(1)}=0, P^{''(2)} = \frac{\mu^{(2)}}{\mu^{(2:3)}}, P^{''(3)} = \frac{\mu^{(3)}}{\mu^{(2:3)}}.
\end{eqnarray}

{\bf Transmission in the first time-slot:} In the first time-slot, the $j^{th}$ transmitter, $\forall j\in \{1, \ldots, K\}$, transmits $\sum_{i=1}^2\sqrt{P^{'(i)}}X_j^{(i)}$, where $X_j^{(i)}$ is of length $T$, for $j \in \{1, \cdots, K\}$.

{\bf Decoding and feedback:}  The first receiver receives
\begin{eqnarray}
Y^{(1)}_1=\sqrt{\mathsf {SNR}}\sum_{i=1}^2\sqrt{P^{'(i)}}X_1^{(i)}+\sqrt{\mathsf {INR}}\sum\limits_{j\neq 1}\sum_{i=1}^2\sqrt{P^{'(i)}}X_j^{(i)}+Z^{(1)}_1.
\end{eqnarray}

Let $I_1 \triangleq \sum\limits_{j\neq 1}{X_j^{(1)}}$ and $I_2 \triangleq \sum\limits_{j\neq 1}{X_j^{(2)}}$. Note that $I_1$ is a lattice point in ${\mathcal C}_1$, and thus we can obtain $I_1$ treating the rest of the signals as noise if \eqref{1s} holds, which follows from Lemma \ref{sjkhgsjkdh2}  (with $X_2^{(1)},\dots,X_K^{(1)}$ as the $(K-1)$ signals which are all received at the same power level). Further, $I_2$ is a lattice point in ${\mathcal C}_2$, and  we can obtain $I_2$ treating the rest of the signals as noise if \eqref{2s} holds, which follows from  Lemma \ref{sjkhgsjkdh2} (with $X_2^{(2)},\dots,X_K^{(2)}$ as the $(K-1)$ signals which are all received at the same power level). After cancelling $I_2$, the residual signal is
\begin{eqnarray}
\sqrt{\mathsf {SNR}}\sum_{i=1}^2\sqrt{P^{'(i)}}X_1^{(i)}+Z^{(1)}_1,
\end{eqnarray}
from which $X^{(1)}_1$ can be obtained by treating $X^{(2)}_1$ as noise due to \eqref{3s}.

Also, after obtaining $I_2$, it is sent back to the transmitter. It can be verified using Lemma \ref{sjk2} that the rate of the feedback signal is smaller than the feedback capacity if  \eqref{7s} holds.

{\bf Transmission in the second time-slot:} The first transmitter has received $I_2$ and transmits \begin{eqnarray}
\frac{\sqrt{P^{''(2)}}}{K-1}I_2+\sqrt{P^{''(3)}}X_1^{(3)}.
\end{eqnarray}
In general, the $j^{\text{th}}$ transmitter transmits
\begin{eqnarray}
\frac{\sqrt{P^{''(2)}}}{K-1}\sum\limits_{i\neq j}{X_i^{(2)}}+\sqrt{P^{''(3)}}X_j^{(3)}.
\end{eqnarray}

{\bf Decoding:} The first receiver receives
\begin{eqnarray}
Y^{(2)}_1=\sqrt{\mathsf {SNR}}\left(\frac{\sqrt{P^{''(2)}}}{K-1}I_2+\sqrt{P^{''(3)}}X_1^{(3)}\right)+\\\sqrt{\mathsf {INR}}\sum\limits_{j\neq 1}{\left(\frac{\sqrt{P^{''(2)}}}{K-1}\sum\limits_{i\neq j}{X_i^{(2)}}+\sqrt{P^{''(3)}}X_j^{(3)}\right)}+Z^{(2)}_1.
\end{eqnarray}
Since the receiver knows $I_2$, it can be subtracted to get the residual signal
\begin{eqnarray}
\sqrt{\mathsf {INR}}{\sqrt{P^{''(2)}}}{X_1^{(2)}}+\sqrt{\mathsf {SNR}}\sqrt{P^{''(3)}}X_1^{(3)} +\sqrt{\mathsf {INR}}\sqrt{P^{''(3)}}\sum\limits_{j\neq 1}X_j^{(3)}+Z^{(2)}_1.
\end{eqnarray}
From this, $\sum\limits_{j\neq 1}X_j^{(3)}$ can be obtained if \eqref{4s} holds, which  follows from Lemma \ref{sjkhgsjkdh2} (with $X_2^{(3)},\dots,X_K^{(3)}$ as the $(K-1)$ signals which are all received at the same power level). Afterwards, the residual signal is
\begin{eqnarray}
\sqrt{\mathsf {INR}}{\sqrt{P^{''(2)}}}{X_1^{(2)}}+\sqrt{\mathsf {SNR}}\sqrt{P^{''(3)}}X_1^{(3)}+Z^{(2)}_1.
\end{eqnarray}
Then, $X^{(2)}_1$ can be decoded by treating $X_1^{(3)}$ as noise due to \eqref{5s}. Finally, after cancelling $X_1^{(2)}$, then $X^{(3)}_1$ can be decoded due to \eqref{6s}.
\end{proof}

The following corollary improves the achievability region in the above theorems for the case of $K=2$.
\begin{corollary}
For the case of two-user channel ($K=2$):
\begin{itemize}
\item Theorem \ref{thm_gauss1} without extra $\log{K}$ terms  in equations \eqref{7w1}-\eqref{7w2} still hold.
\item Theorem \ref{thm_gauss2} without extra $\log{K}$ terms  in equations \eqref{11.1}-\eqref{11.2} still hold.
\end{itemize}
More formally, for the case of two-user channel ($K = 2$), the region given in Theorem 2 is still achievable if we replace
\begin{eqnarray}
R^{(1)}  &\le& 2 C_{FB},\nonumber\\
R^{(3)}  &\le& 2 C_{FB},\nonumber
\end{eqnarray}
with equations \eqref{7w1}-\eqref{7w2}, and the region given in Theorem 3 is still achievable if we replace
\begin{eqnarray}
R^{(2)} &\le& 2 C_{FB},\nonumber\\
R^{(3)} &\le& 2 C_{FB}.\nonumber
\end{eqnarray}
with equations \eqref{11.1}-\eqref{11.2}.
\end{corollary}
\begin{proof}
Here we only provide the proof for the statement on the case of $K=2$ for Theorem \ref{thm_gauss1}. The statement regarding Theorem \ref{thm_gauss2} can be shown similarly. Since we have set ${\mathsf{SNR}}{P^{'(3)}}={\mathsf{INR}}{P^{'(1)}}$,  then $X^{(3)}_1 + X_2^{(1)}$ is a lattice point. We consider a slightly modified achievability strategy than that in Theorem \ref{thm_gauss1},  where after receiver $1$ derives $X^{(3)}_1 + X_2^{(1)}$ by treating other codewords as noise, instead of sending back $X^{(3)}_1 + X_2^{(1)}$ as in Theorem \ref{thm_gauss1}, we only feed back $[X^{(3)}_1 + X_2^{(1)}] \ {\Lambda_{c}}$ to transmitter $1$. The rate of the feedback is lower than the capacity of the feedback link, $R^{(i)} \le 2 C_{FB}$, $i=1,3$. Then, transmitter $1$, given $[X^{(3)}_1 + X_2^{(1)}] \ {\Lambda_{c}}$ and  $X^{(3)}_1$, can find $X_2^{(1)}$ and the rest of the strategy is the same as that in Theorem \ref{thm_gauss1}.
\end{proof}

{\bf Proof.} Here we only provide the proof for the statement on the case of $K=2$ for Theorem 2. The statement regarding Theorem 3 can be shown similarly. Since we have set ${\mathsf{SNR}}{P^{'(3)}}={\mathsf{INR}}{P^{'(1)}}$, then $X^{(3)}_1 + X_2^{(1)}$ is a lattice point. We consider a slightly modified achievability strategy than that in Theorem 2,  where after receiver $1$ derives $X^{(3)}_1 + X_2^{(1)}$ by treating other codewords as noise, instead of sending back $X^{(3)}_1 + X_2^{(1)}$ as in Theorem 2, we only feed back $[X^{(3)}_1 + X_2^{(1)}] \ {\Lambda_{c}}$ to transmitter $1$. The rate of the feedback is lower than the capacity of the feedback link, $R^{(i)} \le 2 C_{FB}$, $i=1,3$. Then, transmitter $1$, given $[X^{(3)}_1 + X_2^{(1)}] \ {\Lambda_{c}}$ and  $X^{(3)}_1$, can find $X_2^{(1)}$ and the rest of the strategy is the same as that in Theorem 2.$\hspace{5in}\blacksquare$

\begin{remark}
Based on Lemma 1 of \cite{nazer2012successive}, as long as constraints (8)-(15), (26)-(37), and (47)-(53) hold in the statements of Theorems 2, 3, and 4, respectively, in all places that sum of codewords are declared decodable over modulo algebra in proofs of these theorems, then consequently sum of codewords are decodable over reals, too.
\end{remark}


\subsubsection{A Conjectured Upper Bound}

According to Theorem 1 of \cite{Etkin}, an upper bound on the symmetric capacity without feedback is given by
\begin{eqnarray}\label{sefr}
{R}^u_{sym,0}=\min\left\{\log(1+{\mathsf{SNR}}),\log\left(1+{\mathsf{INR}}+\frac{{\mathsf{SNR}}}{1+\mathsf{INR}}\right)\right\}.
\end{eqnarray}
Moreover, according to Section VI of \cite{Mohajer}, an upper bound on the symmetric capacity with infinite feedback is given by
\begin{eqnarray}\label{binaha}
{R}^u_{sym,\infty}=\frac{1}{2}\log\left(1+\frac{{\mathsf{SNR}}}{1+\mathsf{INR}}\right)+
\frac{1}{2}\log\left(1+{\mathsf{SNR}}+{\mathsf{INR}}\right)+\frac{K-1}{2}+\log K.
\end{eqnarray}
We conjecture that the following upper bound holds for a $K$-user symmetric Gaussian IC with rate-limited feedback
\begin{equation}\label{conj}
{R}^u_{sym}=\min\left\{{R}^u_{sym,\infty},{R}^u_{sym,0}+C_{FB}\right\}.
\end{equation}
Note that the conjecture holds true for $K=2$ as shown in \cite{Vahid}.

The next result shows that the achievable symmetric rate given in the last section is within a constant number of bits to the conjectured upper bound ${R}^u_{sym}$ for a particular choice of the parameters for each interference regime.

\begin{theorem}\label{thm_gauss}
For the $K$-user symmetric Gaussian IC with rate-limited feedback with $\frac{\mathsf{INR}}{\mathsf{SNR}}\notin \left(\frac{1}{2},2\right)$ and $\mathsf{SNR}, \mathsf{INR} \ge 1$, there is an achievability scheme that achieves a symmetric rate within a constant $L$ bits to ${R}^u_{sym}$, where
\begin{eqnarray}\label{distance}
L&=&\max\left\{\frac{1}{2}\log\left(2304\left(K-1\right)^2{K^2}\left(K+
{\frac{1}{3}}\right)
\left(K+{\frac{2}{3}}\right)^2\left(K+2\right)^2\left(K
+{\frac{11}{4}}\right)\right),\log3+16+\log K^3 \right\}\nonumber\\
&&+\frac{K-1}{2}.
\end{eqnarray}
\end{theorem}

\begin{proof}
The detailed proof for this result is provided in Appendix \ref{apdx_gap}. The parameters $\mu^{(i)}$ of the achievability scheme that are chosen for this result are as follows.

Case 1 ($\alpha \le \frac{1}{2}$): We take $\mu^{(1)} = \frac{1}{2{\mathsf {INR}}} \min \{2^{2 C_{FB}},{\mathsf {INR}}-1\}$, $\mu^{(2)} =\frac{1}{{\mathsf {INR}}}-\frac{1}{2{\mathsf {SNR}}} \min \{2^{2 C_{FB}},{\mathsf {INR}}-1\}$, and $\mu^{(4)} = \frac{1}{{\mathsf {INR}}}$ in Theorem \ref{thm_gauss1}.

Case 2 ($\frac{1}{2} \le \alpha \le \frac{2}{3}$): We take $\mu^{(4)} = \frac{1}{4{\mathsf {INR}}} \max \{2^{-2 C_{FB}},\frac{{\mathsf {INR}}^3}{{\mathsf {SNR}}^2}\},$ $\mu^{(6)} = \mu^{(3)} = \frac{1}{3\mathsf {INR}}-\frac{1}{4{\mathsf {INR}}} \max \{2^{-2 C_{FB}},\frac{{\mathsf {INR}}^3}{{\mathsf {SNR}}^2}\},$ $\mu^{(1)} = 1 - \mu^{(2:4)},$ and $\mu^{(5)} = 1 - \mu^{(2)}- \mu^{(6)}$ in Theorem \ref{thm_gauss2}.

Case 3 ($2 \le \alpha$): We take $\mu^{(2)} = \frac{{\mathsf{SNR}}}{2\mathsf{INR}}\min \{2^{2 C_{FB}},\frac{{\mathsf {INR}}}{{\mathsf {SNR}}^{2}}\},$ and $\mu^{(1)}= \mu^{(3)} = 1-\mu^{(2)}$ in Theorem \ref{thm_gauss3}.

The rest of the proof follows by simple manipulations of the gap, and is thus omitted. The reader can see the detailed steps in \cite{arxivver}.
\end{proof}


\begin{remark}
For the special cases of no feedback and infinite feedback, $R_{sym}^u$ in \eqref{conj} becomes the true symmetric upper bounds given in \cite{Etkin} and \cite{Mohajer}, respectively. Furthermore, the achievability schemes in \cite{Ordentlich} and \cite{Mohajer} achieve symmetric rates within constant gaps of $9+\log(K^2)$ and $\frac{1}{2}\log(16K^4(K+1))+\frac{K-1}{2}$ bits to the corresponding upper bounds, for no feedback and infinite feedback, respectively. Although these gaps are tighter, they are only for the two extreme cases.
\end{remark}

\subsubsection{Numerical Results}

We now provide numerical results on symmetric rate of the $K$-user symmetric Gaussian IC with limited feedback. In Fig. \ref{fig:subfigureExample1}, we consider three different values of $\alpha$ corresponding to the three interference regions - very weak, weak and strong interferences, and plot the symmetric rate as a function of ${\mathsf {SNR}}$ for $K=3$. It is seen that the achievable symmetric rate increases with the feedback capacity.

\begin{figure}[htbp]
\centering
\subfigure[Very weak interference with $\alpha=\frac{1}{4}$.]{
	\includegraphics[width=8.5cm]{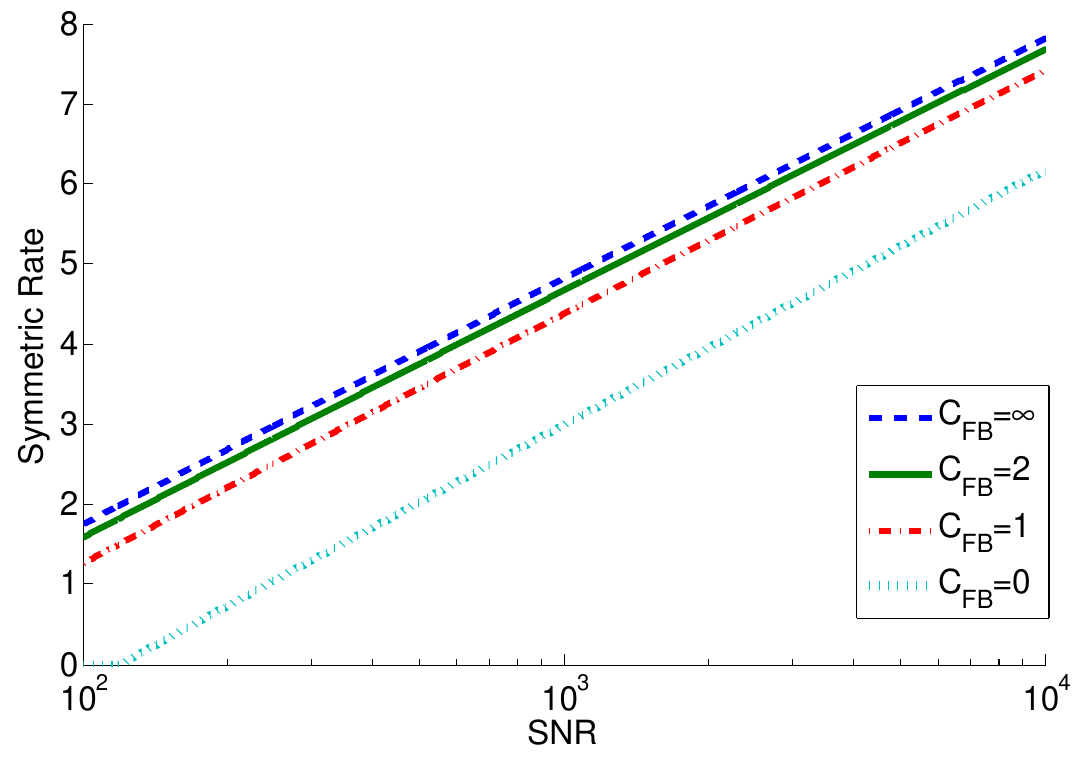}
    \label{fig:subfig1}
}
\subfigure[Weak interference with $\alpha=\frac{7}{12}$.]{
	\includegraphics[width=8.5cm]{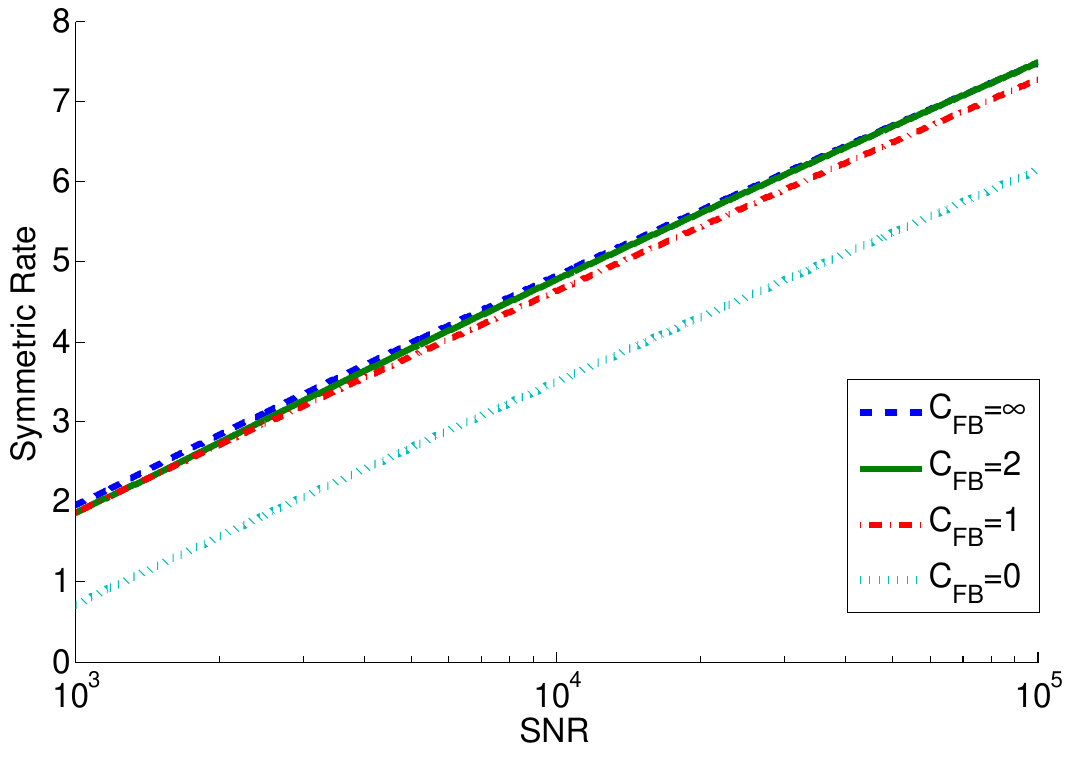}
    \label{fig:subfig2}
}
\subfigure[Strong interference with $\alpha=\frac{5}{2}$.]{
	\includegraphics[width=8.5cm]{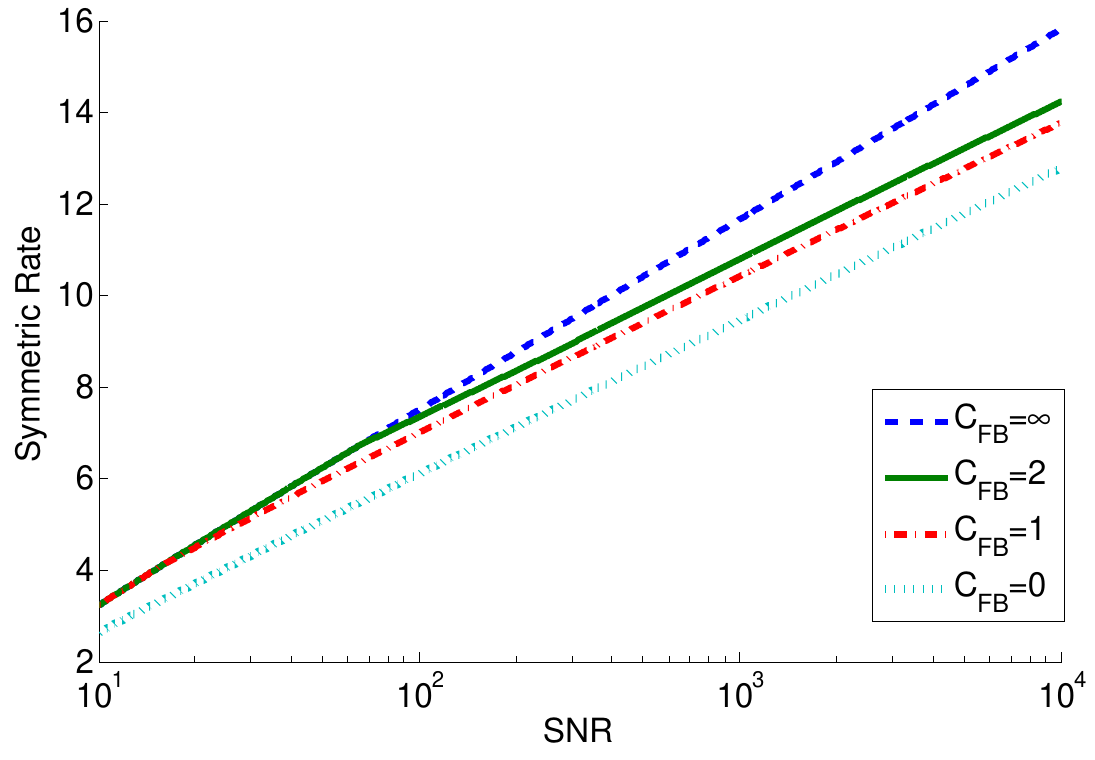}
    \label{fig:subfig3}
}
\caption[Optional caption for list of figures]{Achievable symmetric rate as a function of ${\mathsf {SNR}}$ for $K=3$.}
\label{fig:subfigureExample1}
\end{figure}


We next consider the special case of no feedback and compare the achievable rate of our scheme to that of the scheme in \cite{Ordentlich}. We let  $C_{FB}=0$, and consider some values of $\alpha$ corresponding to the different interference regimes. Note that in this case, the conjectured upper bound in \eqref{distance} becomes the upper bound in \cite{Etkin}. The achievable rate of our proposed scheme and that of the scheme in \cite{Ordentlich} as well as the upper bound, are plotted in Fig. \ref{fig:fignofeed} for $K=3$. We note that the proposed achievable symmetric rate is better than that in \cite{Ordentlich} for the parameters considered in weak and strong interference regimes. Also although for the case of very weak interference the achievable rate in \cite{Ordentlich} is higher, the slope of our scheme is higher. We can compare the constant gaps between the upper and lower bounds given in Theorem 1 of \cite{Ordentlich} for $C_{FB}=0$ and those given in Appendix \ref{apdx_gap} of this paper for general $C_{FB}$, for the parameters of Fig. \ref{fig:fignofeed}: for the very weak, weak, and strong interference regimes, the gaps of \cite{Ordentlich} are 3 bits, 11 bits, and 2 bits, respectively; for our scheme with $C_{FB}=0$, the gaps are 5.02 bits, 13.7 bits, and 4.45 bits, respectively, according to \eqref{weakbound1}, \eqref{medbound1}, and \eqref{strongbound1}, respectively.

\begin{figure}[htbp]
\centering
\subfigure[Very weak interference with $\alpha=\frac{1}{4}$.]{
	\includegraphics[width=8.5cm]{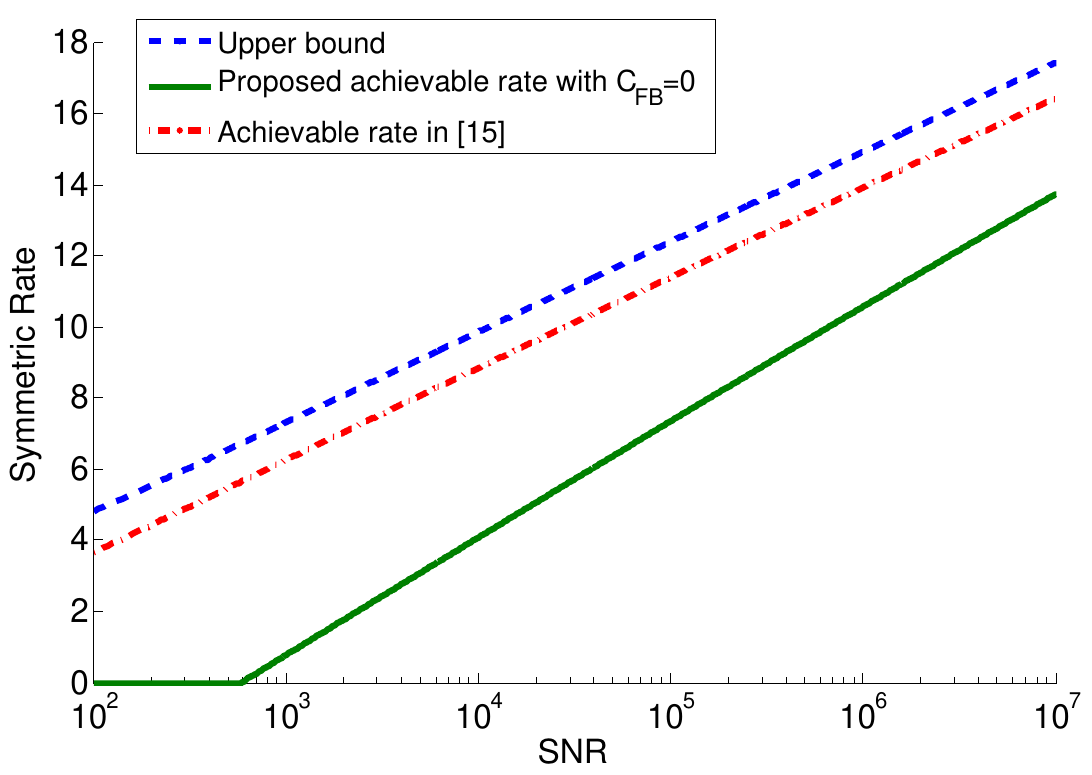}
    \label{fig:subfig1}
}
\subfigure[Weak interference with $\alpha=\frac{7}{12}$.]{
	\includegraphics[width=8.5cm]{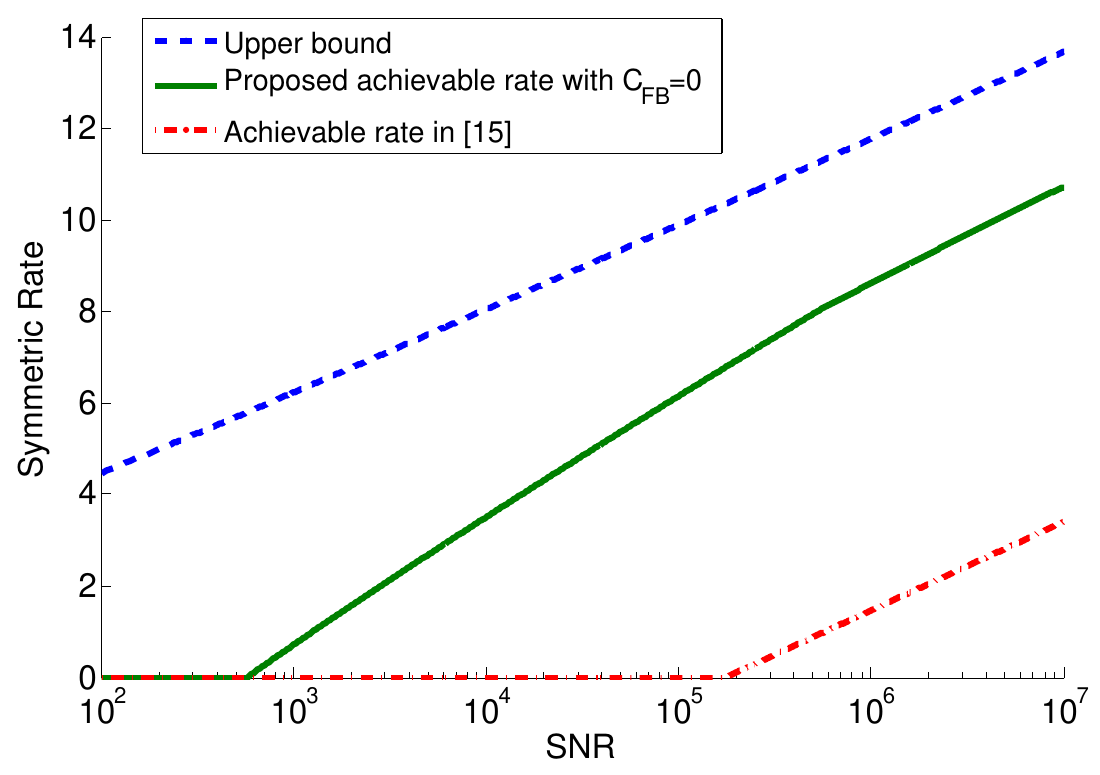}
    \label{fig:subfig2}
}
\subfigure[Strong interference with $\alpha=\frac{5}{2}$.]{
	\includegraphics[width=8.5cm]{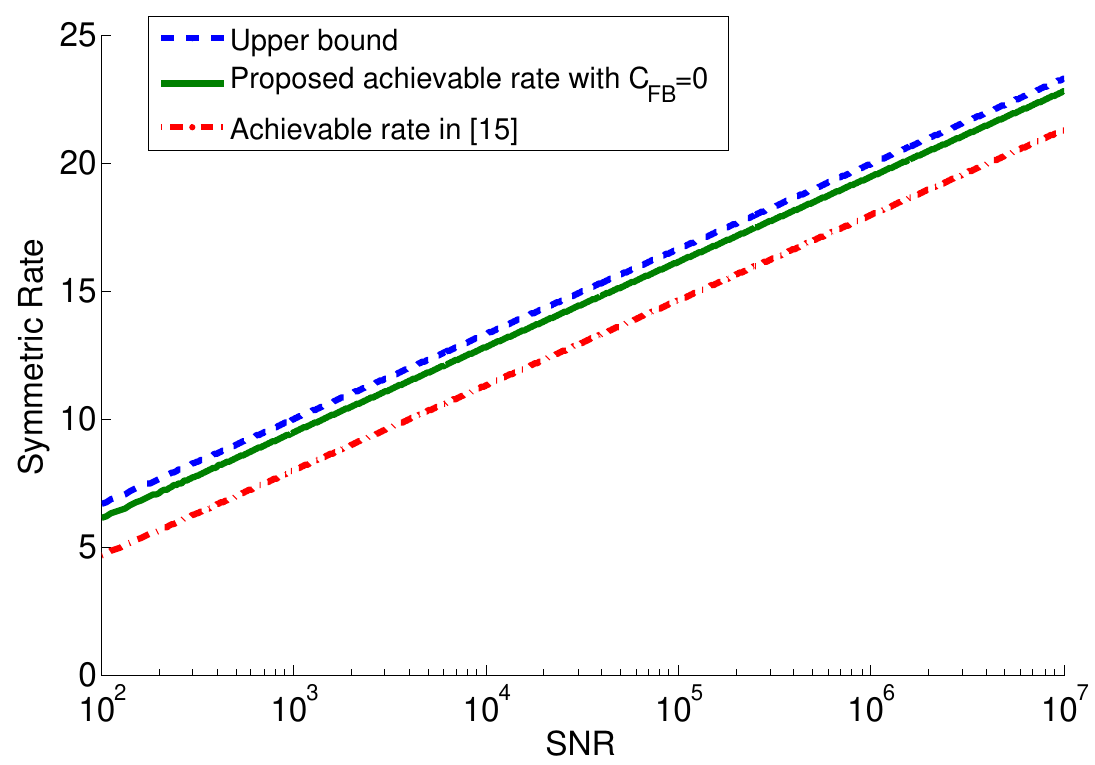}
    \label{fig:subfig3}
}
\caption[Optional caption for list of figures]{Comparison of our results with that in \cite{Ordentlich} and \cite{Etkin} for the case of no feedback and $K=3$.}
\label{fig:fignofeed}
\end{figure}



We next consider the special case of infinite feedback. In this case, the conjectured upper bound in \eqref{conj} becomes the upper bound in \cite{Mohajer}. In Fig. \ref{fig:figinffeed}, we compare the achievable rate of the proposed scheme when $C_{FB}=\infty$ to that of the scheme in \cite{Mohajer} for some values of $\alpha$ corresponding to the different interference regimes for $K=3$. We note that the proposed achievable symmetric rate is better than the achievable rate in \cite{Mohajer} for the parameters considered in strong and very weak interference regimes. For the weak interference regime, our achievability is better for high $\mathsf{SNR}$ as compared to that in \cite{Mohajer} and the slope of our scheme is higher. We can also compare the constant rate gaps for $C_{FB}=\infty$ in \cite{Mohajer} and our constant gaps for general $C_{FB}$. In particular, according to the proof of Theorem 1 (Section V and Section VI) of \cite{Mohajer}, for the very weak, weak, and strong interference regimes, the gaps are 7.17 bits, 7.17 bits, and 4.38 bits, respectively; and our corresponding constant gaps are 9.84 bits, 15.49 bits, and 7.62 bits, respectively, according to \eqref{weakbound2}, \eqref{medbound2}, and \eqref{strongbound2},  respectively.

\begin{figure}[htbp]
\centering
\subfigure[Very weak interference with $\alpha=\frac{1}{4}$.]{
	\includegraphics[width=8.5cm]{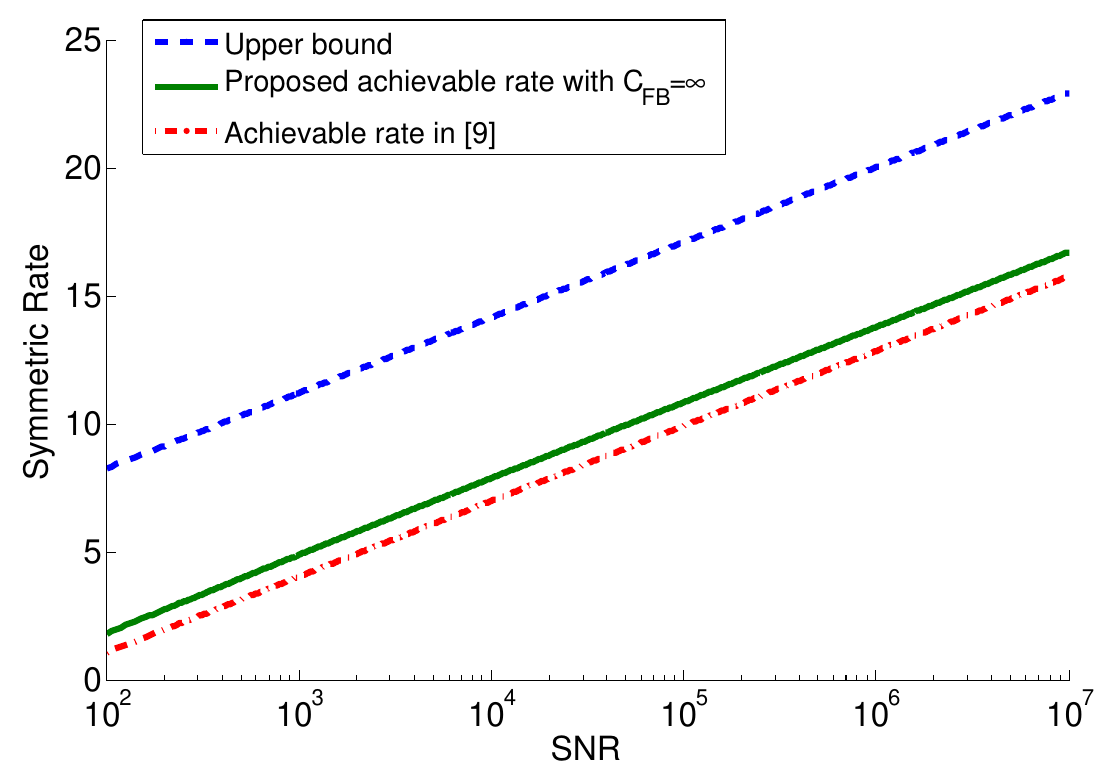}
    \label{fig:subfig1}
}
\subfigure[Weak interference with $\alpha=\frac{7}{12}$.]{
	\includegraphics[width=8.5cm]{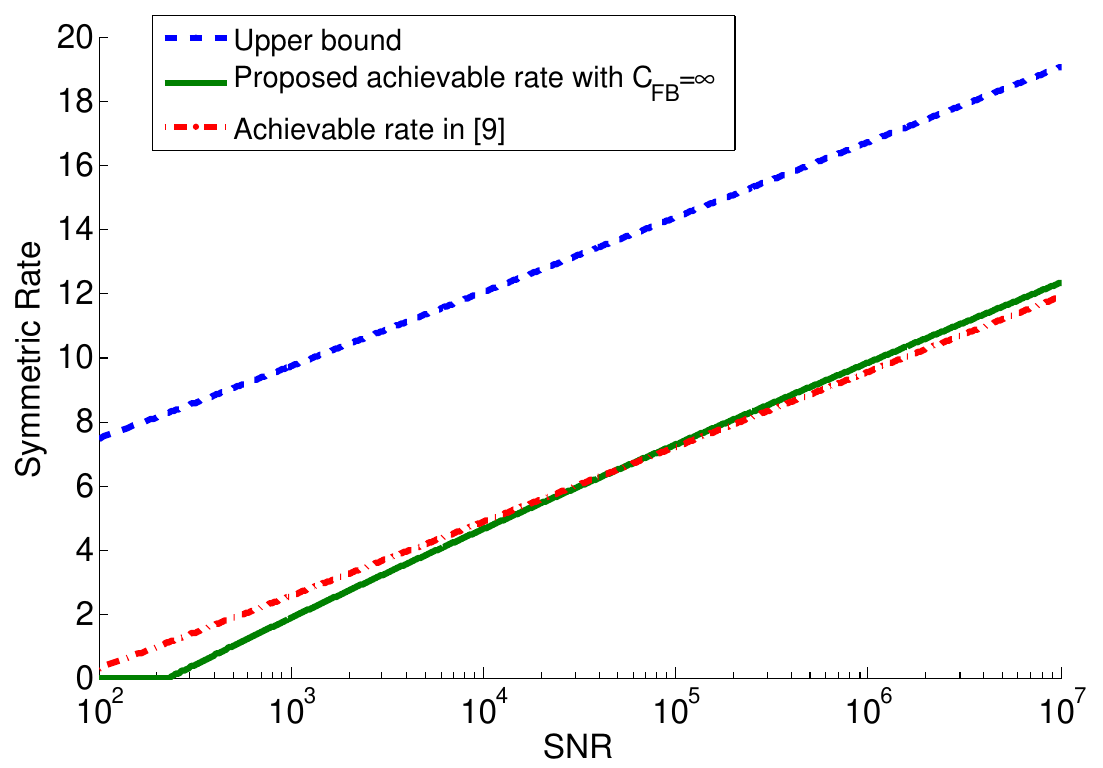}
    \label{fig:subfig2}
}
\subfigure[Strong interference with $\alpha=\frac{5}{2}$.]{
	\includegraphics[width=8.5cm]{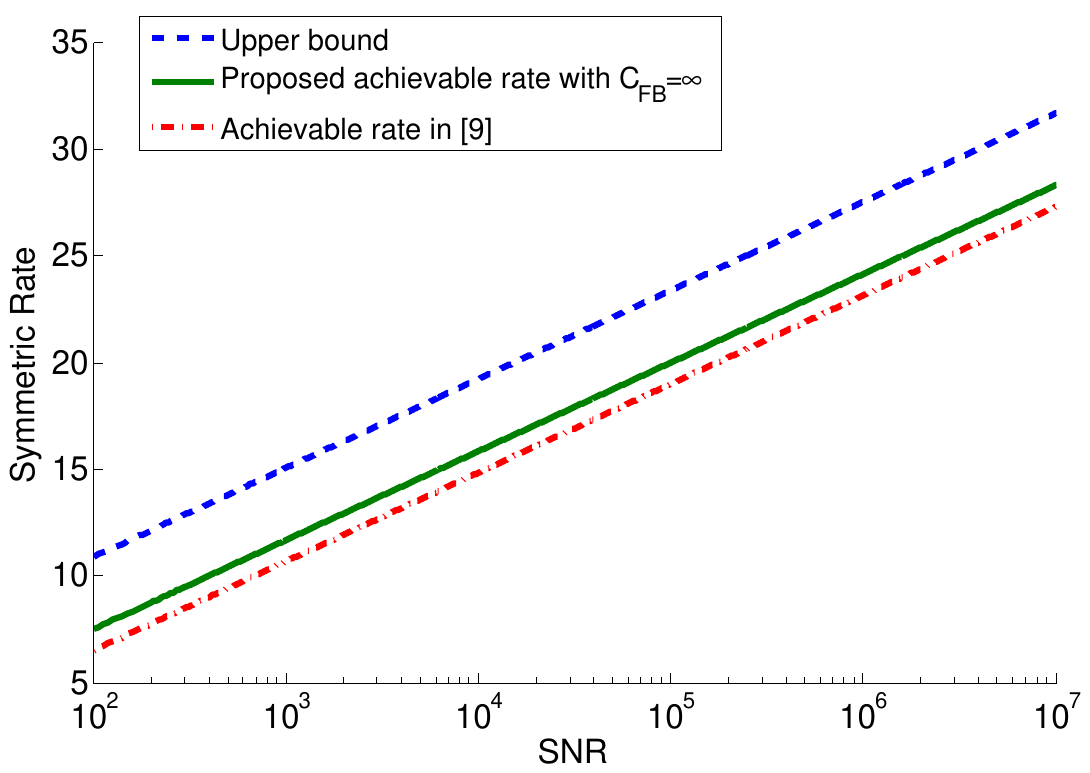}
    \label{fig:subfig3}
}
\caption[Optional caption for list of figures]{Comparison of our results with that in \cite{Mohajer} for the case of infinite feedback and $K=3$.}
\label{fig:figinffeed}
\end{figure}

Finally we consider the special case of two-user IC with limited feedback. We set $K=2$, $C_{FB}=1$. In Fig. \ref{fig:subfigureExample5}, we compare our achievable symmetric rate with that in \cite{Vahid} in different interference regimes. In this case, the conjectured upper bound in \eqref{conj} becomes the true upper bound in \cite{Vahid}. It is seen that our rate is better in the strong interference regime. And for the other two regions, our scheme has higher slopes and are better at high ${\mathsf {SNR}}$. We can also compare the constant gaps between the upper and lower bounds given in Appendix D of \cite{Vahid} for $K=2$ and those given in Appendix \ref{apdx_gap} of this paper for general $K$, for the parameters of Fig. \ref{fig:subfigureExample5}: for the very weak and strong interference regimes, the gaps of \cite{Vahid} are 9.6 bits and 7 bits, respectively; and our corresponding gaps are 5.27 bits and 4.6 bits, respectively. Hence our bounds are tighter in these two regimes. For the weak interference regime, i.e., $1/2<\alpha<2/3$ as will be noted in Section \ref{xxlir}, the achievable rate in \cite{Vahid} is actually not within a constant gap to the upper bound; whereas our proposed achievability scheme achieves a symmetric rate that is within 21.085 bits to the upper bound.


\begin{figure}[htbp]
\centering
\subfigure[Very weak interference with $\alpha=\frac{1}{4}$.]{
	\includegraphics[width=8.5cm]{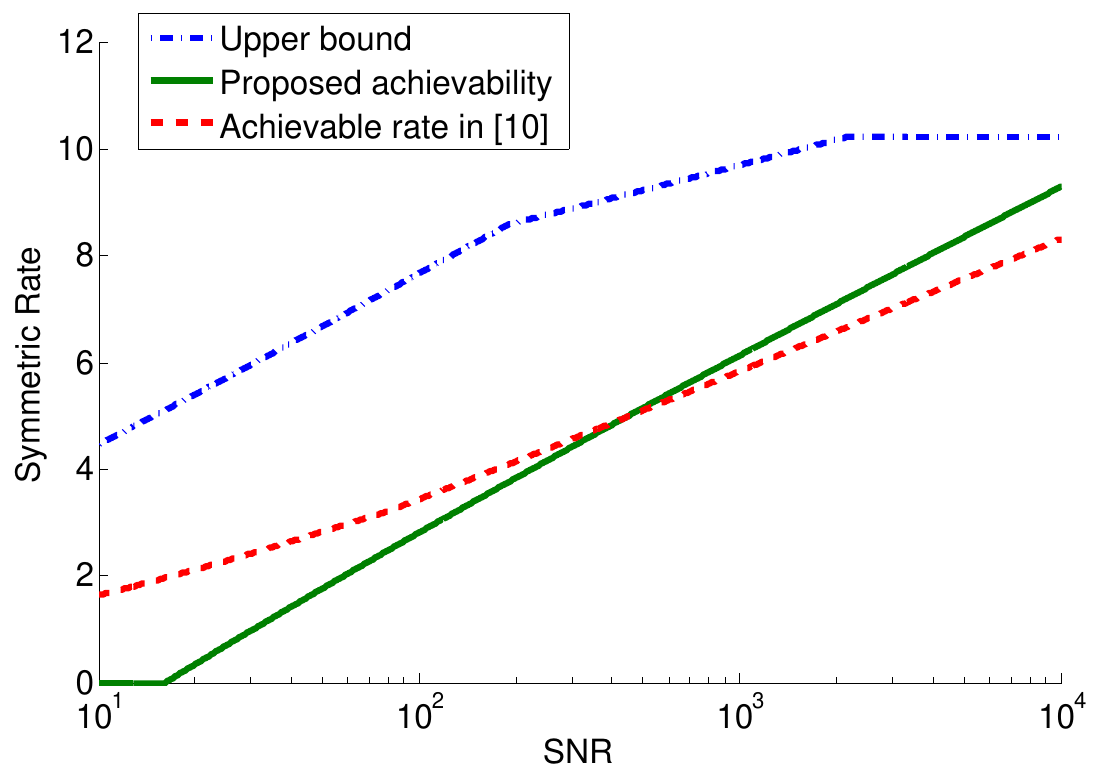}
    \label{fig:subfig1}
}
\subfigure[Weak interference with $\alpha=\frac{7}{12}$.]{
	\includegraphics[width=8.5cm]{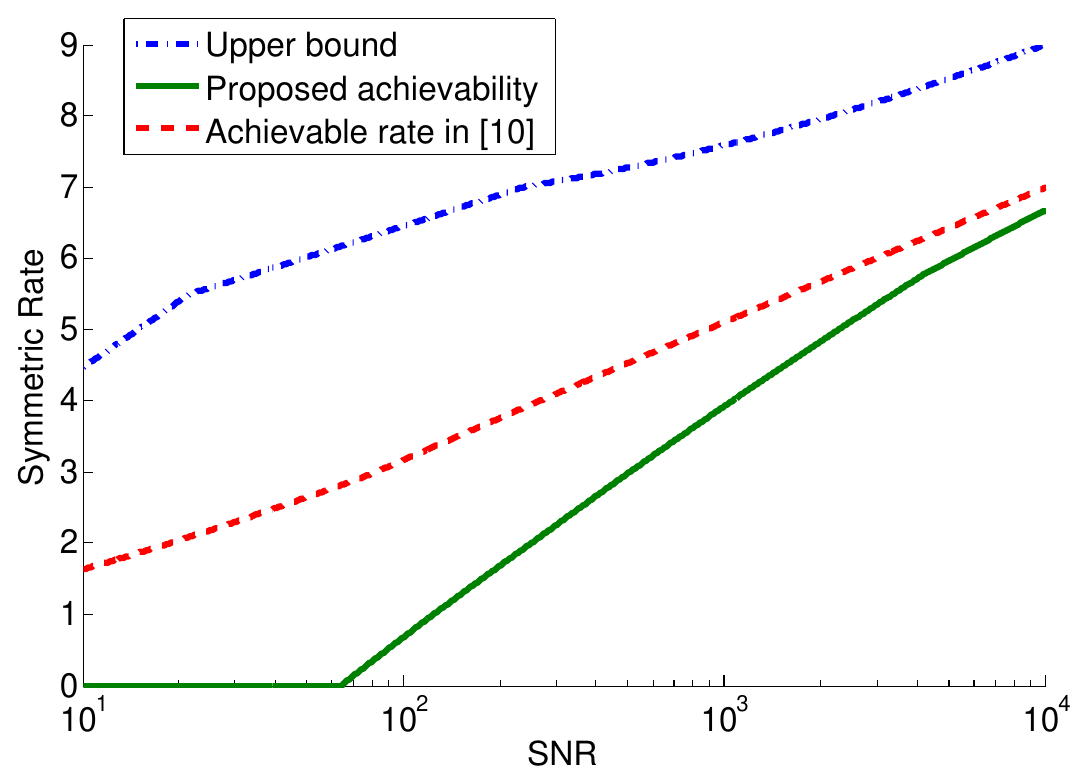}
    \label{fig:subfig2}
}
\subfigure[Strong interference with $\alpha=\frac{5}{2}$.]{
	\includegraphics[width=8.5cm]{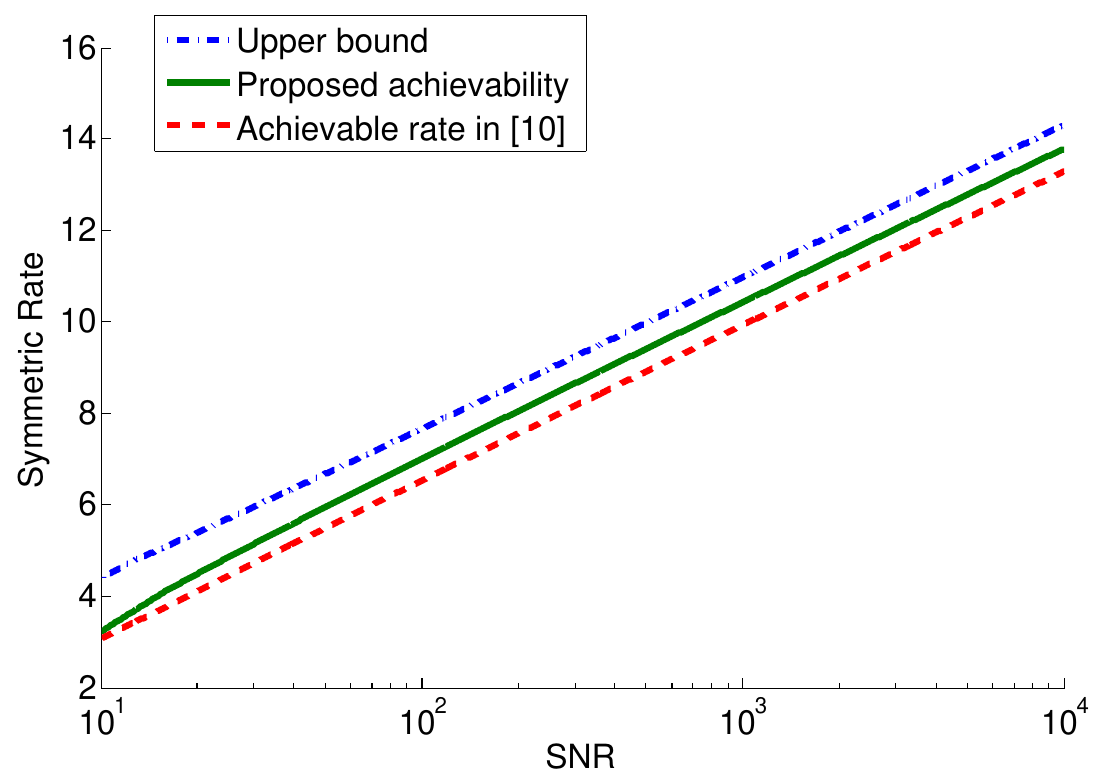}
    \label{fig:subfig2}
}
\caption[Optional caption for list of figures]{Comparison of the proposed achievability scheme with that in \cite{Vahid} for the two-user case, $K=2$, $C_{FB}=1$.}
\label{fig:subfigureExample5}
\end{figure}

\subsubsection{Achievable Symmetric ${\mathsf {GDoF}}$}

The symmetric ${\mathsf {GDoF}}$ characterize the ratio of the symmetric capacity to $\log \mathsf{SNR}$ as $\mathsf{SNR}$ goes to infinity, i.e., ${\mathsf {GDoF}}=\lim_{{\mathsf{SNR}}\to\infty}\frac{C_{sym}}{\log{\mathsf{SNR}}}$. Recall that ${\alpha}={\frac{{\log \mathsf{INR}}}{{\log \mathsf{SNR}}}\ }$ and ${\beta}={\frac{C_{FB}}{{\log \mathsf{SNR}}}}$. We have the following result.

\begin{theorem}\label{thm_gdof}
The symmetric ${\mathsf {GDoF}}$ of a  $K$-user symmetric Gaussian IC with rate-limited feedback satisfies
\begin{eqnarray}\label{gdof}
{\mathsf {GDoF}}_{sym}\ge\left\{ \begin{array}{ll}
\min\{1-\alpha+\beta,1-\frac{\alpha}{2}\},& \text{ if }0\le \alpha \le \frac{1}{2}, \\
\min\{\alpha+\beta,1-\frac{\alpha}{2}\},& \text{ if }\frac{1}{2}\le \alpha \le \frac{2}{3}, \\
1-\frac{\alpha}{2},& \text{ if }\frac{2}{3}\le \alpha < 1, \\
\text{not well defined},& \text{ if } \ \ \ \ \ \ \alpha=1, \\
\frac{\alpha}{2},& \text{ if } 1 < \alpha \le 2, \\
\min\{1+\beta,\frac{\alpha}{2}\},& \text{ if } 2 \le \alpha.
\end{array}
\right.
\end{eqnarray}
\end{theorem}

\begin{proof}
Since the achievable symmetric rate is within a constant gap to $R^u_{sym}$ in \eqref{conj}, we can write ${\mathsf {GDoF}}_{sym}\ge \frac{R_{sym}}{\log{\mathsf {SNR}} }= \lim_{{\mathsf{SNR}}\to\infty}\frac{R^u_{sym}}{\log{\mathsf {SNR}} }=\min\{{\mathsf {GDoF}}_{sym,\infty},{\mathsf {GDoF}}_{sym,0}+\beta\}$ where ${\mathsf {GDoF}}_{sym,0}$ and ${\mathsf {GDoF}}_{sym,\infty}$ are given in Theorem 3.1 of \cite{Jafar3} and Theorem 1 of \cite{Mohajer}, respectively.
\end{proof}

We note that if we normalize \eqref{r-d} by $n$ and use the definitions of $\alpha = m/n$, $\beta=p/n$, then we obtain \eqref{gdof}, except for $\alpha=1$. There is a discussion on $\alpha=1$ in \cite{Mohajer}. Hence Fig. \ref{fig:Example0d} describes the achievable symmetric ${\mathsf {GDoF}}$ of a $K$-user symmetric Gaussian IC as well.

\subsubsection{Comparison to literature}\label{xxlir}

\cite{Jafar3} and \cite{Mohajer} considered the cases of no feedback and unlimited feedback of $K$-user Gaussian IC, respectively. Also, the impact of rate-limited feedback is introduced and studied for a two-user Gaussian IC in \cite{Vahid}.

Our achievability scheme is different from there in the literature.  Consider the achievability scheme for the two-user symmetric Gaussian IC in \cite{Vahid} for the case of $1/2<\alpha<2/3$. We set the feedback capacity as $C_{FB} = \log\left(\frac{{\mathsf {SNR}}^2}{{\mathsf {INR}}^3}-1\right)$. In this case, the ${\mathsf {GDoF}}$s corresponding to the six terms in Eq.(55) in \cite{Vahid} under the power allocation given by Eq. (84) in \cite{Vahid} are $1-\alpha,0,2\alpha-1,1-\alpha,0,$ and $2\alpha-1$, respectively, with a sum of $2\alpha$. However, the sum ${\mathsf {GDoF}}$ of the achievability scheme which is the sum of these six terms, is claimed in Eq. (87) of \cite{Vahid} to be $2-\alpha=2\alpha+(2-3\alpha)>2\alpha$ which is incorrect. Since in this range of $\alpha$, the upper bound on sum rate satisfies $\lim_{{\mathsf{SNR}}\to\infty}\frac{2 R^u_{sym}}{\log{\mathsf {SNR}} }= 2-\alpha$, the gap between the upper and lower bounds for high ${\mathsf {SNR}}$ is $(2-3\alpha)\log{\mathsf {SNR}}+o(\log{\mathsf {SNR}})$, i.e., it is unbounded. Our proposed achievability scheme when specialized to $K=2$, results in a symmetric rate that is within a constant of 21.085 bits to the symmetric rate upper bound, according to Theorem \ref{thm_gauss}.

Also \cite{Mohajer} treats only the case of perfect feedback, i.e., $C_{FB}=\infty$, whereas we treat the general case of arbitrary $C_{FB}$.

Our proposed conjectured upper bound is the best known upper bound for the special cases; for the $K$-user IC without feedback \cite{Ordentlich}, $K$-user IC with infinite feedback \cite{Mohajer}, and $K=2$ with general $C_{FB}$ \cite{Vahid}. However, it remains open for general $K$ and $C_{FB}$.

\section{Conclusions}
We have developed achievability schemes for symmetric $K$-user interference channels with rate-limited feedback, for both the linear deterministic model, and the Gaussian model. For the deterministic model, the achievable symmetric rate is the minimum of the symmetric capacity with infinite feedback, and the sum of the symmetric capacity without feedback and the amount of symmetric feedback. And for the Gaussian model, the achievable rate is  within a constant gap to the minimum of the  symmetric capacity with infinite feedback, and the sum of the symmetric capacity without feedback and the amount of symmetric feedback. For the Gaussian model, the proposed achievability scheme employs lattice codes to perform  Han-Kobayashi message splitting, interference-decoding, and decode-and-forward. Further, the achievable generalized degrees of freedom (${\mathsf {GDoF}}$) is characterized with rate-limited feedback. It is shown that the per-user ${\mathsf {GDoF}}$ does not depend on the number of users, so that it is the same as that of the two-user interference channel with rate-limited feedback.

We conjecture that the minimum of the upper bound of the symmetric capacity with infinite feedback, and the sum of the upper bound of the symmetric capacity without feedback and the amount of symmetric feedback is an upper bound for the symmetric capacity of the Gaussian IC with rate-limited feedback for any number of users $K$. This conjecture has been shown to hold for the $K$-user IC without feedback in \cite{Ordentlich}, the $K$-user IC with infinite feedback in \cite{Mohajer}, and $K=2$ in \cite{Vahid}. However, it remains open for general $K$ and $C_{FB}$. The achievability for $\alpha=1$ in \cite{Ordentlich,Mohajer} assumes that channel gains are outside an outage set. Investigating whether this outage set shrinks with feedback is an interesting open problem.

\section*{Acknowledgement}
We are grateful to the anonymous referees and the Associate Editor for their valuable comments, which have helped us improve the presentation of this paper substantially.

\begin{appendices}
\section{Proof of  Theorem \ref{thm_det}} \label{apdx_innerdet}

In this section, we prove Theorem \ref{thm_det} by breaking the result into three regimes. We denote that $\overline{a_{i,j}}\triangleq \sum_{k=1,k\neq i}^{K}{a_{k,j}}$

\begin{lemma}\label{lemma-det1}
For the $K$-user linear deterministic IC, a symmetric rate of $n \min\{1-\alpha+\beta,1-\frac{\alpha}{2}\}$ is achievable for $0 \le \alpha \le \frac{1}{2}$. \end{lemma}

\begin{proof} Define $l\triangleq {(m-2p)}^+$. For the $i^{\text{th}}$ transmitter, $i \in \{1,...,K\}$, we transmit $a_{i,1},...,a_{i,2n-m-l}$ in two transmission slots.

{\bf First Round:}

1. Transmission: In the first round, the $i^{\text{th}}$ transmitter sends $a_{i,1},...,a_{i,n-l}$ on the highest $n-l$ transmission levels, respectively, and nothing on the lowest $l$ transmission levels.

2. Reception: Since $0 \le \alpha \le \frac{1}{2}$, the $i^{\text{th}}$ receiver receives $a_{i,1},...,a_{i,n-m}$ on the highest $n-m$ reception levels, respectively, and $a_{i,n-m+1}+\overline{a_{i,1}},...,a_{i,n-l}+\overline{a_{i,m-l}}$ on the next $m-l$ levels, respectively, and throws away whatever it receives on the last $l$ levels.

{\bf Feedback:}

Receiver $i$ sends back $a_{i,n-m+1}+\overline{a_{i,1}},...,a_{i,n-l}+\overline{a_{i,m-l}}$ over the feedback channel to transmitter $i$ ($m-l$ levels). Since $0\le m-l\le 2p$, the feedback rate is $p$ levels per channel use. With this feedback, transmitter $i$ decodes $\overline{a_{i,1}},...,\overline{a_{i,m-l}}$. Since the feedback does not increase the achievable rate in the statement of the Theorem beyond $p=m/2$, we only use $m/2$ levels of feedback if $p>m/2$.

{\bf Second Round:}

1. Transmission: In the second round, the $i^{\text{th}}$ transmitter sends $\overline{a_{i,1}},...,\overline{a_{i,m-l}}$ on the highest $m-l$ transmission levels, respectively, nothing on the next lower $l$ levels, and new levels of $a_{i,n-l+1},...,a_{i,2n-m-l}$ on the last $n-m$ levels, respectively.

2. Reception: The $i^{\text{th}}$ receiver receives $\overline{a_{i,1}},...,\overline{a_{i,m-l}}$ on the highest $m-l$ levels, nothing on the next $l$ levels, $a_{i,n-l+1},...,a_{i,2n-2m-l}$ on the next $n-2m$ levels, $a_{i,2n-2m-l+1}+(K-2)\overline{a_{i,1}}+(K-1)a_{i,1},...,a_{i,2n-m-2l}+(K-2)\overline{a_{i,m-l}}+(K-1)a_{i,m-l}$ on the next $m-l$ levels,  and $a_{i,2n-m-2l+1},...,a_{i,2n-m-l}$ on the lowest $l$ levels.

{\bf Decoding:}

Decoding by the $i^{\text{th}}$ receiver, $i \in \{1,...,K\}$, is performed as follows. First, $a_{i,1},...,a_{i,n-m}$ are decoded from the highest $n-m$ levels of the  first reception. Then, $\overline{a_{i,1}},...,\overline{a_{i,m-l}}$ are decoded from the highest $m-l$ levels of the second reception. Then, having $\overline{a_{i,1}},...,\overline{a_{i,m-l}}$, the receiver decodes $a_{i,n-m+1},...,a_{i,n-l}$ from $a_{i,n-m+1}+\overline{a_{i,1}},...,a_{i,n-l}+\overline{a_{i,m-l}}$ on the next $m-l$ levels of the  first reception. Then, the receiver decodes $a_{i,n-l+1},...,a_{i,2n-2m-l}$ from the ${(m+1)}^{\text{th}}$ to ${(n-m)}^{\text{th}}$ highest levels of the second reception, respectively. Then, having  $a_{i,1},...,a_{i,m-l}$, and $\overline{a_{i,1}},...,\overline{a_{i,m-l}}$, the receiver decodes $a_{i,2n-2m-l+1},...,a_{i,2n-m-2l}$ from $a_{i,2n-2m-l+1}+(K-2)\overline{a_{i,1}}+(K-1)a_{i,1},...,a_{i,2n-m-2l}+(K-2)\overline{a_{i,m-l}}+(K-1)a_{i,m-l}$ on the next $m-l$ lower levels of the second reception. Finally, the receiver decodes $a_{i,2n-m-2l+1},...,a_{i,2n-m-l}$ from the lowest $l$ levels of the second reception.

{\bf Rate:}

With the above strategy, each user transmits $2n-m-l$ levels in two uses of the channel which proves the lemma because
$\frac{1}{2}(2n-m-l)=\frac{1}{2}(2n-m-(m-2p)^+)=\frac{1}{2}\min\{2n-m,2n-2m+2p\}=\min\{n-\frac{1}{2}m,n-m+p\}=n\min\{1-\frac{\alpha}{2},1-\alpha+\beta\}$.
\end{proof}

\begin{lemma}\label{lemma-det2}
For the $K$-user linear deterministic IC, a symmetric rate of  $n \min\{\alpha+\beta,1-\frac{\alpha}{2}\}$  is achievable for $\frac{1}{2} \le \alpha \le \frac{2}{3}$.
\end{lemma}

\begin{proof} Define $l^{'}\triangleq {(2n-3m-2p)}^+$. For the $i^{\text{th}}$ transmitter, $i \in \{1,...,K\}$, we transmit $a_{i,1},...,a_{i,2n-m-l^{'}}$ in two transmission slots.

{\bf First Round:}

1. Transmission: In the first round, the $i^{\text{th}}$ transmitter sends $a_{i,1},...,a_{i,n-m- l^{'}}$ on the highest $n-m- l^{'}$ transmission levels, nothing on the next lower $2m-n+ l^{'}$ levels, and $a_{i, n-m- l^{'} +1},...,a_{i, 2n-2m- l^{'} }$ on the lowest $n-m$ levels.

2. Reception: Since $\frac{1}{2} \le \alpha \le \frac{2}{3}$, the $i^{\text{th}}$ receiver receives $a_{i,1},...,a_{i,n-m- l^{'}}$ on the highest $n-m- l^{'}$ reception levels, nothing on the next lower $l^{'}$ levels, $\overline{a_{i,1}},...,\overline{a_{i,2m-n}}$ on the next lower $2m-n$ levels, $a_{i, n-m- l^{'} +1}+\overline{a_{i,2m-n+1}},...,a_{i, 3n-4m-2 l^{'} }+\overline{a_{i, n-m- l^{'}}}$ on the next $2n-3m- l^{'} $ levels,  and $a_{i, 3n-4m-2 l^{'} +1},...,a_{i, 2n-2m- l^{'} }$ on the lowest $2m-n+l^{'} $ levels.

{\bf Feedback:}

Receiver $i$ sends back $a_{i, n-m- l^{'} +1}+\overline{a_{i,2m-n+1}},...,a_{i, 2n-3m-2 l^{'} }+\overline{a_{i, n-m- l^{'}}}$ over the feedback channel to transmitter $i$ ($2n-3m- l^{'} $ levels).  Since $0\le 2n-3m-l^{'} \le 2p$, the feedback rate is $p$ levels per channel use. With this feedback, transmitter $i$ decodes $\overline{a_{i,2m-n+1}},...,\overline{a_{i, n-m- l^{'}}}$.

{\bf Second Round:}

1. Transmission: In the second round, the $i^{\text{th}}$ transmitter sends the new signals $a_{i, 2n-2m- l^{'} +1},...,a_{i, n-l^{'}}$ on the highest $2m-n$ transmission levels,  $\overline{a_{i,2m-n+1}},...,\overline{a_{i, n-m- l^{'}}}$ on the next $2n-3m- l^{'} $ levels, nothing on the next lowest $2m-n+l^{'} $ levels,  and the new signals $a_{i, n-l^{'}+1},...,a_{i, 2n-m-l^{'}}$ on the lowest $n-m$ levels.

2. Reception: In this round, the $i^{\text{th}}$ receiver receives $a_{i, 2n-2m- l^{'} +1},...,a_{i, n-l^{'}}$ on the highest $2m-n$ reception levels, $\overline{a_{i,2m-n+1}},...,\overline{a_{i, n-m- l^{'}}}$ on the next $2n-3m- l^{'} $ levels,  nothing on the next lower $l^{'} $ levels, $\overline{a_{i, 2n-2m- l^{'} +1}},...,\overline{a_{i, n-l^{'}}}$ on the next lower $2m-n$ levels, $a_{i, n-l^{'}+1}+(K-2)\overline{a_{i,2m-n+1}}+(K-1) a_{i,2m-n+1},...,a_{i, 3n-3m-2l^{'}}+(K-2)\overline{a_{i,2n-3m-l^{'} }}+(K-1)a_{i, 2n-3m-l^{'} }$ on the next lower $2n-3m-l^{'}$ levels, and $a_{i, 3n-3m-2l^{'}+1},...,a_{i, 2n-m-l^{'}}$ on the lowest $2m-n+l^{'}$ levels.

{\bf Decoding:}

Decoding by the $i^{\text{th}}$ receiver, $i \in \{1,...,K\}$, is performed as follows. First, $a_{i,1},...,a_{i,n-m- l^{'}}$ are decoded from the highest $n-m- l^{'}$ levels of the  first reception. Then, $a_{i, 3n-4m-2 l^{'} +1},...,a_{i, 2n-2m- l^{'} }$ are decoded from the lowest $2m-n+l^{'} $ levels of the first reception. Further, $a_{i, 2n-2m- l^{'} +1},...,a_{i, n-l^{'}}$ are decoded from the highest $2m-n$ levels of the second reception,  and $\overline{a_{i,2m-n+1}},...,\overline{a_{i, n-m- l^{'}}}$ are decoded from the next $2n-3m- l^{'} $ levels of the second reception. Moreover,  $a_{i, 3n-3m-2l^{'}+1},...,a_{i, 2n-m-l^{'}}$ are decoded from the lowest $2m-n+l^{'}$ levels of the first transmission, respectively.

Then, having $\overline{a_{i,2m-n+1}},...,\overline{a_{i, n-m- l^{'}}}$, the receiver decodes  $a_{i, n-m- l^{'} +1},...,a_{i, 2n-3m-2 l^{'}}$ from  $a_{i, n-m- l^{'} +1}+\overline{a_{i,2m-n+1}},...,a_{i, 2n-3m-2 l^{'} }+\overline{a_{i, n-m- l^{'}}}$ in the first reception. Finally, having  $\overline{a_{i,2m-n+1}},...,\overline{a_{i,2n-3m-l^{'}}}$, and $a_{i,2m-n+1},...,a_{i, 2n-3m-l^{'} }$, the receiver decodes $a_{i, n-l^{'}+1},...,a_{i, 3n-3m-2l^{'}}$  from $a_{i, n-l^{'}+1}+(K-2)\overline{a_{i,2m-n+1}}+(K-1) a_{i,2m-n+1},...,a_{i, 3n-3m-2l^{'}}+(K-2)\overline{a_{i,2n-3m-l^{'} }}+(K-1)a_{i, 2n-3m-l^{'} }$ in the second reception.

{\bf Rate:}

With the above strategy, each user transmits $2n-m-l^{'}$ levels in two uses of the channel which proves the lemma because
$\frac{1}{2}(2n-m-l^{'})=\frac{1}{2}(2n-m-(2n-3m-2p)^+)=\frac{1}{2}\min\{2n-m,2m+2p\}=\min\{n-\frac{1}{2}m,m+p\}=n\min\{1-\frac{\alpha}{2},\alpha+\beta\}$.
\end{proof}

\begin{lemma}\label{lemma-det3}
For the $K$-user linear deterministic IC, a symmetric rate of $n \min\{1+\beta,\frac{\alpha}{2}\}$  is achievable for $ \alpha \ge 2$.
\end{lemma}

\begin{proof} Define $l^{''}\triangleq {(m-2n-2p)}^+$. For the $i^{\text{th}}$ transmitter, $i \in \{1,...,K\}$, we transmit $a_{i,1},...,a_{i,m-l^{''}}$ in two transmission slots.

{\bf First Round:}

1.Transmission: In the first round, the $i^{\text{th}}$ transmitter sends $a_{i,1},...,a_{i,m - n - l^{''}}$ on the highest $ m-n- l^{''}$ transmission levels, respectively, and nothing on the lower $n+ l^{''}$ levels.

2. Reception: Since $ \alpha \ge 2$, the $i^{\text{th}}$ receiver receives $\overline{a_{i,1}},...,\overline{a_{i,m-n-l^{''}}}$ on the highest $ m-l^{''}$ reception levels, nothing on the next lower $ l^{''}$ levels, and $a_{i,1},...,a_{i,n}$  on the lowest $n $ levels.

{\bf Feedback:}

Receiver $i$ sends back $\overline{a_{i,n+1}},...,\overline{a_{i, m-n-l^{''}}}$ over the feedback channel to the $i^{\text{th}}$ transmitter ($m-2n-l^{''}$ levels). Since $0\le m-2n-l^{''} \le 2p$, the feedback rate is $p$ levels per channel use.

{\bf Second Round:}

1.Transmission: In the second round, the $i^{\text{th}}$ transmitter sends new levels $a_{i, m - n - l^{''}+1},...,a_{i,m - l^{''}}$  on the highest $n$ transmission levels,  $\overline{a_{i,n+1}},...,\overline{a_{i, m-n-l^{''}}}$ on the next $m-2n- l^{''} $ levels, and nothing on the lower $n+ l^{''}$ levels.

2. Reception: The $i^{\text{th}}$ receiver receives $\overline{a_{i, m - n - l^{''}+1}},...,\overline{a_{i,m - l^{''}}}$  on the highest $n$ reception levels,  $(K-1) a_{i,n+1}+(K-2)\overline{a_{i,n+1}},...,(K-1) a_{i, m-n-l^{''}}+(K-2)\overline{a_{i, m-n-l^{''}}}$ on the next $m-2n- l^{''} $ levels, nothing on the next lower $l^{''}$ levels, and $a_{i, m - n - l^{''}+1},...,a_{i,m - l^{''}}$  on the lowest $n$ levels.

{\bf Decoding:}

Decoding at the $i^{\text{th}}$ receiver, $i \in \{1,...,K\}$ is performed as follows. First, $a_{i,1},...,a_{i,n}$ are decoded from the lowest $n$ levels of the  first reception, $\overline{a_{i,n+1}},...,\overline{a_{i, m-n-l^{''}}}$ are decoded from the ${(n+1)}^{\text{th}}$ to ${( m-n-l^{''})}^{\text{th}}$ highest levels of the first reception,  and $a_{i,m-n- l^{''}+1},...,a_{i,m -l^{''}}$ are decoded from the lowest $n$ levels of the second reception. Then, having  $\overline{a_{i,n+1}},...,\overline{a_{i, m-n-l^{''}}}$, the receiver decodes $a_{i,n+1},...,a_{i, m-n-l^{''}}$ from $(K-1) a_{i,n+1}+(K-2)\overline{a_{i,n+1}},...,(K-1)a_{i, m-n-l^{''}}+(K-2)\overline{a_{i, m-n-l^{''}}}$ in the second reception.

{\bf Rate:}

With the above strategy, each user transmits $m- l^{''}$ levels in two uses of the channel which proves the lemma because
$\frac{1}{2}(m-l^{''})=\frac{1}{2}(m-(m-2n-2p)^+)=\frac{1}{2}\min\{m,2n+2p\}=\min\{\frac{m}{2},n+p\}=n\min\{\frac{\alpha}{2},1+\beta\}$.
\end{proof}

\section{Some Lemmas Used in Proofs of Achievability for Gaussian Channel}

\subsection{Lemmas for the proof of decodability of forward transmission} \label{sjkhgsjkdh}

In Lemmas \ref{sjkfgv}-\ref{sjkhgsjkdh2} in the following, assume an interference network with $K$ transmitters and $M$ receivers, where the discrete-time real Gaussian channel has the vector representation
\begin{equation}
{\bf y}_m = \sum_{k=1}^{K} h_{m,k} {\bf x}_k + {\bf z}_m,
\end{equation}
with ${\bf y}_m\in{\mathbb R}^T$, ${\bf x}_k\in{\mathbb R}^T$, $h_{m,k} \in{\mathbb R}$ denoting the channel output of receiver $m$, channel input of transmitter $k$ and the channel gain, respectively. The Gaussian white noise with unit variance is denoted by ${\bf z}_m\in{\mathbb R}^T$. Also, assume the power constraint ${\mathbb E} \{ {\| {\bf x}_k \|}^2 \} \le T P$ on all the transmitters, and each transmitted signal ${\bf x}_k$ is built from the lattice points ${\bf s}_k=\phi({\bf w}_k)$ and using a dither as described in Section \ref{vjhdz}.

The following lemma is taken from \cite{zhu2014asymmetric}:

\begin{lemma}\label{sjkfgv}\cite[Theorem~2]{zhu2014asymmetric}
For any given set of positive numbers $\beta_1,\dots,\beta_K$, and the lattice codes ${\mathcal C}_1, \dots , {\mathcal C}_K$ as described in Section \ref{vjhdz},  the capacity region is such that the desired functions $f_m = \sum_k a_{m,k} {\bf s}_k$, $m\in\{1,\dots,M\}$ are obtainable at destinations, with $a_{m,k}\in{\mathbb Z}$ and the set of rates  $(R_1, \dots , R_K )$ satisfying
\begin{equation}
R_k < \min_{\left\{m | m\in{\mathbb Z}, 1\le m\le M, a_{m,k}\neq 0\right\}} {\left[ \log{\left( {\| {\tilde{\bf a}_m} \|}^2 - \frac{P{( {\bf h}_m^t{\tilde{\bf a}_m} )}^2 }{1+P { \| {\bf h}_m \| }^2 } \right)}^{-1} + \log{\beta_k^2} \right] }^{+},
\end{equation}
for all $k$, where ${{\bf h}_m} \triangleq {[ h_{m,1}, \dots , h_{m,K}]}^t$, ${\tilde{\bf a}_m} \triangleq {[ \beta_1 a_{m,1}, \dots , \beta_K a_{m,K}]}^t$ and $a_{m,k} \in {\mathbb Z}$ for all $k \in \{1,\dots,K\}$.
\end{lemma}

The following lemma is also similar to \cite[Lemma~1]{nazer2012successive} with some modifications:
\begin{lemma}\label{sjh2}
The receiver can make an estimate of the real sum of codewords, $\sum_{k=1}^{J}{\bf x}_k$, with vanishing probability of error so long as the rate constraints proposed in Lemma \ref{sjkfgv} hold.
\end{lemma}
\begin{proof}
In \cite[Lemma~1]{nazer2012successive}, it is shown that if the conditions in Lemma \ref{sjkfgv} hold and we are able to derive $\sum_k a_{k} {\bf s}_k$ for the case that the rates of all messages are equal, the real sum of codewords, $\sum_k a_{k} {\bf x}_k$ can be obtained. The proof in \cite[Lemma~1]{nazer2012successive} can be easily extended to the case where the message rates are different, thus giving the result as in the statement of the lemma.
\end{proof}


Using the following lemma on the properties of lattice codes, the achievability constraints of our theorems on recovering the summation of lattices are obtained: 

\begin{lemma}\label{sjkhgsjkdh2}
Assuming $h_1 = \dots = h_J = h$, $J\leq K$ we are able to obtain $\sum_{k=1}^{J}{\bf x}_k$ with vanishing probability of error as $T\to\infty$, if the following constraints hold:
\begin{equation}
R_i \le \log\left( \frac{1}{J} + \frac{P h^2}{ P \sum_{j = J+1}^{K} h_j^2 +1} \right), \ \ \ \ \ \ \ \ i\in\{1,\dots,J\}.
\end{equation}
\end{lemma}

\begin{proof} 
In Lemma \ref{sjkfgv}, assume $\beta_1=\dots=\beta_K=1$ and $a_1=\dots=a_J=1$ and $a_{J+1}=\dots=a_K=0$. Then, the rate constraints $R_i$, $i\in\{1,\dots,J\}$, need to satisfy:
\begin{eqnarray}
R_i &\le& \log{\left( {\| {\bf a} \|}^2 - \frac{P {({\bf h}^{t}{\bf a})}^2}{ P {\| {\bf h} \|}^2 +1} \right)}^{-1} \nonumber\\
&=& \log{\left( J - \frac{P {(Jh)}^2}{ P \sum_{j = 1}^{K} h_j^2 +1} \right)}^{-1}\nonumber\\
&=& \log{\left( \frac{ J(P \sum_{j = 1}^{K} h_j^2 +1) - P {(Jh)}^2}{ P \sum_{j = 1}^{K} h_j^2 +1} \right)}^{-1}\nonumber\\
&=& \log{\left( \frac{ P \sum_{j = 1}^{K} h_j^2 +1}{ J(P \sum_{j = 1}^{K} h_j^2 +1) - P {(Jh)}^2} \right)}\nonumber\\
&=& \log\left( \frac{1}{J} + \frac{P h^2}{ P \sum_{j = J+1}^{K} h_j^2 +1} \right).
\end{eqnarray}

Therefore, we are able to derive $\sum_{k=1}^{J}{\bf s}_k$ using Lemma \ref{sjkfgv}. Then, by applying Lemma \ref{sjh2} it can be seen that the real sum of codewords, $\sum_{k=1}^{J}{\bf x}_k$, can be obtained which completes the proof.
\end{proof}

\subsection{A lemma for the proof of decodability of feedback transmission} \label{sjk}

Using the following lemma on the properties of lattice codes, the achievability constraints of our theorems on feeding back the summation of multiple lattices to the transmitters are obtained:

\begin{lemma}\label{sjk2}
Assume that each transmitter $k\in\{1,\dots,K\}$, is equipped with an encoder ${\mathcal E}_k$ of rate $R$ which maps its message into the channel input as ${{\bf x}_k}$ that is chosen from a lattice and is a discrete subgroup of ${\mathbb R}^T$ (as described in Section \ref{vjhdz}). If $R_{sum}$ is the minimum rate needed for transmitting $\sum^{K}_{k=1}{{\bf x}_k}$ with error going to zero on feedback links as the block size $T \to\infty$, then $R_{sum} \le R + \log {K}$.
\end{lemma}

\begin{proof}
Assume that each ${\bf x}_k$, $k\in\{1,\dots,K\}$, is a lattice codeword, with rate $R$. Depending on the Voronoi region which is in a $T$-dimensional space, the number of possible values for each ${\bf x}_k$ (the number of channel coding lattice points in Voronoi cell) is $|\Lambda_{f} \cap {\mathcal{V}}_{\Lambda_{c}}|$ where $\Lambda_{c}$ is the quantization lattice with channel coding lattice $\Lambda_{f}$, and ${{\mathcal{V}}_{\Lambda_{c}}}$ is the $T$-dimensional Voronoi cell of the lattice $\Lambda_{c}$. Since if all of the $K$ lattices of ${\bf x}_k$'s are along the same direction, their sum has the maximum length which is $K$ times the length of each individual ${\bf x}_k$, the number of possible values for the sum of messages $\sum_{k=1}^{K} {\bf x}_k$ is up to $|\Lambda_{f} \cap \left(K^T {\mathcal{V}}_{\Lambda_{c}}\right)|$. Therefore, if $R_{sum}$ is the rate needed to transmit the sum of ${\bf x}_k$'s, given $\frac{2^{TR_{sum}}}{2^{TR}} \le \frac{|\Lambda_{f} \cap \left(K^T {\mathcal{V}}_{\Lambda_{c}}\right)|}{|\Lambda_{f} \cap {\mathcal{V}}_{\Lambda_{c}}|}$, we have $R_{sum} - R \le \log {K}$. 
\end{proof}

\section{Proof of Theorem \ref{thm_gauss}} \label{apdx_gap}
We split the proof into three cases: $\alpha \le \frac{1}{2}$, $1/2<\alpha \le \frac{2}{3}$, and $\alpha \ge 2$.

{\bf Case 1 ($\alpha \le \frac{1}{2}$):} We use the following parameters in Theorem \ref{thm_gauss1}: $\mu^{(1)} = \frac{1}{2{\mathsf {INR}}} \min \{2^{2 C_{FB}},{\mathsf {INR}}-1\}$, $\mu^{(2)} =\frac{1}{{\mathsf {INR}}}-\frac{1}{2{\mathsf {SNR}}} \min \{2^{2 C_{FB}},{\mathsf {INR}}-1\}$, and $\mu^{(4)} = \frac{1}{{\mathsf {INR}}}$. We first lower bound the right-hand sides (RHS) of \eqref{1w}-\eqref{6w} as follows.

RHS of \eqref{1w}:
\begin{eqnarray}
&&\log \left( \frac{{\mathsf {SNR}} \mu^{(1)}}{{\mathsf {SNR}} \mu^{(2:3)}+{\mathsf {SNR}}^{\alpha}\mu^{(1:3)}(K-1)+\mu^{(1:3)}} \right)\nonumber\\
&=& \log\left( \frac{\frac{1}{2}{\mathsf {SNR}^{1-\alpha}} \min \{2^{2 C_{FB}},{\mathsf {INR}}-1\}}{{\mathsf {SNR}^{1-\alpha}} +(K-1)
+\frac{1}{2}\min \{2^{2 C_{FB}},{\mathsf {INR}}-1\}(K-1)+\mu^{(1:3)}} \right)\nonumber\\
&\stackrel{(a)}{\ge}& \log\left( \frac{\frac{1}{2}{\mathsf {SNR}^{1-\alpha}} \min \{2^{2 C_{FB}},{\mathsf {INR}}-1\}}{{\mathsf {SNR}^{1-\alpha}} +(K-1)
+\frac{1}{2}\min \{2^{2 C_{FB}},{\mathsf {INR}}-1\}(K-1)+1} \right)\nonumber\\
&\stackrel{(b)}{\ge}& \log\left( \frac{\frac{1}{2}{\mathsf {SNR}^{1-\alpha}} \min \{2^{2 C_{FB}},{\mathsf {INR}}-1\}}{{\mathsf {SNR}^{1-\alpha}} +(K-1)
+\frac{1}{2}({\mathsf {INR}}-1)(K-1)+1} \right)\nonumber\\
&\stackrel{(c)}{\ge}& \log\left( \frac{\frac{1}{2}{\mathsf {SNR}^{1-\alpha}} \min \{2^{2 C_{FB}},{\mathsf {INR}}-1\}}{{\mathsf {SNR}^{1-\alpha}}\left(\frac{K+1}{2}\right)+\left(\frac{K+1}{2}\right)} \right)\nonumber\\
&=& \log\left( \frac{\frac{1}{2}{\mathsf {SNR}^{1-\alpha}} \min \{2^{2 C_{FB}},{\mathsf {INR}}-1\}}{{\mathsf {SNR}^{1-\alpha}}+1} \right)-\log(K+1)\nonumber\\
&\stackrel{(d)}{\ge}& \log\left( \frac{\frac{1}{2}{\mathsf {SNR}^{1-\alpha}} \min \{2^{2 C_{FB}},{\mathsf {INR}}-1\}}{{2\mathsf {SNR}^{1-\alpha}}} \right)-\log(K+1)\nonumber\\
&=& \log\left( \frac{{\mathsf {SNR}^{1-\alpha}} \min \{2^{2 C_{FB}},{\mathsf {INR}}-1\}}{{\mathsf {SNR}^{1-\alpha}}} \right)-\log4(K+1)\nonumber\\
&=& \log\left({ \min \{2^{2 C_{FB}},{\mathsf {INR}}-1\}} \right)-\log(4(K+1)),
\end{eqnarray}
where (a) follows since  $\mu^{(1:3)}\le 1$, (b) follows since  $\min \{2^{2 C_{FB}},{\mathsf {INR}}-1\}\le {\mathsf {INR}}-1$, (c) follows since  ${\mathsf {INR}}\le{\mathsf {SNR}^{1-\alpha}}$, and (d) follows since  $1\le{\mathsf {SNR}^{1-\alpha}}$.

RHS of \eqref{2w}:
\begin{eqnarray}
&&\log \left( \frac{{\mathsf {SNR}} \mu^{(2)}}{{\mathsf {SNR}} \mu^{(3)}+{\mathsf {SNR}}^{\alpha}\mu^{(1:3)}(K-1)+\mu^{(1:3)}} \right)\nonumber\\
&=& \log \left( \frac{{\mathsf {SNR}}^{1-\alpha}-\frac{1}{2}\min\{2^{2 C_{FB}},{\mathsf {INR}}-1\}}{\frac{1}{2}\min\{2^{2 C_{FB}},{\mathsf {INR}}-1\}+(K-1)+\frac{1}{2}\min\{2^{2 C_{FB}},{\mathsf {INR}}-1\}(K-1)+\mu^{(1:3)}}\right)\nonumber\\
&\stackrel{(a)}{\ge}& \log \left( \frac{{\mathsf {SNR}}^{1-\alpha}-\frac{1}{2}\min\{2^{2 C_{FB}},{\mathsf {INR}}-1\}}{\frac{1}{2}\min\{2^{2 C_{FB}},{\mathsf {INR}}-1\}+K+\frac{1}{2}\min\{2^{2 C_{FB}},{\mathsf {INR}}-1\}(K-1)}\right)\nonumber\\
&=& \log \left( \frac{{\mathsf {SNR}}^{1-\alpha}-\frac{1}{2}\min\{2^{2 C_{FB}},{\mathsf {INR}}-1\}}{K+\frac{1}{2}\min\{2^{2 C_{FB}},{\mathsf {INR}}-1\}K}\right)\nonumber\\
&\stackrel{(b)}{\ge}& \log \left( \frac{\frac{1}{2}({{\mathsf {SNR}}^{1-\alpha}+1})}{K+\frac{1}{2}\min\{2^{2 C_{FB}},{\mathsf {INR}}-1\}K}\right)\nonumber\\
&=& \log \left( \frac{\frac{1}{2}({{\mathsf {SNR}}^{1-\alpha}+1})}{1+\frac{1}{2}\min\{2^{2 C_{FB}},{\mathsf {INR}}-1\}}\right)-\log(K)\nonumber\\
&=& \log \left( \frac{({{\mathsf {SNR}}^{1-\alpha}+1})}{2+\min\{2^{2 C_{FB}},{\mathsf {INR}}-1\}}\right)-\log(K)\nonumber\\
&\stackrel{(c)}{\ge}& \log \left( \frac{({{\mathsf {SNR}}^{1-\alpha}+1})}{3\min\{2^{2 C_{FB}},{\mathsf {INR}}-1\}}\right)-\log(K)\nonumber\\
&=& \log \left({1+{\mathsf {SNR}}^{1-\alpha}}\right)
-\min \{2 C_{FB},\log({\mathsf {INR}}-1)\}-\log(3K),
\end{eqnarray}
where (a) follows since $\mu^{(1:3)}\le 1$, (b) follows since  $\min \{2^{2 C_{FB}},{\mathsf {INR}}-1\}\le {\mathsf {INR}}-1$ and ${\mathsf {INR}}\le{\mathsf {SNR}^{1-\alpha}}$, and (c) follows since $0 \le C_{FB}$ and $1 \le {\mathsf {INR}}$.

RHS of \eqref{3w}:
\begin{eqnarray}\label{1,3}
&&\log \left( \frac{{\mathsf {SNR}} \mu^{(3)}}{{\mathsf {SNR}}^{\alpha} \mu^{(2:3)} (K-1)+\mu^{(1:3)}} \right)\nonumber\\
&=& \log \left(  \frac{\frac{1}{2}\min\{2^{2 C_{FB}},{\mathsf {INR}}-1\}}{(K-1)+\mu^{(1:3)}} \right)\nonumber\\
&\stackrel{(a)}{\ge}& \log \left(  \frac{\frac{1}{2}\min\{2^{2 C_{FB}},{\mathsf {INR}}-1\}}{K} \right)\nonumber\\
&=& \log \left( {\min\{2^{2 C_{FB}},{\mathsf {INR}}-1\}}\right) -\log(2K)\nonumber\\
&=& \min \{2 C_{FB},\log({\mathsf {INR}}-1)\}-\log(2K),
\end{eqnarray}
where (a) follows since $\mu^{(1:3)}\le 1$.

RHS of \eqref{4w} is equal to RHS of \eqref{3w} since ${\mathsf {SNR}} \mu^{(3)}={\mathsf {SNR}}^{\alpha} \mu^{(1)}$.

RHS of \eqref{5w}:
\begin{eqnarray}
&& \log \left(  \frac{(\sqrt{\mathsf {SNR}}{\left(\frac{1}{K-1}\right)}+\sqrt{{\mathsf {SNR}}^{\alpha}}\left(\frac{K-2}{K-1}\right))^{2} \mu^{(1)}}{{\mathsf {SNR}}\mu^{(4)}+{\mathsf {SNR}}^{\alpha} \mu^{(4)}(K-1)+(\mu^{(1)}+\mu^{(4)})} \right)\nonumber\\
&\stackrel{(a)}{\ge}& \log \left(  \frac{{\mathsf {SNR}}{\left(\frac{1}{K-1}\right)^{2}} \mu^{(1)}}{{\mathsf {SNR}}\mu^{(4)}+{\mathsf {SNR}}^{\alpha} \mu^{(4)}(K-1)+(\mu^{(1)}+\mu^{(4)})} \right)\nonumber\\
&=& \log \left(  \frac{{\mathsf {SNR}}{\left(\frac{1}{K-1}\right)^{2}} \mu^{(1)}}{{\mathsf {SNR}}\mu^{(4)}+(K-1)+(\mu^{(1)}+\mu^{(4)})} \right)\nonumber\\
&\stackrel{(b)}{\ge}& \log \left(  \frac{{\mathsf {SNR}}{\left(\frac{1}{K-1}\right)^{2}} \mu^{(1)}}{{\mathsf {SNR}}\mu^{(4)}+K} \right)\nonumber\\
&=& \log \left(  \frac{\frac{1}{2}{\left(\frac{1}{K-1}\right)^{2}}
{\mathsf {SNR}}^{1-\alpha}\min\{2^{2 C_{FB}},{\mathsf {INR}}-1\}}{{\mathsf {SNR}}^{1-\alpha}+K} \right)\nonumber\\
&\stackrel{(c)}{\ge}& \log \left(  \frac{\frac{1}{2}{\left(\frac{1}{K-1}\right)^{2}}
{\mathsf {SNR}}^{1-\alpha}\min\{2^{2 C_{FB}},{\mathsf {INR}}-1\}}{(K+1){\mathsf {SNR}}^{1-\alpha}} \right)\nonumber\\
&=& \log \left(  \frac{\frac{1}{2}{\left(\frac{1}{K-1}\right)^{2}}
\min\{2^{2 C_{FB}},{\mathsf {INR}}-1\}}{(K+1)} \right)\nonumber\\
&=& \log \left( \frac{1}{2} {{\left(\frac{1}{K-1}\right)^{2}}
\min\{2^{2 C_{FB}},{\mathsf {INR}}-1\}} \right)-\log(K+1)\nonumber\\
&=& \log \left(  \min\{2^{2 C_{FB}},{\mathsf {INR}}-1\} \right)-\log2(K+1)-2\log(K-1)\nonumber\\
&=& \min \{2 C_{FB},\log({\mathsf {INR}}-1)\}-\log2(K+1)-2\log(K-1),
\end{eqnarray}
where (a) follows since  $\sqrt{\mathsf {SNR}}{\left(\frac{1}{K-1}\right)}\le \sqrt{\mathsf {SNR}}{\left(\frac{1}{K-1}\right)}+\sqrt{{\mathsf {SNR}}^{\alpha}}\left(\frac{K-2}{K-1}\right)$, (b) follows since  $\mu^{(1)}+\mu^{(4)}\le 1$, and (c) follows since  $1 \le {\mathsf {SNR}}$.

RHS of \eqref{6w}:
\begin{eqnarray}
&& \log \left(  \frac{{\mathsf {SNR}} \mu^{(4)}}{{\mathsf {SNR}}^{\alpha}\mu^{(4)}(K-1)+(\mu^{(1)}+\mu^{(4)})} \right)\nonumber\\
&=& \log \left(  \frac{{\mathsf {SNR}}^{1-\alpha} }{(K-1)+(\mu^{(1)}+\mu^{(4)})} \right)\nonumber\\
&\stackrel{(a)}{\ge}& \log \left(  \frac{{\mathsf {SNR}}^{1-\alpha} }{K} \right)\nonumber\\
&=& \log \left( {\mathsf {SNR}}^{1-\alpha}\right)-\log(K),
\end{eqnarray}
where (a) follows since  $\mu^{(1)}+\mu^{(4)}\le 1$.

Also we do not need \eqref{7w1} and \eqref{7w2} anymore, because we have tighter bounds for $R^{(1)}$ and $R^{(3)}$ in \eqref{1,3}. Thus, we find the achievable rate expressions can be reduced as follows:
\begin{eqnarray}
R^{(1)} &\le& \min \{2 C_{FB},\log({\mathsf {INR}}-1)\}-\log2(K+1)-2\log(K-1)\\
R^{(2)} &\le&  \log \left({1+{\mathsf {SNR}}^{1-\alpha}}\right)
-\min \{2 C_{FB},\log({\mathsf {INR}}-1)\}-\log(3K)\\
R^{(3)} &\le& \min \{2 C_{FB},\log({\mathsf {INR}}-1)\}-\log(2K)\\
R^{(4)} &\le& \log \left( {\mathsf {SNR}}^{1-\alpha}\right)-\log(K).
\end{eqnarray}
Putting these bounds all together, we achieve $\frac{R^{(1)}+R^{(2)}+R^{(3)}+R^{(4)}}{2}=$
\begin{eqnarray}\label{ac1}
&&\frac{1}{2}\log \left( 1+{\mathsf {SNR}}^{1-\alpha}\right)+\min \{C_{FB},\frac{1}{2}\log({\mathsf {INR}}-1)\}
+\frac{1}{2}\log \left({\mathsf {SNR}}^{1-\alpha}\right)
-\frac{1}{2}\log(K+1)\nonumber\\
&&-\log(K-1)
-\frac{3}{2}\log(K)-\frac{1}{2}\log(12).
\end{eqnarray}
Next we will bound the gap between \eqref{ac1} and the conjectured rate upper bound in \eqref{conj}. We split this into 2 regimes. The first is when  $C_{FB}\le \frac{1}{2}\log({\mathsf {INR}}-1)$, and the second is when $C_{FB}> \frac{1}{2}\log({\mathsf {INR}}-1)$. In the first case, we find the distance between \eqref{ac1} and ${R}^u_{sym,0}+C_{FB}$ to get
\begin{eqnarray}\label{weakbound1}
&&{R}^u_{sym,0}+C_{FB}-\left(\frac{1}{2}\log \left( 1+{\mathsf {SNR}}^{1-\alpha}\right)+\min \{C_{FB},\frac{1}{2}\log({\mathsf {INR}}-1)\}\frac{1}{2}\log \left({\mathsf {SNR}}^{1-\alpha}\right)\right.\nonumber\\
&&\left.-\frac{1}{2}\log(K+1)-\log(K-1)-\frac{3}{2}\log(K)
-\frac{1}{2}\log(12)\right)\nonumber\\
&\le&\left(\log(1+{\mathsf{INR}}+\frac{{\mathsf{SNR}}}{1+\mathsf{INR}})
+C_{FB}\right)-\left(\frac{1}{2}\log \left( 1+{\mathsf {SNR}}^{1-\alpha}\right)+\min \{C_{FB},\frac{1}{2}\log({\mathsf {INR}}-1)\}\right.\nonumber\\
&&\left.\frac{1}{2}\log \left({\mathsf {SNR}}^{1-\alpha}\right)-\frac{1}{2}\log(K+1)-\log(K-1)-\frac{3}{2}\log(K)
-\frac{1}{2}\log(12)\right)\nonumber\\
&=& \log(1+{\mathsf{INR}}+\frac{{\mathsf{SNR}}}{1+\mathsf{INR}})
-\frac{1}{2}\log \left( 1+{\mathsf {SNR}}^{1-\alpha}\right)-\frac{1}{2}\log \left( {\mathsf {SNR}}^{1-\alpha}\right)+\frac{1}{2}\log(K+1)+\log(K-1)\nonumber\\
&&+\frac{3}{2}\log(K)+\frac{1}{2}\log(12)\nonumber\\
&\le& \log(1+{\mathsf{INR}}+\frac{{\mathsf{SNR}}}{\mathsf{INR}})
-\frac{1}{2}\log \left( 1+{\mathsf {SNR}}^{1-\alpha}\right)-\frac{1}{2}\log \left({\mathsf {SNR}}^{1-\alpha}\right)+\frac{1}{2}\log(K+1)+\log(K-1)\nonumber\\
&&+\frac{3}{2}\log(K)+\frac{1}{2}\log(12)\nonumber\\
&\le& \log(1+\frac{2{\mathsf{SNR}}}{\mathsf{INR}})
-\frac{1}{2}\log \left( 1+{\mathsf {SNR}}^{1-\alpha}\right)-\frac{1}{2}\log \left( {\mathsf {SNR}}^{1-\alpha}\right)+\frac{1}{2}\log(K+1)+\log(K-1)\nonumber\\
&&+\frac{3}{2}\log(K)+\frac{1}{2}\log(12)\nonumber\\
&=& \log(1+2{\mathsf {SNR}}^{1-\alpha})
-\frac{1}{2}\log \left( 1+{\mathsf {SNR}}^{1-\alpha}\right)-\frac{1}{2}\log \left( {\mathsf {SNR}}^{1-\alpha}\right)+\frac{1}{2}\log(K+1)+\log(K-1)\nonumber\\
&&+\frac{3}{2}\log(K)+\frac{1}{2}\log(12)\nonumber\\
&=& \frac{1}{2}\log(4(\frac{1}{2}+{\mathsf {SNR}}^{1-\alpha})^2)
-\frac{1}{2}\log \left( 1+{\mathsf {SNR}}^{1-\alpha}\right)-\frac{1}{2}\log \left( {\mathsf {SNR}}^{1-\alpha}\right)+\frac{1}{2}\log(K+1)+\log(K-1)\nonumber\\
&&+\frac{3}{2}\log(K)+\frac{1}{2}\log(12)\nonumber\\
&=& \frac{1}{2}\log\left(\frac{4\left(\frac{1}{2}+{\mathsf {SNR}}^{1-\alpha}\right)^2}{\left( 1+{\mathsf {SNR}}^{1-\alpha}\right)\left( {\mathsf {SNR}}^{1-\alpha}\right)}\right) +\frac{1}{2}\log(K+1)+\log(K-1)+\frac{3}{2}\log(K)+\frac{1}{2}
\log(12)\nonumber\\
&=& \frac{1}{2}\log\left(\frac{\left(\frac{1}{2}+{\mathsf {SNR}}^{1-\alpha}\right)^2}{\left( 1+{\mathsf {SNR}}^{1-\alpha}\right)\left( {\mathsf {SNR}}^{1-\alpha}\right)}\right) +\frac{1}{2}\log(K+1)+\log(K-1)+\frac{3}{2}\log(K)+\frac{1}{2}
\log(48)\nonumber\\
&=& \frac{1}{2}\log\left(1+\frac{1}{4\left( 1+{\mathsf {SNR}}^{1-\alpha}\right)\left( {\mathsf {SNR}}^{1-\alpha}\right)}\right) +\frac{1}{2}\log(K+1)+\log(K-1)+\frac{3}{2}\log(K)+\frac{1}{2}
\log(48)\nonumber\\
&\le& \frac{1}{2}\log\left(1+\frac{1}{8}\right) +\frac{1}{2}\log(K+1)+\log(K-1)+\frac{3}{2}\log(K)+\frac{1}{2}
\log(48)\nonumber\\
&=& \frac{1}{2}\log(K+1)+\log(K-1)+\frac{3}{2}\log(K)+\frac{1}{2}
\log(54).
\end{eqnarray}
In the second case when $C_{FB}> \frac{1}{2}\log({\mathsf {INR}}-1)$, we find the distance between \eqref{ac1} and ${R}^u_{sym,\infty}$ as follows.

Since $\min \{C_{FB},\frac{1}{2}\log({\mathsf {INR}}-1)\}=\frac{1}{2}\log({\mathsf {INR}}-1)$
\begin{eqnarray}\label{weakbound2}
&&{R}^u_{sym,\infty}-\left(\frac{1}{2}\log \left( 1+{\mathsf {SNR}}^{1-\alpha}\right)+\min \{C_{FB},\frac{1}{2}\log({\mathsf {INR}}+1)\}-\frac{1}{2}\log \left({\mathsf {SNR}}^{1-\alpha}\right)\right.\nonumber\\
&&\left.-\frac{1}{2}\log(K+1)-\log(K-1)
-\frac{3}{2}\log(K)-\frac{1}{2}\log(12)\right)\nonumber\\
&=&\left(\frac{1}{2}\log(1+\frac{{\mathsf{SNR}}}{1+\mathsf{INR}})+
\frac{1}{2}\log(1+{\mathsf{SNR}}+{\mathsf{INR}})+\frac{K-1}{2}+\log K\right)-\left(\frac{1}{2}\log \left( 1+{\mathsf {SNR}}^{1-\alpha}\right)\right.\nonumber\\
&&\left.+\min \{C_{FB},\frac{1}{2}\log({\mathsf {INR}}+1)\}-\frac{1}{2}\log \left({\mathsf {SNR}}^{1-\alpha}\right)-\frac{1}{2}\log(K+1)-\log(K-1)
-\frac{3}{2}\log(K)\right.\nonumber\\
&&\left.-\frac{1}{2}\log(12)\right)\nonumber\\
&=&\frac{1}{2}\log(1+\frac{{\mathsf{SNR}}}{1+\mathsf{INR}})+
\frac{1}{2}\log(1+{\mathsf{SNR}}+{\mathsf{INR}})
-\frac{1}{2}\log \left( 1+{\mathsf {SNR}}^{1-\alpha}\right)-\frac{1}{2}\log({\mathsf {INR}}+1)\nonumber\\
&&-\frac{1}{2}\log \left({\mathsf {SNR}}^{1-\alpha}\right)+\frac{K-1}{2}+\log K+\frac{1}{2}\log(K+1)+\log(K-1)+\frac{3}{2}\log(K)+\frac{1}{2}\log(12)\nonumber\\
&=&\frac{1}{2}\log(1+\frac{{\mathsf{SNR}}}{1+{{\mathsf{SNR}}^{\alpha}}})+
\frac{1}{2}\log(1+{\mathsf{SNR}}+{\mathsf{SNR}}^{\alpha})
-\frac{1}{2}\log \left( 1+{\mathsf {SNR}}^{1-\alpha}\right)-\frac{1}{2}\log({\mathsf {SNR}}^{\alpha}+1)\nonumber\\
&&-\frac{1}{2}\log \left( {\mathsf {SNR}}^{1-\alpha}\right)+\frac{K-1}{2}+\log K+\frac{1}{2}\log(K+1)+\log(K-1)+\frac{3}{2}\log(K)+\frac{1}{2}\log(12)\nonumber\\
&=&\frac{1}{2}\log(\frac{1+{{\mathsf{SNR}}^{\alpha}}+{\mathsf{SNR}}}{1+{{\mathsf{SNR}}^{\alpha}}})+
\frac{1}{2}\log(1+{\mathsf{SNR}}+{\mathsf{SNR}}^{\alpha})
-\frac{1}{2}\log \left( 1+{\mathsf {SNR}}^{1-\alpha}\right)-\frac{1}{2}\log({\mathsf {SNR}}^{\alpha}+1)\nonumber\\
&&-\frac{1}{2}\log \left({\mathsf {SNR}}^{1-\alpha}\right)+\frac{K-1}{2}+\log K+\frac{1}{2}\log(K+1)+\log(K-1)+\frac{3}{2}\log(K)+\frac{1}{2}
\log(12)\nonumber\\
&\le&\frac{1}{2}\log(\frac{1+2{\mathsf{SNR}}}{1+{{\mathsf{SNR}}^{\alpha}}})+
\frac{1}{2}\log(1+{2\mathsf{SNR}})
-\frac{1}{2}\log \left( 1+{\mathsf {SNR}}^{1-\alpha}\right)-\frac{1}{2}\log({\mathsf {SNR}}^{\alpha}+1)\nonumber\\
&&-\frac{1}{2}\log \left({\mathsf {SNR}}^{1-\alpha}\right)+\frac{K-1}{2}+\log K+\frac{1}{2}\log(K+1)+\log(K-1)+\frac{3}{2}\log(K)+\frac{1}{2}\log(12)\nonumber\\
&=&\frac{1}{2}\log\left(\frac{(1+2{\mathsf{SNR}})(1+{2\mathsf{SNR}})}
{(1+{{\mathsf{SNR}}^{\alpha}})^2{(1+{\mathsf {SNR}}^{1-\alpha})}({\mathsf {SNR}}^{1-\alpha})}\right)
+\frac{K-1}{2}+\log K+\frac{1}{2}\log(K+1)+\log(K-1)\nonumber\\
&&+\frac{3}{2}\log(K)+\frac{1}{2}\log(12)\nonumber\\
&\le&\frac{1}{2}\log\left(\frac{(3{\mathsf{SNR}})({3\mathsf{SNR}})}
{({{\mathsf{SNR}}^{\alpha}})^2({\mathsf {SNR}}^{1-\alpha})^2}\right)
+\frac{K-1}{2}+\log K+\frac{1}{2}\log(K+1)+\log(K-1)\nonumber\\
&&+\frac{3}{2}\log(K)+\frac{1}{2}\log(12)\nonumber\\
&=&\frac{1}{2}\log\left(\frac{9{\mathsf{SNR}}^2}
{{\mathsf{SNR}}^2}\right)
+\frac{K-1}{2}+\log K+\frac{1}{2}\log(K+1)+\log(K-1)+\frac{3}{2}\log(K)+\frac{1}{2}\log(12)\nonumber\\
&=& \frac{K-1}{2}+\frac{1}{2}\log(K+1)+\log(K-1)+\frac{5}{2}\log(K)+\frac{1}{2}\log\left(108\right)\nonumber\\
&=& \frac{K-1}{2}+\frac{1}{2}\log\left(108(K-1){(K)^5}(K+1)\right).
\end{eqnarray}
From \eqref{weakbound1} and \eqref{weakbound2}, we find that the achievable symmetric rate is within $\frac{1}{2}\log\left(108(K-1){(K)^5}(K+1)\right)+\frac{K-1}{2}$ bits to the conjectured upper bound in \eqref{conj} when $\alpha \le \frac{1}{2}$.

{\bf Case 2 ($\frac{1}{2} \le \alpha \le \frac{2}{3}$):} We use the following parameters in Theorem \ref{thm_gauss2}: $\mu^{(4)} = \frac{1}{4{\mathsf {INR}}} \max \{2^{-2 C_{FB}},\frac{{\mathsf {INR}}^3}{{\mathsf {SNR}}^2}\},$ $\mu^{(6)} = \mu^{(3)} = \frac{1}{3\mathsf {INR}}-\frac{1}{4{\mathsf {INR}}} \max \{2^{-2 C_{FB}},\frac{{\mathsf {INR}}^3}{{\mathsf {SNR}}^2}\},$ $\mu^{(1)} = 1 - \mu^{(2:4)},$ and $\mu^{(5)} = 1 - \mu^{(2)}- \mu^{(6)}$. We first lower bound the RHS of  \eqref{1}-\eqref{10} as follows.

RHS of \eqref{1}:
\begin{eqnarray}
&& \log \left( \frac{{\mathsf {SNR}} \mu^{(1)}}{{\mathsf {SNR}} \mu^{(2:4)}+{\mathsf {SNR}}^{\alpha} \mu^{(1:4)} (K-1)+1} \right)\nonumber\\
&\stackrel{(a)}{\ge}& \log \left( \frac{\frac{2}{3}{\mathsf {SNR}}}{{\frac{2}{3}{\mathsf {SNR}}^{2-2\alpha}}+{\mathsf {SNR}}^{\alpha}(K-1)+1} \right)\nonumber\\
&\stackrel{(b)}{\ge}& \log \left( \frac{\frac{2}{3}{\mathsf {SNR}}}{\left(K+\frac{2}{3}\right){\mathsf {SNR}}^{2-2\alpha}} \right)\nonumber\\
&=& \log \left( \frac{{\mathsf {SNR}}^{2\alpha-1}}{\frac{3}{2}\left(K+\frac{2}{3}\right)} \right)\nonumber\\
&=& \log \left( {{\mathsf {SNR}}^{2\alpha-1}}\right)-\log\left(\frac{3}{2}\left(K+\frac{2}{3}
\right)\right),
\end{eqnarray}
where (a) follows since  $\mu^{(1:4)}=1$, $\mu^{(1)}\ge \frac{2}{3}$, and $\mu^{(2:4)}\le \frac{2}{3}{\mathsf {SNR}^{1-2\alpha}}$, and (b) follows since  ${\mathsf {SNR}^{2-2\alpha}}\ge {\mathsf {SNR}^{\alpha}}\ge 1$.

RHS of \eqref{2}:
\begin{eqnarray}
&& \log \left( \frac{{\mathsf {SNR}} \mu^{(2)}}{{\mathsf {SNR}} \mu^{(3:4)}+{\mathsf {SNR}}^{\alpha} \mu^{(1:4)}(K-1)+1} \right)\nonumber\\
&\stackrel{(a)}{\ge}& \log \left( \frac{{\frac{1}{12}}{{\mathsf {SNR}}^{2-2\alpha}}}{{\frac{1}{3}}{{\mathsf {SNR}}^{1-\alpha}}+{\mathsf {SNR}}^{\alpha}(K-1)+1} \right)\nonumber\\
&\stackrel{(b)}{\ge}& \log \left( \frac{{\frac{1}{12}}{{\mathsf {SNR}}^{2-2\alpha}}}{{\mathsf {SNR}}^{\alpha}(K+{\frac{1}{3}})} \right)\nonumber\\
&=& \log \left( \frac{{{\mathsf {SNR}}^{2-3\alpha}}}{12(K+{\frac{1}{3}})} \right)\nonumber\\
&=& \log \left({{{\mathsf {SNR}}^{2-3\alpha}}}\right)-\log\left(12\left(K+{\frac{1}{3}}
\right)\right),
\end{eqnarray}
where (a) follows since  $\mu^{(1:4)}=1$, $\mu^{(1)}\ge \frac{2}{3}$, and $\mu^{(2:4)}\le \frac{2}{3}{\mathsf {SNR}^{1-2\alpha}}$, and (b) follows since  ${\mathsf {SNR}^{2-2\alpha}}\ge {\mathsf {SNR}^{\alpha}}\ge 1$.

RHS of \eqref{3}:
\begin{eqnarray}
&& \log \left( \frac{{\mathsf {SNR}}^{\alpha} \mu^{(1)}}{{\mathsf {SNR}}\mu^{(3:4)}+{\mathsf {SNR}}^{\alpha} \mu^{(2:4)} (K-1)+1} \right)\nonumber\\
&\stackrel{(a)}{\ge}& \log \left( \frac{{\frac{2}{3}}{\mathsf {SNR}}^{\alpha}}{{\frac{1}{3}}{{\mathsf {SNR}}^{1-\alpha}}+{\frac{1}{3}}{{\mathsf {SNR}}^{1-\alpha}}(K-1)+1} \right)\nonumber\\
&=& \log \left( \frac{2{\mathsf {SNR}}^{\alpha}}{{{\mathsf {SNR}}^{1-\alpha}}(K)+3} \right)\nonumber\\
&\stackrel{(b)}{\ge}& \log \left( \frac{2{\mathsf {SNR}}^{\alpha}}{{{\mathsf {SNR}}^{1-\alpha}}(K+3)} \right)\nonumber\\
&=& \log \left( \frac{{\mathsf {SNR}}^{\alpha}}{{{\mathsf {SNR}}^{1-\alpha}}} \right)-\log\left(\frac{1}{2}\left(K+3\right)\right)\nonumber\\
&=& \log \left( {{\mathsf {SNR}}^{2\alpha-1}}\right)-\log\left(\frac{1}{2}\left(K+3\right)
\right),
\end{eqnarray}
where (a) follows since  $\mu^{(1:4)}=1$, $\mu^{(1)}\ge \frac{2}{3}$, and $\mu^{(3:4)}\le \frac{1}{3}{\mathsf {SNR}^{-\alpha}}$, and (b) follows since  ${\mathsf {SNR}^{1-\alpha}}\ge 1$.

RHS of \eqref{4}:
\begin{eqnarray}\label{2,,3}
&& \log \left( \frac{{\mathsf {SNR}}^{\alpha} \mu^{(2)}}{{\mathsf {SNR}}\mu^{(4)}+{\mathsf {SNR}}^{\alpha} \mu^{(3:4)} (K-1)+1} \right)\nonumber\\
&\stackrel{(a)}{\ge}& \log \left( \frac{{\frac{1}{12}}{\mathsf {SNR}}^{1-\alpha}}{{\frac{1}{4}}{{\mathsf {SNR}}^{1-\alpha}}\max \{2^{-2 C_{FB}},\frac{{\mathsf {INR}}^3}{{\mathsf {SNR}}^2}\}+{\frac{1}{3}}(K-1)+1} \right)\nonumber\\
&{\ge}& \log \left( \frac{{\mathsf {SNR}}^{1-\alpha}}{3{{\mathsf {SNR}}^{1-\alpha}}\max \{2^{-2 C_{FB}},\frac{{\mathsf {INR}}^3}{{\mathsf {SNR}}^2}\}+4(K-1)+12} \right)\nonumber\\
&=& \log \left( \frac{{\mathsf {SNR}}^{1-\alpha}}{4(K+{\frac{11}{4}}){{\mathsf {SNR}}^{1-\alpha}}\max \{2^{-2 C_{FB}},\frac{{\mathsf {INR}}^3}{{\mathsf {SNR}}^2}\}} \right)\nonumber\\
&=& \log \left( \frac{{\mathsf {SNR}}^{1-\alpha}}{4(K+{\frac{11}{4}}){{\mathsf {SNR}}^{1-\alpha}}\max \{2^{-2 C_{FB}},\frac{{\mathsf {INR}}^3}{{\mathsf {SNR}}^2}\}} \right)\nonumber\\
&=& \log \left( \frac{\min \{2^{2 C_{FB}},\frac{{\mathsf {SNR}}^2}{{\mathsf {INR}}^3}\}}{4\left(K+{\frac{11}{4}}\right)} \right)\nonumber\\
&=& \log \left( {\min \{2^{2 C_{FB}},\frac{{\mathsf {SNR}}^2}{{\mathsf {INR}}^3}\}} \right)-\log\left(4\left(K+{\frac{11}{4}}\right)\right)\nonumber\\
&=& \min \{2 C_{FB},\log\left({\mathsf {SNR}}^{2-3\alpha}\right)\}-\log\left(4\left(K+{\frac{11}{4}}\right)
\right),
\end{eqnarray}
where (a) follows since  $\mu^{(2)}\ge \frac{1}{12}{\mathsf {SNR}^{1-2\alpha}}$, and $\mu^{(3:4)}\le \frac{1}{3}{\mathsf {SNR}^{-\alpha}}$.

RHS of \eqref{5} is equal to RHS of \eqref{4} since ${\mathsf {SNR}} \mu^{(3)}={\mathsf {SNR}}^{\alpha} \mu^{(2)}$.

RHS of \eqref{6}:
\begin{eqnarray}
&& \log \left( \frac{{\mathsf {SNR}} \mu^{(4)}}{{\mathsf {SNR}}^{\alpha} \mu^{(3:4)} (K-1)+1} \right)\nonumber\\
&\stackrel{(a)}{\ge}& \log \left( \frac{{\frac{1}{4}}{\mathsf {SNR}}^{1-\alpha}\max \{2^{-2 C_{FB}},\frac{{\mathsf {INR}}^3}{{\mathsf {SNR}}^2}\}}{{\frac{1}{3}}(K-1)+1} \right)\nonumber\\
&=& \log \left( \frac{{\mathsf {SNR}}^{1-\alpha}\max \{2^{-2 C_{FB}},\frac{{\mathsf {INR}}^3}{{\mathsf {SNR}}^2}\}}{{\frac{4}{3}}(K+2)} \right)\nonumber\\
&=& \log \left({\mathsf {SNR}}^{1-\alpha}\max \{2^{-2 C_{FB}},\frac{{\mathsf {INR}}^3}{{\mathsf {SNR}}^2}\} \right)-\log\left(\frac{4}{3}\left(K+2\right)\right)\nonumber\\
&=& \log \left( {\mathsf {SNR}}^{1-\alpha}\right)+\log \left(\max \{2^{-2 C_{FB}},\frac{{\mathsf {INR}}^3}{{\mathsf {SNR}}^2}\}\right)
-\log\left(\frac{4}{3}\left(K+2\right)\right)\nonumber\\
&=& \log \left( {\mathsf {SNR}}^{1-\alpha}\right)-\log \left(\min \{2^{2 C_{FB}},\frac{{\mathsf {SNR}}^2}{{\mathsf {INR}}^3}\}\right)-\log\left(\frac{4}{3}\left(K+2\right)\right)\nonumber\\
&=& \log \left({\mathsf {SNR}}^{1-\alpha}\right)-\min \{2 C_{FB},\log\left({\mathsf {SNR}}^{2-3\alpha}\right)\}-\log\left(\frac{4}{3}
\left(K+2\right)\right),
\end{eqnarray}
where (a) follows since  $\mu^{(3:4)}\le \frac{1}{3}{\mathsf {SNR}^{-\alpha}}$.

RHS of \eqref{7}:
\begin{eqnarray}
&& \log \left( \frac{{\mathsf {SNR}} \mu^{(5)}}{{\mathsf {SNR}}(\mu^{(2)}+\mu^{(6)})+{\mathsf {SNR}}^{\alpha} (\mu^{(2)}+\mu^{(5:6)}) (K-1)+1-{\mathsf {SNR}}^{\alpha} \mu^{(2)}} \right)\nonumber\\
&\stackrel{(a)}{\ge}& \log \left( \frac{{\frac{1}{3}}{\mathsf {SNR}}}{{\mathsf {SNR}}(\mu^{(2)}+\mu^{(6)})+{\mathsf {SNR}}^{\alpha} (K-1)+1} \right)\nonumber\\
&\stackrel{(b)}{\ge}& \log \left( \frac{{\frac{1}{3}}{\mathsf {SNR}}}{{\mathsf {SNR}}
({\frac{2}{3}}{{\mathsf {SNR}}^{1-2\alpha}})
+{\mathsf {SNR}}^{\alpha} (K-1)+1} \right)\nonumber\\
&\stackrel{(c)}{\ge}& \log \left( \frac{{\frac{1}{3}}{\mathsf {SNR}}}{({{\mathsf {SNR}}^{2-2\alpha}})(K+{\frac{2}{3}})} \right)\nonumber\\
&=& \log \left( {{{\mathsf {SNR}}^{2\alpha-1}}} \right)-\log\left(3\left(K+{\frac{2}{3}}\right)\right),
\end{eqnarray}
where (a) follows since  $\mu^{(2)}+\mu^{(5:6)}=1$, and $\mu^{(5)}\ge \frac{1}{3}$, (b) follows since  $\mu^{(2)}+\mu^{(6)}\le\frac{2}{3}{\mathsf {SNR}^{1-2\alpha}}$, and (c) follows since  ${\mathsf {SNR}^{2-2\alpha}}\ge {\mathsf {SNR}^{\alpha}}\ge 1$.

RHS of \eqref{8}:
\begin{eqnarray}
&& \log \left( \frac{(\sqrt{\mathsf {SNR}}\left(\frac{1}{K-1}\right)+\sqrt{{\mathsf {SNR}}^{\alpha}}\left(\frac{K-2}{K-1}\right))^{2} \mu^{(2)}}{{\mathsf {SNR}}\mu^{(6)}+{\mathsf {SNR}}^{\alpha} \mu^{(5:6)}(K-1)+1} \right)\nonumber\\
&\stackrel{(a)}{\ge}& \log \left( \frac{{\mathsf {SNR}}\left(\frac{1}{K-1}\right)^{2} \mu^{(2)}}{{\mathsf {SNR}}\mu^{(6)}+{\mathsf {SNR}}^{\alpha} \mu^{(5:6)}(K-1)+1} \right)\nonumber\\
&\stackrel{(b)}{\ge}& \log \left( \frac{{\frac{1}{12}}{{\mathsf {SNR}}^{2-2\alpha}}\left(\frac{1}{K-1}\right)^{2} }{{\frac{1}{3}}{\mathsf {SNR}}^{1-\alpha}+{\mathsf {SNR}}^{\alpha}(K-1)+1} \right)\nonumber\\
&\stackrel{(c)}{\ge}& \log \left( \frac{{\frac{1}{12}}{{\mathsf {SNR}}^{2-2\alpha}}\left(\frac{1}{K-1}\right)^{2} }{{\mathsf {SNR}}^{\alpha}(K+{\frac{1}{3}})} \right)\nonumber\\
&=& \log \left({{\mathsf {SNR}}^{2-3\alpha}} \right)-\log\left(12{(K-1)^2}(K+{\frac{1}{3}})\right),
\end{eqnarray}
where (a) follows since  $\sqrt{\mathsf {SNR}}\left(\frac{1}{K-1}\right)\le \sqrt{\mathsf {SNR}}\left(\frac{1}{K-1}\right)+\sqrt{{\mathsf {SNR}}^{\alpha}}\left(\frac{K-2}{K-1}\right)$, (b) follows since  $\mu^{(5:6)}\le 1$, $\mu^{(6)}\le \frac{1}{3}{\mathsf {SNR}^{-\alpha}}$, and $\mu^{(2)}\ge \frac{1}{12}{\mathsf {SNR}^{1-2\alpha}}$, and (c) follows since  ${\mathsf {SNR}^{\alpha}}\ge {\mathsf {SNR}^{1-\alpha}}\ge 1$.

RHS of \eqref{9}:
\begin{eqnarray}
&& \log \left( \frac{{\mathsf {SNR}}^{\alpha} \mu^{(5)}}{{\mathsf {SNR}}\mu^{(6)}+{\mathsf {SNR}}^{\alpha} \mu^{(6)} (K-1)+1} \right)\nonumber\\
&\stackrel{(a)}{\ge}& \log \left( \frac{{\frac{1}{3}}{\mathsf {SNR}}^{\alpha}}{{\frac{1}{3}}{\mathsf {SNR}}^{1-\alpha}+{\frac{1}{3}}(K-1)+1} \right)\nonumber\\
&=& \log \left( \frac{{\mathsf {SNR}}^{\alpha}}{{\mathsf {SNR}}^{1-\alpha}+(K-1)+3} \right)\nonumber\\
&\stackrel{(b)}{\ge}& \log \left( \frac{{\mathsf {SNR}}^{\alpha}}{(K+3){\mathsf {SNR}}^{1-\alpha}} \right)\nonumber\\
&=& \log \left( \frac{{\mathsf {SNR}}^{\alpha}}{{\mathsf {SNR}}^{1-\alpha}} \right)-\log(K+3)\nonumber\\
&=& \log \left({\mathsf {SNR}}^{2\alpha-1}\right)-\log(K+3),
\end{eqnarray}
where (a) follows since  $\mu^{(6)}\le \frac{1}{3}{\mathsf {SNR}^{-\alpha}}$, and $\mu^{(5)}\ge \frac{1}{3}$, and (b) follows since  ${\mathsf {SNR}^{1-\alpha}}\ge 1$.

RHS of \eqref{10}:
\begin{eqnarray}
&& \log \left(  \frac{{\mathsf {SNR}}\mu^{(6)}}{{\mathsf {SNR}}^{\alpha} \mu^{(6)} (K-1)+1} \right)\nonumber\\
&=& \log \left(  \frac{\frac{1}{3}{{\mathsf {SNR}}^{1-\alpha}}-\frac{1}{4} {{\mathsf {SNR}}^{1-\alpha}} \max \{2^{-2 C_{FB}},\frac{{\mathsf {INR}}^3}{{\mathsf {SNR}}^2}\}}{({\frac{1}{3}-\frac{1}{4} \max \{2^{-2 C_{FB}},\frac{{\mathsf {INR}}^3}{{\mathsf {SNR}}^2}\}})(K-1)+1} \right)\nonumber\\
&{\ge}& \log \left( \frac{\frac{1}{12}{{\mathsf {SNR}}^{1-\alpha}}}{{\frac{1}{3}}(K-1)+1} \right)\nonumber\\
&=& \log \left( \frac{\frac{1}{4}{{\mathsf {SNR}}^{1-\alpha}}}{(K+2)} \right)\nonumber\\
&=& \log \left({{\mathsf {SNR}}^{1-\alpha}}\right)-\log(4(K+2)).
\end{eqnarray}
Also we do not need \eqref{11.1} and \eqref{11.2} anymore, because we have tighter bounds for $R^{(2)}$ and $R^{(3)}$ in \eqref{2,,3}. Thus, we find the achievable rate expressions can be reduced as follows
\begin{eqnarray}
R^{(1)} &\le& \log \left( {{\mathsf {SNR}}^{2\alpha-1}}\right)-\log\left(\frac{3}{2}\left(K+
\frac{2}{3}\right)\right)\\
R^{(2)} &\le&  \min \{2 C_{FB},\log\left({\mathsf {SNR}}^{2-3\alpha}\right)\}-\log\left(12{(K-1)^2}
(K+{\frac{1}{3}})\right)\\
R^{(3)} &\le& \min \{2 C_{FB},\log\left({\mathsf {SNR}}^{2-3\alpha}\right)\}-\log\left(4\left(K+
{\frac{11}{4}}\right)\right)\\
R^{(4)} &\le& \log \left({\mathsf {SNR}}^{1-\alpha}\right)-\min \{2 C_{FB},\log\left({\mathsf {SNR}}^{2-3\alpha}\right)\}-\log\left(\frac{4}{3}
\left(K+2\right)\right)\\
R^{(5)} &\le& \log \left( {{{\mathsf {SNR}}^{2\alpha-1}}} \right)-\log\left(3\left(K+{\frac{2}{3}}\right)\right)\\
R^{(6)} &\le& \log \left({{\mathsf {SNR}}^{1-\alpha}}\right)-\log(4(K+2)).
\end{eqnarray}
Putting these bounds all together, we achieve the rate $\frac{R^{(1)}+R^{(2)}+R^{(3)}+R^{(4)}+R^{(5)}+R^{(6)}}{2}=$
\begin{eqnarray}\label{ac2}
&& \frac{1}{2}\left(\log \left( {{\mathsf {SNR}}^{2\alpha-1}}\right)-\log\left(\frac{3}{2}\left(K+
\frac{2}{3}\right)\right)+\right.\nonumber\\
&&\min \{2 C_{FB},\log\left({\mathsf {SNR}}^{2-3\alpha}\right)\}-\log\left(12{(K-1)^2}
(K+{\frac{1}{3}})\right)+\nonumber\\
&&\min \{2 C_{FB},\log\left({\mathsf {SNR}}^{2-3\alpha}\right)\}-\log\left(4\left(K+
{\frac{11}{4}}\right)\right)+\nonumber\\
&&\log \left({\mathsf {SNR}}^{1-\alpha}\right)-\min \{2 C_{FB},\log\left({\mathsf {SNR}}^{2-3\alpha}\right)\}-\log\left(\frac{4}{3}
\left(K+2\right)\right)+\nonumber\\
&&\log \left( {{{\mathsf {SNR}}^{2\alpha-1}}} \right)-\log\left(3\left(K+{\frac{2}{3}}\right)\right)+\nonumber\\
&& \log \left({{\mathsf {SNR}}^{1-\alpha}}\right)-\log(4(K+2))\Bigg)\nonumber\\
&=& \log \left({{\mathsf {SNR}}^{2\alpha-1}}\right)
+\log \left({{\mathsf {SNR}}^{1-\alpha}}\right)
+\min \{C_{FB},\frac{1}{2}\log\left({\mathsf {SNR}}^{2-3\alpha}\right)\}\nonumber\\
&&-\frac{1}{2}\log\left(768\left(K-1\right)^2\left(K+
{\frac{1}{3}}\right)
\left(K+{\frac{2}{3}}\right)^2\left(K+2\right)^2\left(K+
{\frac{11}{4}}\right)\right).
\end{eqnarray}
Next we will bound the gap between \eqref{ac2} and the conjectured upper bound in \eqref{conj}. We split this into 2 regimes. The first is when  $2C_{FB}\le \log\left({\mathsf {SNR}}^{2-3\alpha}\right)$, and the second is when $C_{FB}> \log\left({\mathsf {SNR}}^{2-3\alpha}\right)$. In the first case, we find the distance between \eqref{ac2} and the bound ${R}^u_{sym,0}+C_{FB}$ as follows.
\begin{eqnarray}\label{medbound1}
&&{R}^u_{sym,0}+C_{FB}-\left(\log \left( {{\mathsf {SNR}}^{2\alpha-1}}\right)
+\log \left({{\mathsf {SNR}}^{1-\alpha}}\right)
+\min \{C_{FB},\frac{1}{2}\log\left({\mathsf {SNR}}^{2-3\alpha}\right)\}\right.\nonumber\\
&&\left.-\frac{1}{2}\log\left(768\left(K-1\right)^2\left(K+{\frac{1}{3}}\right)
\left(K+{\frac{2}{3}}\right)^2\left(K+2\right)^2\left(K+{\frac{11}{4}}\right)\right)\right)
\nonumber\\
&\le&\left(\log(1+{\mathsf{INR}}+\frac{{\mathsf{SNR}}}{1+\mathsf{INR}})
+C_{FB}\right)-\nonumber\\
&&\left(\log \left( {{\mathsf {SNR}}^{2\alpha-1}}\right)
+\log \left({{\mathsf {SNR}}^{1-\alpha}}\right)
+\min \{C_{FB},\frac{1}{2}\log\left({\mathsf {SNR}}^{2-3\alpha}\right)\}\right.\nonumber\\
&&\left.-\frac{1}{2}\log\left(768\left(K-1\right)^2\left(K+{\frac{1}{3}}\right)
\left(K+{\frac{2}{3}}\right)^2\left(K+2\right)^2\left(K+{\frac{11}{4}}\right)\right)\right)
\nonumber\\
&=& \left(\log(1+{\mathsf{INR}}+\frac{{\mathsf{SNR}}}{1+\mathsf{INR}})
+C_{FB}\right)-\left(\log \left( {{\mathsf {SNR}}^{2\alpha-1}}\right)
+\log \left({{\mathsf {SNR}}^{1-\alpha}}\right)
+ C_{FB}\right.\nonumber\\
&&\left.-\frac{1}{2}\log\left(768\left(K-1\right)^2\left(K+{\frac{1}{3}}\right)
\left(K+{\frac{2}{3}}\right)^2\left(K+2\right)^2\left(K+{\frac{11}{4}}\right)\right)\right)
\nonumber\\
&=& \log(1+{\mathsf{INR}}+\frac{{\mathsf{SNR}}}{1+\mathsf{INR}})-\log \left( {{\mathsf {SNR}}^{2\alpha-1}}\right)
-\log \left({{\mathsf {SNR}}^{1-\alpha}}\right)\nonumber\\
&&+\frac{1}{2}\log\left(768\left(K-1\right)^2\left(K+{\frac{1}{3}}\right)
\left(K+{\frac{2}{3}}\right)^2\left(K+2\right)^2\left(K+{\frac{11}{4}}\right)\right)
\nonumber\\
&=& \log \left( \frac{(1+{\mathsf{INR}}+\frac{{\mathsf{SNR}}}{1+\mathsf{INR}})}{({{\mathsf {SNR}}^{2\alpha-1}})({{\mathsf {SNR}}^{1-\alpha}})}\right)
\nonumber\\
&&+\frac{1}{2}\log\left(768\left(K-1\right)^2\left(K+{\frac{1}{3}}\right)
\left(K+{\frac{2}{3}}\right)^2\left(K+2\right)^2\left(K+{\frac{11}{4}}\right)\right)
\nonumber\\
&\le& \log \left( \frac{(3{\mathsf{INR}})}{({{\mathsf {SNR}}^{2\alpha-1}})({{\mathsf {SNR}}^{1-\alpha}})}\right)
\nonumber\\
&&+\frac{1}{2}\log\left(768\left(K-1\right)^2\left(K+{\frac{1}{3}}\right)
\left(K+{\frac{2}{3}}\right)^2\left(K+2\right)^2\left(K+{\frac{11}{4}}\right)\right)
\nonumber\\
&\le& \log 3+\frac{1}{2}\log\left(768\left(K-1\right)^2\left(K+{\frac{1}{3}}\right)
\left(K+{\frac{2}{3}}\right)^2\left(K+2\right)^2\left(K+{\frac{11}{4}}\right)\right)
\nonumber\\
&=&\frac{1}{2}\log\left(6912\left(K-1\right)^2\left(K+{\frac{1}{3}}\right)
\left(K+{\frac{2}{3}}\right)^2\left(K+2\right)^2
\left(K+{\frac{11}{4}}\right)\right).
\end{eqnarray}
In the second case when $C_{FB}> \log\left({\mathsf {SNR}}^{2-3\alpha}\right)$, we find the gap between \eqref{ac2} and ${R}^u_{sym,\infty}$ as follows.
%
\begin{eqnarray}\label{medbound2}
&&{R}^u_{sym,\infty}-\left(\log \left( {{\mathsf {SNR}}^{2\alpha-1}}\right)
+\log \left({{\mathsf {SNR}}^{1-\alpha}}\right)
+\min \{C_{FB},\frac{1}{2}\log\left({\mathsf {SNR}}^{2-3\alpha}\right)\}\right.\nonumber\\
&&\left.-\frac{1}{2}\log\left(768\left(K-1\right)^2\left(K+{\frac{1}{3}}\right)
\left(K+{\frac{2}{3}}\right)^2\left(K+2\right)^2\left(K+{\frac{11}{4}}\right)\right)\right)
\nonumber\\
&=&\left(\frac{1}{2}\log(1+\frac{{\mathsf{SNR}}}{1+\mathsf{INR}})+
\frac{1}{2}\log(1+{\mathsf{SNR}}+{\mathsf{INR}})+\frac{K-1}{2}+\log K\right)-\nonumber\\
&&\left(\log \left( {{\mathsf {SNR}}^{2\alpha-1}}\right)
+\log \left({{\mathsf {SNR}}^{1-\alpha}}\right)
+\min \{C_{FB},\frac{1}{2}\log\left({\mathsf {SNR}}^{2-3\alpha}\right)\}\right.\nonumber\\
&&\left.-\frac{1}{2}\log\left(768\left(K-1\right)^2\left(K+{\frac{1}{3}}\right)
\left(K+{\frac{2}{3}}\right)^2\left(K+2\right)^2\left(K+{\frac{11}{4}}\right)\right)\right)
\nonumber\\
&=& \frac{1}{2}\log(1+\frac{{\mathsf{SNR}}}{1+\mathsf{INR}})+
\frac{1}{2}\log(1+{\mathsf{SNR}}+{\mathsf{INR}})+\frac{K-1}{2}+\log K\nonumber\\
&&-\log \left( {{\mathsf {SNR}}^{2\alpha-1}}\right)
-\log \left({{\mathsf {SNR}}^{1-\alpha}}\right)
-\frac{1}{2}\log\left({\mathsf {SNR}}^{2-3\alpha}\right)\nonumber\\
&&+\frac{1}{2}\log\left(768\left(K-1\right)^2\left(K+{\frac{1}{3}}\right)
\left(K+{\frac{2}{3}}\right)^2\left(K+2\right)^2\left(K+{\frac{11}{4}}\right)\right)
\nonumber\\
&\le& \frac{1}{2}\log(1+{{\mathsf{SNR}}^{1-\alpha}})+
\frac{1}{2}\log(3{\mathsf{SNR}})+\frac{K-1}{2}+\log K\nonumber\\
&&-\log \left( {{\mathsf {SNR}}^{2\alpha-1}}\right)
-\log \left({{\mathsf {SNR}}^{1-\alpha}}\right)
-\frac{1}{2}\log\left({\mathsf {SNR}}^{2-3\alpha}\right)\nonumber\\
&&+\frac{1}{2}\log\left(768\left(K-1\right)^2\left(K+{\frac{1}{3}}\right)
\left(K+{\frac{2}{3}}\right)^2\left(K+2\right)^2\left(K+{\frac{11}{4}}\right)\right)
\nonumber\\
&\le& \frac{1}{2}\log(3{\mathsf{SNR}})-\log \left( {{\mathsf {SNR}}^{2\alpha-1}}\right)
-\frac{1}{2}\log \left({{\mathsf {SNR}}^{1-\alpha}}\right)
-\frac{1}{2}\log\left({\mathsf {SNR}}^{2-3\alpha}\right)\nonumber\\
&&+\frac{1}{2}\log\left(768\left(K-1\right)^2{K^2}\left(K+{\frac{1}{3}}\right)
\left(K+{\frac{2}{3}}\right)^2\left(K+2\right)^2\left(K+{\frac{11}{4}}\right)\right)+\frac{K-1}{2}
\nonumber\\
&=& \frac{1}{2}\log \left(\frac{(3{\mathsf{SNR}})}{({{\mathsf {SNR}}^{2\alpha-1}})^2({\mathsf {SNR}}^{2-3\alpha})({{\mathsf {SNR}}^{1-\alpha}})}\right)+\frac{K-1}{2}\nonumber\\
&&+\frac{1}{2}\log\left(768\left(K-1\right)^2{K^2}\left(K+{\frac{1}{3}}\right)
\left(K+{\frac{2}{3}}\right)^2\left(K+2\right)^2\left(K+{\frac{11}{4}}\right)\right)
\nonumber\\
&=& \frac{1}{2}\log 3 +\frac{1}{2}\log\left(768\left(K-1\right)^2{K^2}\left(K+{\frac{1}{3}}\right)
\left(K+{\frac{2}{3}}\right)^2\left(K+2\right)^2\left(K+{\frac{11}{4}}\right)\right)+\frac{K-1}{2}
\nonumber\\
&=& \frac{1}{2}\log\left(2304\left(K-1\right)^2{K^2}\left(K+{\frac{1}{3}}\right)
\left(K+{\frac{2}{3}}\right)^2\left(K+2\right)^2\left(K+
{\frac{11}{4}}\right)\right)+\frac{K-1}{2}.
\end{eqnarray}
From \eqref{medbound1} and \eqref{medbound2}, we find that the achievable symmetric rate is within $\frac{1}{2}\log\left(2304\left(K-1\right)^2{K^2}\left(K+{\frac{1}{3}}\right)\right.$
$\left.\left(K+{\frac{2}{3}}\right)^2\left(K+2\right)^2\left(K+{\frac{11}{4}}\right)\right)+\frac{K-1}{2}$ bits to the conjectured upper bound \eqref{conj} when $\frac{1}{2} \le \alpha \le \frac{2}{3}$.

{\bf Case 3 ($ \alpha \ge 2$)}: We use the following parameters in Theorem \ref{thm_gauss3}: $\mu^{(2)} = \frac{{\mathsf{SNR}}}{2\mathsf{INR}}\min \{2^{2 C_{FB}},\frac{{\mathsf {INR}}}{{\mathsf {SNR}}^{2}}\},$ and $\mu^{(1)}= \mu^{(3)} = 1-\mu^{(2)}$. We first lower bound the RHS of \eqref{1s}-\eqref{6s} as follows.

RHS of \eqref{1s}:
\begin{eqnarray}
&&\log \left( \frac{{\mathsf {INR}} \mu^{(1)}}{{\mathsf {SNR}} \mu^{(1:2)}+{\mathsf {SNR}}^{\alpha}\mu^{(2)}(K-1)+1} \right)\nonumber\\
&\stackrel{(a)}{\ge}& \log \left( \frac{\frac{1}{2}{\mathsf {INR}} }{{\mathsf {SNR}} +\frac{1}{2}{\mathsf {SNR}}\min \{2^{2 C_{FB}},\frac{{\mathsf {INR}}}{{\mathsf {SNR}}^{2}}\}(K-1)+1} \right)\nonumber\\
&\stackrel{(b)}{\ge}& \log \left( \frac{\frac{1}{2}{\mathsf {INR}} }{{\mathsf {SNR}}\min \{2^{2 C_{FB}},\frac{{\mathsf {INR}}}{{\mathsf {SNR}}^{2}}\}\frac{(K+3)}{2}} \right)\nonumber\\
&=& \log \left( \frac{{\mathsf {INR}} }{{\mathsf {SNR}}\min \{2^{2 C_{FB}},\frac{{\mathsf {INR}}}{{\mathsf {SNR}}^{2}}\}{(K+3)}} \right)\nonumber\\
&=& \log \left( \frac{{\mathsf {SNR}}^{\alpha-1} }{\min \{2^{2 C_{FB}},\frac{{\mathsf {INR}}}{{\mathsf {SNR}}^{2}}\}{(K+3)}} \right)\nonumber\\
&=& \log \left( {{\mathsf {SNR}}^{\alpha-1}\max \{2^{-2 C_{FB}},\frac{{\mathsf {SNR}}^{2}}{{\mathsf {INR}}}\}} \right)-\log(K+3)\nonumber\\
&{\ge}& \log \left(  {{\mathsf {SNR}}^{\alpha-1}\left(\frac{{\mathsf {SNR}}^{2}}{{\mathsf {INR}}}\right)} \right)-\log(K+3)\nonumber\\
&=& \log \left( {{\mathsf {SNR}}} \right)-\log(K+3),
\end{eqnarray}
where (a) follows since  $\mu^{(1)}\ge \frac{1}{2}$, and $\mu^{(1:2)}= 1$, (b) follows since  $\min \{2^{2 C_{FB}},\frac{{\mathsf {INR}}}{{\mathsf {SNR}}^{2}}\} \ge1$, and ${\mathsf {SNR}}\ge1$.

RHS of \eqref{2s}:
\begin{eqnarray}
&& \log \left(  \frac{{\mathsf {INR}} \mu^{(2)}}{{\mathsf {SNR}} \mu^{(1:2)}+1} \right)\nonumber\\
&=& \log \left( \frac{\frac{1}{2}{\mathsf {SNR}}\min \{2^{2 C_{FB}},\frac{{\mathsf {INR}}}{{\mathsf {SNR}}^{2}}\}
}{{\mathsf {SNR}}+1} \right)\nonumber\\
&\stackrel{(a)}{\ge}& \log \left( \frac{\frac{1}{2}{\mathsf {SNR}}\min \{2^{2 C_{FB}},\frac{{\mathsf {INR}}}{{\mathsf {SNR}}^{2}}\}
}{2{\mathsf {SNR}}} \right)\nonumber\\
&=&\log \left( \frac{{\mathsf {SNR}}\min \{2^{2 C_{FB}},\frac{{\mathsf {INR}}}{{\mathsf {SNR}}^{2}}\}
}{{\mathsf {SNR}}} \right)-\log(4)\nonumber\\
&=&\log \left( \min \{2^{2 C_{FB}},\frac{{\mathsf {INR}}}{{\mathsf {SNR}}^{2}}\}
\right)-\log(4)\nonumber\\
&=&\min \{2 C_{FB},\log\left({{\mathsf {SNR}}^{\alpha-2}} \right)\} -\log(4),
\end{eqnarray}
where (a) follows since  ${\mathsf {SNR}}\ge1$.

RHS of \eqref{3s}:
\begin{eqnarray}
&& \log \left( \frac{{\mathsf {SNR}} \mu^{(1)}}{{\mathsf {SNR}}\mu^{(2)}+1} \right)\nonumber\\
&\stackrel{(a)}{\ge}& \log \left( \frac{\frac{1}{2}{\mathsf {SNR}}}{\frac{1}{2}+1} \right)\nonumber\\
&=& \log \left( \frac{{\mathsf {SNR}}}{3} \right)\nonumber\\
&=& \log \left( {\mathsf {SNR}} \right)-\log(3),
\end{eqnarray}
where (a) follows since  $\mu^{(1)}\ge \frac{1}{2}$.

RHS of \eqref{4s}:
\begin{eqnarray}
&& \log \left(  \frac{{\mathsf {SNR}}^{\alpha} \mu^{(3)}}{{\mathsf {SNR}}^{\alpha} \mu^{(2)}+{\mathsf {SNR}}\mu^{(2:3)}+1} \right)\nonumber\\
&\stackrel{(a)}{\ge}& \log \left(  \frac{\frac{1}{2}{\mathsf {INR}} }{\frac{1}{2}{\mathsf {SNR}}\min \{2^{2 C_{FB}},\frac{{\mathsf {INR}}}{{\mathsf {SNR}}^{2}}\}+{\mathsf {SNR}} +1} \right)\nonumber\\
&\stackrel{(b)}{\ge}& \log \left(  \frac{\frac{1}{2}{\mathsf {INR}} }{{\mathsf {SNR}}\min \{2^{2 C_{FB}},\frac{{\mathsf {INR}}}{{\mathsf {SNR}}^{2}}\}\frac{(5)}{2}} \right)\nonumber\\
&=& \log \left(  \frac{{\mathsf {INR}} }{{\mathsf {SNR}}\min \{2^{2 C_{FB}},\frac{{\mathsf {INR}}}{{\mathsf {SNR}}^{2}}\}{(5)}} \right)\nonumber\\
&=& \log \left(  \frac{{\mathsf {SNR}}^{\alpha-1} }{\min \{2^{2 C_{FB}},\frac{{\mathsf {INR}}}{{\mathsf {SNR}}^{2}}\}{(5)}} \right)\nonumber\\
&=& \log \left(  {{\mathsf {SNR}}^{\alpha-1}\max \{2^{-2 C_{FB}},\frac{{\mathsf {SNR}}^{2}}{{\mathsf {INR}}}\}} \right)-\log(5)\nonumber\\
&{\ge}& \log \left(  {{\mathsf {SNR}}^{\alpha-1}\left(\frac{{\mathsf {SNR}}^{2}}{{\mathsf {INR}}}\right)} \right)-\log(5)\nonumber\\
&=& \log \left(  {{\mathsf {SNR}}} \right)-\log(5),
\end{eqnarray}
where (a) follows since  $\mu^{(3)}\ge \frac{1}{2}$, and $\mu^{(2:3)}= 1$, (b) follows since  $\min \{2^{2 C_{FB}},\frac{{\mathsf {INR}}}{{\mathsf {SNR}}^{2}}\} \ge1$, and ${\mathsf {SNR}}\ge1$.

RHS of \eqref{5s}:
\begin{eqnarray}
&& \log \left(  \frac{{\mathsf {SNR}}^{\alpha} \mu^{(2)}}{{\mathsf {SNR}}\mu^{(3)}+1} \right)\nonumber\\
&\stackrel{(a)}{\ge}& \log \left(  \frac{\frac{1}{2}{\mathsf {SNR}}\min \{2^{2 C_{FB}},\frac{{\mathsf {INR}}}{{\mathsf {SNR}}^{2}}\}}{{\mathsf {SNR}}+1} \right)\nonumber\\
&\stackrel{(b)}{\ge}& \log \left(  \frac{\frac{1}{2}{\mathsf {SNR}}\min \{2^{2 C_{FB}},\frac{{\mathsf {INR}}}{{\mathsf {SNR}}^{2}}\}}{2{\mathsf {SNR}}} \right)\nonumber\\
&=& \log \left(  \frac{{\mathsf {SNR}}\min \{2^{2 C_{FB}},\frac{{\mathsf {INR}}}{{\mathsf {SNR}}^{2}}\}}{{\mathsf {SNR}}} \right)-\log(4)\nonumber\\
&=& \log \left( \min \{2^{2 C_{FB}},\frac{{\mathsf {INR}}}{{\mathsf {SNR}}^{2}}\} \right)-\log(4)\nonumber\\
&=&\min \{2 C_{FB},\log\left({{\mathsf {SNR}}^{\alpha-2}} \right)\} -\log(4),
\end{eqnarray}
where (a) follows since  $\mu^{(3)} \le 1$, and $\mu^{(2:3)}= 1$, (b) follows since  ${\mathsf {SNR}}\ge1$.

RHS of \eqref{6s}:
\begin{eqnarray}
&& \log \left(  \frac{{\mathsf {SNR}} \mu^{(3)}}{1} \right)\nonumber\\
&\stackrel{(a)}{\ge}& \log \left(  \frac{{\mathsf {SNR}} \frac{1}{2}}{1} \right)\nonumber\\
&=& \log \left(  \mathsf {SNR} \right)-\log(2),
\end{eqnarray}
where (a) follows since  $\mu^{(3)}\ge \frac{1}{2}$.

Thus, by considering \eqref{7s}, we find the achievable rate expressions can be reduced as follows:
\begin{eqnarray}
R^{(1)} &\le& \log \left( {{\mathsf {SNR}}} \right)-\log(K+3)\\
R^{(2)} &\le&  \min \{2 C_{FB},\log\left({{\mathsf {SNR}}^{\alpha-2}} \right)\} -\log\left(\max\{4,K-1\}\right)\\
R^{(3)} &\le& \log \left(  {{\mathsf {SNR}}} \right)-\log(5).
\end{eqnarray}
Putting these bounds all together, we achieve $\frac{R^{(1)}+R^{(2)}+R^{(3)}}{2}=$
\begin{eqnarray}\label{ac3}
&&\frac{1}{2}\left(\log \left( {{\mathsf {SNR}}} \right)-\log(K+3)
+\min \{2 C_{FB},\log\left({{\mathsf {SNR}}^{\alpha-2}} \right)\} -\log\left(\max\{4,K-1\}\right)\right.\nonumber\\
&&\left.+\log \left(  {{\mathsf {SNR}}} \right)-\log(5)\right)\nonumber\\
&=&\log \left(  {{\mathsf {SNR}}} \right)
+\min \{C_{FB},\frac{1}{2}\log\left({{\mathsf {SNR}}^{\alpha-2}} \right)\}
-\frac{1}{2}\log\left(5 (\max\{4,K-1\}) (K+3)\right).
\end{eqnarray}
Define $K'\triangleq \frac{1}{2}\log\left(5 (\max\{4,K-1\}) (K+3)\right)$ for notational simplicity in the following. Next, we will bound the gap between \eqref{ac3} and the conjectured upper bound \eqref{conj}. We split this into 2 regimes. The first is when $2C_{FB}\le \log\left({\mathsf {SNR}}^{\alpha-2}\right)$, and the second is when $C_{FB}> \log\left({\mathsf {SNR}}^{\alpha-2}\right)$. In the first case, we bound the gap between \eqref{ac3} and ${R}^u_{sym,0}+C_{FB}$ as follows:
\begin{eqnarray}\label{strongbound1}
&&{R}^u_{sym,0}+C_{FB}-\left(\log \left( {{\mathsf {SNR}}} \right)
+\min \{C_{FB},\frac{1}{2}\log\left({{\mathsf {SNR}}^{\alpha-2}} \right)\}
-K'\right)\nonumber\\
&\le&\left(\log(1+{\mathsf{SNR}})+C_{FB}\right)-\nonumber\\
&&\left(\log \left( {{\mathsf {SNR}}} \right)
+\min \{C_{FB},\frac{1}{2}\log\left({{\mathsf {SNR}}^{\alpha-2}} \right)\}
- K' \right)\nonumber\\
&=& \left(\log(1+{\mathsf{SNR}})+C_{FB}\right)-
\left(\log \left( {{\mathsf {SNR}}} \right)+C_{FB}- K' \right)\nonumber\\
&\le&\log(2)+ K' \nonumber\\
&=&\frac{1}{2}\log\left(20 (\max\{4,K-1\}) (K+3)\right).
\end{eqnarray}
In the second case we bound the gap between \eqref{ac3} and ${R}^u_{sym,\infty}$ as follows:
\begin{eqnarray}\label{strongbound2}
&&{R}^u_{sym,\infty}-\left(\log \left({{\mathsf {SNR}}} \right)
+\min \{C_{FB},\frac{1}{2}\log\left({{\mathsf {SNR}}^{\alpha-2}} \right)\}
- K' \right)\nonumber\\
&=&\left(\frac{1}{2}\log(1+\frac{{\mathsf{SNR}}}{1+\mathsf{INR}})+
\frac{1}{2}\log(1+{\mathsf{SNR}}+{\mathsf{INR}})+\frac{K-1}{2}+\log K
\right)-\nonumber\\
&&\left(\log \left({{\mathsf {SNR}}} \right)
+\min \{C_{FB},\frac{1}{2}\log\left({{\mathsf {SNR}}^{\alpha-2}} \right)\}
- K' \right)\nonumber\\
&=&\frac{1}{2}\log(1+\frac{{\mathsf{SNR}}}{1+\mathsf{INR}})+
\frac{1}{2}\log(1+{\mathsf{SNR}}+{\mathsf{INR}})+\frac{K-1}{2}+\log K\nonumber\\
&&-\log \left( {{\mathsf {SNR}}} \right)
-\frac{1}{2}\log\left({{\mathsf {SNR}}^{\alpha-2}} \right)
+ K' \nonumber\\
&=&\frac{1}{2}\log\left(\frac{(1+{\mathsf{SNR}}+{\mathsf{INR}})
(1+{\mathsf{SNR}}+{\mathsf{INR}})}
{({\mathsf {SNR}}^{\alpha-2})(1+ {{\mathsf {INR}}})({{\mathsf {SNR}}})^2}\right)+\frac{K-1}{2}+\log K+ K' \nonumber\\
&\le&\frac{1}{2}\log\left(\frac{(3{\mathsf{SNR}}^{\alpha})^2}{({\mathsf {SNR}}^{\alpha-2})({{\mathsf {INR}}})({{\mathsf {SNR}}})^2}\right)+\frac{K-1}{2}+\log K
+ K' \nonumber\\
&=&\frac{1}{2}\log 9+\frac{K-1}{2}+\log K+\frac{1}{2}\log\left(5 (\max\{4,K-1\}) (K+3)\right)\nonumber\\
&=&\frac{K-1}{2}+\frac{1}{2}\log\left( 45 (\max\{4,K-1\}) K^2(K+3)\right).
\end{eqnarray}
Therefore, from  \eqref{strongbound1} and \eqref{strongbound2}, we find that the achievable symmetric rate is within $\frac{K-1}{2} + \frac{1}{2} \log\left(45 (\max\{4,K-1\}) K^2 (K+3)\right)$ bits to the conjectured upper bound in \eqref{conj} when $\alpha \ge 2$.

Combining these three cases together with the gap of $\frac{1}{2}\log9+16+\frac{K-1}{2}+3\log K$ bits  when $2/3<\alpha<1$ regime, and the gap of $\frac{1}{2}\log6+6+\frac{K-1}{2}+\log K$ bits when $1<\alpha<2$ (gap between the upper bound for the symmetric capacity with perfect feedback in Theorem 3 of \cite{Mohajer} and the lower bound for the symmetric capacity with no feedback in Theorem 1 of \cite{Ordentlich}), we find that the achievable symmetric rate is within $L$ bits to the conjectured upper bound in \eqref{conj}, where $L$ is given by \eqref{distance}.

\end{appendices}

\bibliographystyle{IEEETran}
\bibliography{bib}

\end{document}